\documentclass[sigconf]{acmart}
\renewcommand\footnotetextcopyrightpermission[1]{} 

\AtBeginDocument{%
  }
\usepackage{bm}
\usepackage{mathtools}
\usepackage{soul}
\usepackage{amsmath}
\usepackage{amsthm}
\usepackage{apxproof}
\usepackage{algorithmicx,algpseudocode, algorithm}
\usepackage{multirow}
\usepackage{subcaption}
\usepackage{makecell}  
\usepackage{booktabs}
\usepackage{siunitx}
\usepackage{tabularx}

\usepackage{eso-pic}

\usepackage{xcolor}

\DeclareMathOperator{\PLRV}{\mathrm{PLRV}}
\DeclareMathOperator{\Beta}{\mathsf{Beta}}
\DeclareMathOperator{\BP}{\mathsf{BetaPrime}}
\DeclareMathOperator{\R}{\mathbb{R}}

\newtheorem{definition}{Definition}
\newtheorem{remark}{Remark}
\newtheorem{theorem}{Theorem}
\newtheorem{lemma}{Lemma}
\newtheorem{corollary}{Corollary}
\newtheorem{proposition}{Proposition} 
\newtheorem{observation}{Observation} 

\newcommand{\cellvcenter}{\rule{0pt}{2.5ex}\rule[1ex]{0pt}{0pt}}

\setcopyright{acmlicensed}
\copyrightyear{2026}
\acmYear{2026}
\acmDOI{XXXXXXX.XXXXXXX}
\acmConference[SIGMOD '26]{the ACM International Conference on Management of Data }{May 31--June 5,
  2026}{Bengaluru, IN}
  \acmBooktitle{Proceedings of the ACM International Conference on Management of Data (SIGMOD '26), May 31--June 5, 2026, Bengaluru, India}

\acmISBN{978-1-4503-XXXX-X/2018/06}

\begin{document}
\AddToShipoutPicture*{
    \AtPageUpperLeft{
        \raisebox{-1cm}{
            \hspace{0.65in}
            \parbox{\textwidth}{
                \centering \textit{A preliminary version of this work is accepted by   ACM International Conference on Management of Data (SIGMOD '26). This is the full version.}
            }
        }
    }
}

\title{Privacy Loss of Noise Perturbation via Concentration Analysis of A Product Measure}
\author{Shuainan Liu}
\orcid{0009-0000-9477-2922}
\email{Shuainan.Liu@ttu.edu}
\affiliation{%
  \institution{Texas Tech University}
  \city{}
  \state{}
  \country{}}
  
\author{Tianxi Ji }
 \authornote{Corresponding author.}
\email{tiji@ttu.edu}
\affiliation{%
  \institution{Texas Tech University}
  \city{}
  \state{}
  \country{}}

\author{Zhongshuo Fang}
\email{zhonfang@ttu.edu}
\affiliation{%
  \institution{Texas Tech University}
  \city{}
  \state{}
  \country{}}

\author{Lu Wei}
\email{luwei@ttu.edu}
\affiliation{%
  \institution{Texas Tech University}
  \city{}
  \state{}
  \country{}}

\author{Pan Li}
\email{lipan@ieee.org}
\affiliation{%
  \institution{Hangzhou Dianzi University}
  \city{}
  \state{}
  \country{}}
\renewcommand{\shortauthors}{Shuainan Liu, Tianxi Ji, Zhongshuo Fang, Lu Wei, and Pan Li}

\begin{abstract}
Noise perturbation is one of the most fundamental approaches for achieving $(\epsilon,\delta)$-differential privacy (DP) guarantees when releasing the result of a query or function $f(\cdot)\in\R^M$ evaluated on a sensitive dataset $\bm{x}$. In this approach, calibrated noise $\mathbf{n}\in\R^M$ is used to obscure the difference vector $f(\bm{x})-f(\bm{x}')$, where $\bm{x}'$ is known as a neighboring dataset. A DP guarantee is obtained by studying the tail probability bound of a privacy loss random variable (PLRV), defined as the Radon-Nikodym derivative between two distributions. When $\mathbf{n}$ follows a multivariate Gaussian distribution, the PLRV is characterized as a specific univariate Gaussian. In this paper, we propose a novel scheme to generate $\mathbf{n}$ by leveraging the fact that the perturbation noise is typically spherically symmetric (i.e., the distribution is rotationally invariant around the origin). The new noise generation scheme allows us to investigate the privacy loss from a geometric perspective and express the resulting PLRV using a product measure, $W\times U$; measure $W$ is related to a radius random variable controlling the magnitude of $\mathbf{n}$, while measure $U$ involves a directional random variable governing the angle between $\mathbf{n}$ and the difference $f(\bm{x})-f(\bm{x}')$. We derive a closed-form moment bound on the product measure to prove $(\epsilon,\delta)$-DP. Under the same $(\epsilon,\delta)$-DP guarantee, our mechanism yields a   smaller expected noise magnitude than the classic Gaussian noise in high dimensions, thereby significantly improving the utility of the noisy result $f(\bm{x})+\mathbf{n}$. To validate this, we consider privacy-preserving convex and non-convex empirical risk minimization (ERM) problems in high dimensional space. We propose leveraging our developed noise in output perturbation, objective perturbation, and gradient perturbation to establish DP guarantees when solving ERMs.  Experiments on multiple datasets show that our method achieves significant utility improvements for convex ERM models (e.g., regression and SVM) under the same privacy guarantees. For non-convex models (e.g., neural networks), it provides substantially stronger  privacy guarantees under comparable utility.\footnote{Code  Repository: \url{https://github.com/issleepgroup/Product-Noise-ERM}}
\end{abstract}




\maketitle
\thispagestyle{plain} 
\section{Introduction} \label{sec:intro}

Differential privacy (DP)~\cite{dwork2014algorithmic}  is widely applied to protect the privacy of a sensitive dataset $\bm{x}$ when releasing computation results, denoted as $f(\bm{x})\in\R^M$. An important privacy guarantee is $(\epsilon, \delta)$-DP. It means that, with the exception of a small probability $\delta$ (e.g., $<10^{-5}$), the inclusion or exclusion of any single data record in $\bm{x}$ cannot change the probability of observing a specific computation result by more than a multiplicative factor of $e^{\epsilon}$. $(\epsilon,\delta)$-DP enables a unified and comparable privacy notion for privacy-preserving algorithms, and has been widely applied to quantify privacy loss in various privacy-sensitive tasks, such as covariance matrix estimation in principal component analysis (PCA)~\cite{chaudhuri2013near,liu2022dp}, recommendation systems~\cite{McSherry2009differentially,wadhwa2020data}, data mining~\cite{friedman2010data,li2012PrivBasis}, graph data publication~\cite{kasiviswanathan2013analyzing,ji2021differentially}, machine learning and artificial intelligence~\cite{chaudhuri2011differentially,iyengar2019towards,abadi2016deep}.

The classic Gaussian mechanism~\cite{dwork2006our} is the most common way to achieve $(\epsilon,\delta)$-DP. It perturbs the computation result $f(\bm{x})\in\R^M$ using Gaussian noise $\mathbf{n}\in\R^M$ and releases the noisy result $\mathbf{s}$, i.e.,
\begin{equation}\label{eq:add_noise}
    \begin{aligned}
        \mathbf{s} = f(\bm{x}) + \mathbf{n}.
    \end{aligned}
\end{equation}
There are quite a few variants of this classic mechanism, e.g.,~\cite{balle2018improving, chanyaswad2018mvg}. In general, they all establish $(\epsilon,\delta)$-DP guarantee by investigating the tail probability bound of the \textbf{privacy loss random variable} (PLRV), defined as the Radon-Nikodym derivatives (cf. Definition~\ref{def:plrv} in Section~\ref{sec:preliminaries}). In essence, the PLRV quantifies the ratio between the probability density functions (PDFs) of two randomized outputs generated from slightly different inputs.

Looking from the lens of PLRV, all existing DP mechanisms using Gaussian noise share common characteristic:    when Gaussian noise with variance $\sigma^2$ is used to perturb $f(\bm{x})$ with $l_2$ sensitivity $\Delta_2f$,\footnote{$\Delta_2f = \sup_{\bm{x}\sim\bm{x}'}||f(\bm{x})-f(\bm{x}')||_2$, where $\bm{x}'$ differs with $\bm{x}$ by only one data record.}   the resulted  PLRV is characterized as another Gaussian distribution,  $\mathcal{N}\left(\frac{(\Delta_2f)^2}{2\sigma^2},2\frac{(\Delta_2f)^2}{2\sigma^2}\right)$ (cf.~\cite[Proposition 3]{steinke2022composition} and~\cite{balle2018improving}). Then, a certain $(\epsilon,\delta)$-DP guarantee is constructed by analyzing the tail probability of such Gaussian distributed PLRV. Other DP variants that leverage Gaussian noise, such as Concentrated DP~\cite{dwork2016concentrated,bun2016concentrated} or R\'enyi DP~\cite{mironov2017renyi}  handle the PLRV in much the same way; the considered PLRV is still bounded using the measure concentration of a Gaussian random variable. For example, to compose the privacy loss in DP-stochastic gradient descent, Abadi et al.~\cite{abadi2016deep} use the Cram\'er-Chernoff method to bound the moment generation function of a Gaussian-type PLRV.

From the perspective of the noise itself, the multivariate Gaussian noise in (\ref{eq:add_noise}) is known to be spherically symmetric~\cite{paolella2018linear,fang1990generalized}, which means that the noise can be represented as a random unit vector $\bm{h}\in\R^M$ scaled by a non-negative random scalar $R$ and a privacy related constant $\sigma$, i.e., $\mathbf{n} = \sigma R\bm{h}$ (the details are deferred to Section~\ref{sec:a-new-noise}). Since $\|\bm{h}\|_2=1$, we have $\|\mathbf{n}\|_2 \propto R$, thus, a straightforward approach to control the perturbation noise 
in (\ref{eq:add_noise})  is to make the non-negative random scalar $R$ as small as possible. Additionally, to control the event where a query result $f(\bm{x})$ is overwhelmed by large noise, we wish extremely large values of $\|\mathbf{n}\|_2$ (or outliers) are less likely to happen. Statistically speaking, we prefer $\|\mathbf{n}\|_2$ (or equivalently $R$) to have   large kurtosis and skewness~\cite{casella2002statistical,westfall2014kurtosis}. 

Unfortunately, as we will rigorously show in Section~\ref{sec:a-new-noise}, the commonly adopted multivariate Gaussian noise leads to two intrinsic limitations:  (i) a large     $R$, distributed as $\chi_M$ with degrees of freedom equal to the output dimension of $f(\bm{x})\in\R^M$, and
(ii) minimal kurtosis and skewness that scale inversely with the dimension, i.e., $\bm{O}(\frac{1}{M})$.
To overcome these drawbacks, we introduce a new product noise construction, where the non-negative scalar $R$ is instead drawn from $\chi_1$.

As one might expect, there exists a trade-off in the new noise construction 
as the degree of freedom of $R$ is reduced from $M$ to 1. Specifically, this trade-off manifests in the privacy-related constant, which becomes dependent on the dimension $M$, and we denote it as $\sigma_M$ in our noise (a counterpart of $\sigma$ in the Gaussian mechanism). Nevertheless, we will develop advanced techniques to control this trade-off, ensuring the new constant $\sigma_M$ grows only sublinearly with respect to $M$. Our novel techniques have two key components:

$\bullet$ We construct a \textbf{non-Gaussian type} PLRV by interpreting noise perturbation in (\ref{eq:add_noise}) as a random triangle instance in $\R^M$. Under this geometric view, the new form of PLRV is a product of two measures; one (denoted as $W$) is related to a random variable controlling the magnitude of our noise, and the other (indicated by $U$) involves a directional random variable governing the angle between our noise and the sensitive vector that needs to be obfuscated.

$\bullet$ To control the value of the constructed PLRV (i.e., $W\times U$), we adopt the moment bound~\cite{philips1995moment} to directly constrain its tail probability, which can be represented in the form of $\Pr\left[WU\geq \frac{\sigma_M}{\Delta_2f}\epsilon\right]\leq \delta$. Moment bound is always tighter than the Cram\'er–Chernoff bound~\cite{philips1995moment}, which is widely used by existing mechanisms~\cite{bun2016concentrated,mironov2017renyi,abadi2016deep} to analyze the tail probability of a PLRV through bounding its moment generating function. As a result, for a given dimension $M$, we can make the required $\sigma_M$ scale sublinearly with $M$. 

Since the aforementioned trade-off is effectively managed, under the same $(\epsilon,\delta)$-DP guarantee, the expected magnitude of the required noise is significantly reduced compared with the classic Gaussian noise in high-dimensional applications. In fact, our empirical guidance suggests that when fixing $\delta=10^{-5}$, it becomes preferable to adopt our proposed product noise rather than the classic Gaussian noise whenever $M \geq 14$. 

\subsection{Contributions}
Our main contributions are summarized as follows:
 
\textbf{(1) New noise and new interpretation.}
 We propose a novel product noise that facilitates the interpretation of  DP through noise perturbation from the perspective of random geometry and offers an innovative characterization of PLRV as a product measure.  

\textbf{(2) Theory.}
We first rigorously derive the moment bound of the product measure to establish $(\epsilon,\delta)$-DP guarantee for our product noise. Building on this theory, we provide an approach to set the noise magnitude and conduct a systematic comparison with Gaussian noise in both low and high-dimensional settings, yielding actionable guidance for noise selection. 
Under the same DP guarantee, our noise has significantly lower magnitude compared with the classic Gaussian noise. We also theoretically analyze  the privacy and utility guarantees when the product noise is utilized  to learn machine learning models (formulated as  empirical risk minimization, ERM) in a differentially private manner.

\textbf{(3) Application.}
We apply the proposed product noise to differentially private ERM. We consider both convex ERM (logistic regression and support vector machines) and non-convex ERM (deep neural networks). For the convex scenario, we solve the ERM via both output perturbation~\cite{chaudhuri2011differentially,wu2017bolt} and objective perturbation~\cite{iyengar2019towards}. experiment results show that, under the same or even stricter privacy guarantees, perturbation using our proposed noise can achieve higher utility (i.e., testing accuracy) compared with other noise.  We use gradient perturbation~\cite {abadi2016deep,bu2020deep}  to solve the non-convex ERMs. Experiment results demonstrate that, to achieve comparable utility, our proposed noise requires smaller privacy parameters ($\epsilon$ and $\delta$) compared with classic Gaussian noise, thereby providing a stronger privacy guarantee.
 
\noindent\textbf{Roadmap.}
Section~\ref{sec:relatedwork} reviews the related work.
Section~\ref{sec:preliminaries} presents the preliminaries.
Section~\ref{sec:main_results} provides the main results, guidance and simulation of this work.
Section~\ref{sec:case-studies} theoretically explores the feasibility of applying the product noise to differentially private ERMs.
Section~\ref{sec:experiments} empirically evaluates the privacy and utility trade-off of the product noise on various ERM tasks and datasets.
Finally, Section~\ref{sec:Conclusion} concludes the paper.
\section{Related Work} \label{sec:relatedwork}
In this section, we  review some representative works
in various domains related to this work.

\noindent\ul{\textbf{Noise design.}} 
Many works have attempted to develop variants of Gaussian noise to achieve $(\epsilon, \delta)$-DP. For example, the Analytic Gaussian Mechanism~\cite{balle2018improving}     determines noise variance using the exact Gaussian cumulative distribution function (CDF). 
The MVG mechanism~\cite{chanyaswad2018mvg}  addresses matrix-valued queries with directional noise by using prior   knowledge on the data to determine noise distribution parameters. 
Our work generates a novel noise by leveraging the polar decomposition of spherically symmetric distributions to decompose multivariate noise into the product of two independent random components, i.e.,  a magnitude (radius) of the noise and its direction.

\noindent\ul{\textbf{Geometry insights.}}
Our noise generation scheme allows a geometric interpretation of the PLRV in noise perturbation. In particular, we show that PLRV can be upper bounded by the random diameter of a circumcircle of a random triangle formed by two noise instances $\mathbf{n}$ and $\mathbf{n}'$ and the vector $f(\bm{x})-f(\bm{x}')$. This further implies a product measure~\cite{durrett2019probability} and requires advanced measure concentration analysis to study the privacy guarantee.

The random triangle interpretation was first observed by Ji et al.~\cite{ji-r1smg}, yet they equire a strong independence and nearly orthogonality assumption between $\mathbf{n}$ and $\mathbf{n}'$, which are used to perturb $f(\bm{x})$ and $f(\bm{x}')$, respectively. They proposed a DP mechanism that is only valid for extremely restricted privacy regimes (e.g., $\epsilon < \frac{1}{M}$). In contrast, we analyze a random triangle by lifting this less practical assumption of $\mathbf{n}$ and $\mathbf{n}'$ being independent and nearly orthogonal. This allows our proposed noise perturbation mechanism to work for all $\epsilon>0$. 

Some other works also use geometric properties, yet their considered geometry is completely different from ours. In particular, they focus on the structure and error lower bounds of DP database queries. For example, Hardt and Talwar~\cite{hardt2010geometry} established lower bounds for linear queries by estimating the volume of the $\ell_1$ unit ball under query mappings and proposed optimal mechanisms based on $K$-norm sampling. Nikolov et al.~\cite{nikolov2013geometry} extended this approach to $(\epsilon, \delta)$-DP by introducing minimum enclosing ellipsoids and singular value-based projections. 
These works share a common design principle; leveraging the geometric structure of the \textbf{query space} to guide noise mechanism design. In contrast, we focus on the random geometry in   \textbf{noise space}. Reimherr et al.~\cite{reimherr2019elliptical} introduce an elliptical perturbation mechanism by generalizing the Gaussian distribution to have a non-spherical covariance, enabling anisotropic noise tailored to directions of varying sensitivity.  \cite{weggenmann2021differential} considers the randomness of angular statistics and proposes to use the von Mises–Fisher distribution and  Purkayastha distribution to protect the privacy of directional data.

\noindent \ul{\textbf{Ways to bound PLRV.}} 
In the literature, Concentrated DP (CDP)~\cite{dwork2016concentrated} and zero-Concentrated DP (zCDP)~\cite{bun2016concentrated} are established based on  PLRV. They obtain the tail bound of PLRV by using its moment generating function (MGF). In contrast,  we derive the tail bound of the PLRV using its $q$-th moments. An important conclusion in measure concentration theory is that the moment bound (which we adopted) for tail probabilities is always better than the ones derived using MGFs (see detailed discussions in Section~\ref{sec:moments-W-U}). Consequently, by leveraging the sharper moment bounds, it is plausible that CDP/zCDP analyses could be further improved in the Gaussian noise setting. However, this direction of research lies beyond the scope of the present work, which focuses on designing a new noise. Additionally, our PLRV is related to a $\BP$ distribution which does not have an MGF.  Hence, it is not feasible to characterize our mechanism using CDP/zCDP, which requires the underlying PLRV to be sub-Gaussian and to  admit MGF. Notably, a recent work ~\cite{yang2025plrvo} proposes to directly optimize the PLRV in differentially private stochastic gradient descent (DPSGD)   by considering  randomized scale parameters governed by a Gamma distribution. This method enables  task-specific optimization of DP-based  training by adapting noise and clip parameters to the learning setup  in DPSGD.

\section{Preliminaries}\label{sec:preliminaries}

In this section, we review some preliminaries in DP, statistics, and special functions. See more details  in  Appendix~\ref{app:additional-app} and~\ref{app:dp}.

\begin{definition}
\label{def_dp}
A randomized mechanism $\mathcal{M}$ satisfies $(\epsilon,\delta)$-DP if for any two neighboring datasets, $\bm{x},\bm{x}'\in\mathbb{N}^{|\mathcal{X}|}$ that differ by only one data record, 
$\epsilon>0$ and $0< \delta <1$, it satisfies 
\begin{equation*}
    \begin{aligned}
        \Pr[\mathcal{M}(\bm{x})\in \mathcal{S}] \leq e^{\epsilon}\Pr[\mathcal{M}(\bm{x}')\in \mathcal{S}]+\delta,
    \end{aligned}
\end{equation*}
where $\mathcal{S}$ denotes the output set. 
\end{definition}

\begin{definition}[PLRV~\cite{bun2016concentrated}]\label{def:plrv}
Let  $\bm{x}, \bm{x}' \in \mathbb{N}^{|\mathcal{X}|}$ be two neighboring datasets, $\mathcal{M}: \mathbb{N}^{|\mathcal{X}|} \to \mathcal{S}$ be a randomized mechanism, and   $P$ and $Q$ be the   distributions of $\mathcal{M}(\bm{x})$ and $\mathcal{M}(\bm{x}')$, respectively.   PLRV associated with   $\mathcal{M}$ and a randomized output $\mathbf{s}$ is defined as $\PLRV_{\mathcal{M}}(s) = \ln \left( \frac{\mathrm{d} P(\mathcal{M}(\bm{x})=\mathbf{s})}{\mathrm{d} Q(\mathcal{M}(\bm{x}')=\mathbf{s})} \right)$, where $\frac{\mathrm{d} P(\mathcal{M}(\bm{x})=\mathbf{s})}{\mathrm{d} Q(\mathcal{M}(\bm{x}')=\mathbf{s})}$ denotes the Radon-Nikodym derivative of $P$ with respect to $Q$ evaluated at $s \in \mathcal{S}$.
\end{definition}

\begin{remark}\label{remark:plrv}
Note that the traditional notation of  $\frac{\mathrm{d} P(\mathcal{M}(\bm{x})=\mathbf{s})}{\mathrm{d} Q(\mathcal{M}(\bm{x}')=\mathbf{s})}$ is due to the   analogies with differentiation~\cite{billingsley2017probability}. The derivative is a measurable function which only exists when $P$ is absolutely continuous with respect to $Q$, denoted as $P\ll Q$. When both $P$ and $Q$ are continuous and share the same reference measure,  $\frac{\mathrm{d} P(\mathcal{M}(\bm{x})=\mathbf{s})}{\mathrm{d} Q(\mathcal{M}(\bm{x}')=\mathbf{s})}$  reduces to the point-wise ratio of their probability density functions (PDF).
\end{remark}

\begin{definition} [Beta distribution and Beta Prime distribution~\cite{johnson1995continuous}] A random variable $X$ has a Beta distribution, i.e., $X\sim\Beta(\alpha,\beta)$, if it has the following probability density function
\begin{equation*}
    f_{X}(x)=\frac{1}{B(\alpha,\beta)}x^{\alpha-1}(1-x)^{\beta-1},  
\end{equation*}
where $\alpha>0, \beta>0, 0 \leq x\leq 1$, and $B(\alpha, \beta) = \frac{\Gamma (\alpha) \Gamma ( \beta) } { \Gamma ( \alpha+\beta) }$ is the Beta function. 

A random variable $X$ has a Beta Prime distribution, i.e., $X\sim\BP(\alpha,\beta)$, if it has the following probability density function
\begin{equation*}
     f_{X}(x)=\frac{1}{B(\alpha,\beta)}x^{\alpha-1}(1+x)^{-\alpha-\beta}, 
\end{equation*}
where $\alpha>0, \beta>0, x>0$.
\label{app:def-beta-and-beta-prime} 
\end{definition} 

\begin{lemma}\label{lemma:beta-betaprime}
    Beta distributed random variable can be converted to Beta Prime distributed random variable~\cite{siegrist2017probability}.
\begin{enumerate}
    \item If $X \sim \Beta(\alpha, \beta)$, then  $\frac{X}{1-X} \sim \BP(\alpha, \beta)$.
    \item If $X \sim \BP(\alpha, \beta)$, then $\frac{1}{X} \sim \BP(\beta, \alpha)$.
\end{enumerate}
\end{lemma}

\begin{definition} (Chi distribution~\cite[p. 73]{forbes2011statistical}) The probability density function (PDF) of the Chi distribution with  $\nu$  degrees of freedom is
\begin{equation*}\label{eq:chi}
   f(x;\nu)=\begin{cases}
 & \frac{x^{\nu-1}e^{-\frac{x^2}{2}}}{2^{\frac{\nu}{2}-1}\Gamma(\frac{\nu}{2})} , x \ge 0 \\
& 0, \text{otherwise}
\end{cases}.
\end{equation*}
\label{def:chi} 
\end{definition}

\begin{definition}\label{def:hypergeometric-function} The confluent hypergeometric function of the first kind~\cite[p. 322]{daalhuis2010confluent} is defined as $${}_1F_1(a; b; z) = \sum_{k=0}^{\infty } \frac{(a)_k}{(b)_k} \frac{z^k}{k!},$$  
where $(a)_k$ and $(b)_k$ is  the rising factorial, and $(a)_{0}=1$, $(a)_k=a (a+1)(a+2)\cdots(a+k-1)$.
\end{definition}

\begin{definition} \label{de: parabolic-cylinder-function} 
The parabolic cylinder function~\cite[p. 1028]{edition2007table} is 
\begin{equation*}
\resizebox{0.475\textwidth}{!}{$
\begin{aligned}
    D_p(z)  =& \frac{e^{-\frac{z^2}{4}}}{\Gamma (-p)} \int_{0}^{\infty } e^{-xz-\frac{x^2}{2} }x^{-p-1}dx\\
    =& 2^{\frac{p}{2}}e^{-\frac{z^2}{4}} \left \{ \frac{\sqrt{\pi} }{\Gamma (\frac{1-p}{2} )} 
    {}_1F_1\left(-\frac{p}{2} ; \frac{1}{2} ; \frac{z^2}{2} \right)  -\frac{\sqrt{2\pi }z }{\Gamma (-\frac{p}{2} )}
    {}_1F_1\left(\frac{1-p}{2} ; \frac{3}{2} ; \frac{z^2}{2} \right)  \right \},
\end{aligned}
$}
\end{equation*}
where $p<0$ is a real number, and ${}_1F_1(a; b; x)$ is the confluent hypergeometric function of the first kind defined in Definition~\ref{def:hypergeometric-function}.
\end{definition}

\section {Main Results, Guidance and Simulation} \label{sec:main_results}
This section elaborates on the core technical components of this paper. To facilitate readability, we unfold the discussion as follows.

$\bullet$ In Section~\ref{sec:a-new-noise}, we present the design principle of our new product noise based on the polar decomposition of spherical noise. 

$\bullet$ In Section~\ref{sec:new_plrv}, we derive the PLRV  associated with our product noise by leveraging a novel geometric insight on random triangles. 

$\bullet$ In Section~\ref{sec:moments-W-U}, We first construct a product measure based on the derived PLRV and then obtain a tight tail probability by analyzing the moment bound of such measure. 

$\bullet$ In Section~\ref{sec:privacy-guarantee}, we establish the privacy guarantee of our product noise by minimizing the obtained tail probability. 

$\bullet$ In Section~\ref{sec:guidance-and-simulation}, we demonstrate the advantages of product noise in high-dimensional settings, provide empirical guidelines for low dimensions, and validate all theoretical claims through simulations.

\subsection{A New Noise } \label{sec:a-new-noise}
The design principle of our new noise is inspired by the following observation. We notice that when $f(\bm{x})\in\R^M$, the multivariate noise considered in (\ref{eq:add_noise}) 
is spherically symmetric. This type of noise can be equivalently generated via the product between a positive random variable $R$ (usually referred to as the radius random variable) and a directional random variable $\boldsymbol{h}\in\mathbb{S}^{M-1}$ (the unit sphere embedded in $\R^M$). This is known as the polar decomposition~\cite{paolella2018linear}.

\begin{lemma} [Decomposition of Spherically Symmetric Random Variable~\cite{paolella2018linear}]\label{thm:ss_decomp}   
For a spherically symmetric distributed random variable  $\mathbf{n}\in\R^M$, one can express it as $\mathbf{n}\stackrel{d}{=}R\boldsymbol{h}$, where $\stackrel{d}{=}$ means equality in distribution, $R$ is a continuous univariate random variable (the radius) $\Pr[R\geq0]=1$, and $\boldsymbol{h}$ (independent of $R$) is uniformly distributed on the unit sphere surface $\mathbb{S}^{M-1}$ embedded in $\R^M$. 
\end{lemma}

Lemma~\ref{thm:ss_decomp} offers new insights into understanding the PLRV, highlights a drawback of the classic Gaussian noise, and motivates the design of a novel product noise that can be used to perturb $f(\bm{x})$ under $(\epsilon,\delta)$-DP guarantee. To be more specific, the classic Gaussian mechanism uses multivariate noise $\mathbf{n}\sim\mathcal{N}(\mathbf{0},\sigma^2\mathbf{I}_{M\times M})$ to perturb $f(\bm{x})\in\R^M$. 
This noise can be decomposed as 
\begin{equation*}
    \begin{aligned}
        \mathbf{n}  \stackrel{d}{=} \sigma \mathcal{N}(\mathbf{0}, \mathbf{I}_{M\times M}) \stackrel{d}{=} \sigma R \boldsymbol{h},
    \end{aligned}
\end{equation*} 
where $\sigma \ge \frac{\Delta_2f}{\epsilon} \sqrt{2\log (\frac{1.25}{\delta})}$, $R  \sim \chi_M$ (Chi distribution with $M$ degrees of freedom) and $\boldsymbol{h} \sim \mathbb{S}^{M-1}$~[p. 748, Example C.7]~\cite{paolella2018linear}.  Clearly, the squared magnitude of the classic Gaussian noise $||\mathbf{n}||_2^{2} = \sigma^2 ||R \boldsymbol{h}||_2^2 = \sigma^2 R^2$, is a scaled Chi-squared random variable with $M$ degrees of freedom (here $R^2\sim \chi^2_M$ for   multivariate  Gaussian). This creates a potential bottleneck in controlling the utility loss of the query result, particularly in high-dimensional settings. \ul{We argue that   $R$ with a small degree of freedom is promising in improving the utility and privacy trade-off.}

To show this, let a spherically symmetric noise takes the form of $$\mathbf{z} = \sigma_M R \bm{h}\in\R^M,$$ where    $\sigma_M$ is some constant,  $R\sim\chi_{\nu}$ (Chi random variable with $\nu$ degrees of freedom, $\nu \in [1, M]$) and $\bm{h}~\sim\mathbb{S}^{M-1}$. Clearly,  for a given degree of freedom $\nu$,  $\mathbb{E}[||\mathbf{z}||_2^2] = \sigma_M^2 \nu \propto \nu$, thus setting $\nu=1$ is a promising choice to reduce the noise magnitude/improve utility. Additionally, the kurtosis and skewness of $||\mathbf{z}||_2^{2}$   are $3+ \frac{12}{\nu}$ and $\sqrt{\frac{8}{\nu}}$, respectively~\cite{casella2002statistical}, both of which decrease  as $\nu$ increases. 
A smaller kurtosis indicates that outliers or extremely large values are more likely to be generated, and a smaller skewness indicates that most samples are located in the right region of the PDF. Hence,  to ensure that it is less likely to have large noise and let the mass of $||\mathbf{z}||_2^2$ concentrate in the left region of PDF, we \ul{require large values of kurtosis and skewness}, which are also maximized when $\nu=1$~\cite{paolella2018linear}. Thus, in this paper, we propose to explore the  following noise generation scheme
\begin{equation}\label{eq:noise-generation}
\boxed{\text{new  product noise:} \quad \mathbf{n} = \sigma_M R \boldsymbol{h}, R\sim\chi_1\ \text{and\ }\boldsymbol{h} \sim \mathbb{S}^{M-1},}
\end{equation}
where $\sigma_M$ is a to-be-determined constant decided by the dimension $M$ and privacy parameters $\epsilon$ and $\delta$, $\chi_1$ denotes the Chi distribution with 1 degree of freedom, and $\boldsymbol{h} \sim \mathbb{S}^{M-1}$ means uniform distribution on sphere $\mathbb{S}^{M-1}$ (embedded in $\R^M$).

\subsection{A New PLRV} \label{sec:new_plrv}
Clearly, the new noise in (\ref{eq:noise-generation})  involves the product of two different measures; one involves $\chi_1$ (which is the random radius), and the other one involves the direction of the noise. To analyze the privacy loss caused by the additive product noise in (\ref{eq:noise-generation}), we resort to the Radon-Nikodym derivative of the product measure and recall an important Lemma below.

\begin{lemma}[Product Measure~\cite{durrett2019probability}]
\label{lemma:prod-measure}
Let $\left(X, \mathcal{A} \right)$ \(\) and $\left ( (Y, \mathcal{B}) \right)$ be measurable spaces, let $\mu_1$ and $\nu_1$ be finite measures on $\left ( X, \mathcal{A} \right)$, and $\mu_2$ and $\nu_2$ be finite measures on $\left (Y, \mathcal{B}\right)$. If $\left(\nu_1 \ll \mu_1 \right)$ and $\left(\nu_2 \ll \mu_2 \right)$, then $\left(\nu_1 \times \nu_2 \ll \mu_1 \times \mu_2 \right)$, and $\frac{\mathrm{d}(\nu_1 \times \nu_2)}{\mathrm{d}(\mu_1 \times \mu_2)}(x, y) = \frac{\mathrm{d}\nu_1}{\mathrm{d}\mu_1}(x) \frac{\mathrm{d}\nu_2}{\mathrm{d}\mu_2}(y)$, for $\mu_1 \times \mu_2$\text{-almost everywhere}, where $\ll$ stands for absolutely continuous,   $\frac{\mathrm{d}\nu_1}{\mathrm{d}\mu_1}(x)$ denotes the Radon-Nikodym derivative of $\nu_1$ with respect to $\mu_1$ evaluated at $x$,   similarly for $\frac{\mathrm{d}\nu_2}{\mathrm{d}\mu_2}(y)$.
\end{lemma}

As a result, the PLRV associated with the noise defined in (\ref{eq:noise-generation}) can be alternatively characterized as in Proposition~\ref{prop: PLRV}.   This characterization differs from the existing PLRVs discussed in Section~\ref{sec:intro} and  is \textbf{not} a Gaussian-type PLRV anymore.

\begin{proposition}\label{prop: PLRV}The PLRV of perturbation using the  product noise proposed in (\ref{eq:noise-generation}) has the following upper bound
\begin{equation} \label{eq: plrv-pnpm}
    \begin{aligned}
\mathrm{\PLRV}_{\mathcal{M}} (\mathbf{s})\leq  \frac{\Delta_2f}{\sigma_M}\left (R + \frac{\Delta_2f}{\sigma_M}\right) \frac{1}{\sin \theta},
    \end{aligned}
\end{equation}
where   $R\sim\chi_1$  and $\theta$ is the \textbf{random angle} formed by a random noise $\mathbf{n}$ and $\bm{v}$. We use $\bm{v}$ to denote the difference between the computation function $f(\cdot)$ evaluated on a pair of neighboring dataset $\bm{x}$ and $\bm{x}'$, i.e., $\bm{v} \triangleq f(\bm{x})-f(\bm{x}')$ shown in Figure~\ref{fig:geo-dp}.
\end{proposition}

\begin{figure}[htbp]
\centering
\includegraphics[width=0.8\linewidth]{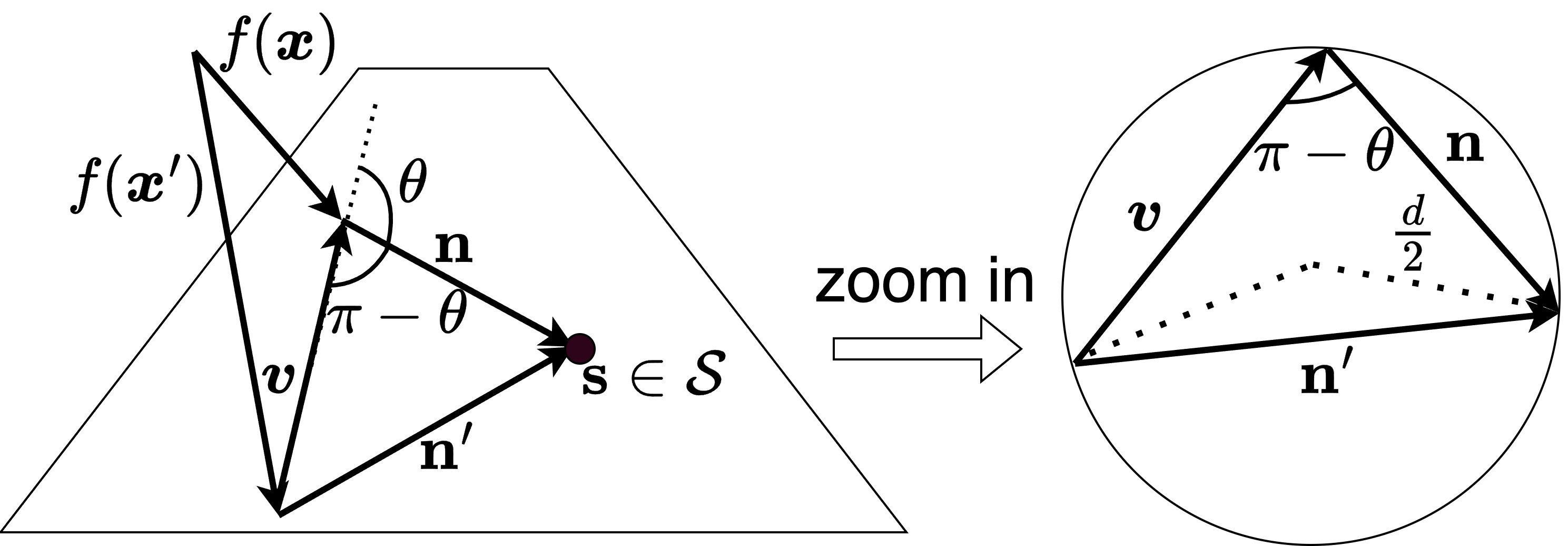}
\Description{Geometric interpretation of differentially private output perturbation.}
\caption{Geometric interpretation of noise perturbation.}
\label{fig:geo-dp}
\end{figure}

Figure~\ref{fig:geo-dp} geometrically visualizes that, to protect privacy, we are interested in the difference $\bm{v} \triangleq f(\bm{x})-f(\bm{x}')$, since it is what additive noise must obscure~\cite[p. 264]{dwork2014algorithmic}. In other words, $\bm{v}$, $\mathbf{n}$ (used to perturb $f(\bm{x})$), and $\mathbf{n}'$ (used to perturb $f(\bm{x}')$)  forms a   \textbf{random} triangle, 
i.e., $\mathbf{n}' - \mathbf{n} = \bm{v}$. The PLRV in (\ref{eq: plrv-pnpm}) differs from the existing works by explicitly considering this geometric interpretation. In particular, the result of (\ref{eq: plrv-pnpm}) represents a diameter (scaled by a constant $\frac{\Delta_2f}{\sigma_M}$) of the circumcircle of a formed \textbf{random} triangle. This geometric idea is inspired by~\cite{ji-r1smg}, which is the seminal work that leverages the randomness of a triangle to perform measure concentration analysis on PLRV. However, \cite{ji-r1smg} imposes a strong assumption on the independence of $\mathbf{n}$ and $\mathbf{n}'$  and consider that the angle formed by $\mathbf{n}$ and $\mathbf{n}'$ are nearly orthogonal (i.e., within a range of $\pm\theta_0$ of $\frac{\pi}{2}$ for some small $\theta_0$). The mechanism developed in \cite{ji-r1smg}  only works under extremely restricted privacy parameters, e.g., $\epsilon<\frac{1}{M}$. In this work, we (i)  lift the assumption on the independence between $\mathbf{n}$ and $\mathbf{n}'$, (ii) consider the random angle $\theta$ formed by $\mathbf{n}$ and $f(\bm{x})-f(\bm{x}')$, (iii) explicitly derive the closed-form PDF of $\frac{1}{\sin \theta}$. Our proposed product noise in (\ref{eq:noise-generation}) works for all $\epsilon>0$.

Next, we provide the proof of Proposition~\ref{prop: PLRV}.
\begin{proof}
    Let $\mathbf{n} = \sigma_M R\boldsymbol{h}$ and $\mathbf{n}' = \sigma_M R'\boldsymbol{h}'$ be the noise used to perturb $f(\bm{x})$ and $f(\bm{x}')$, respectively. Then, we have $f(\bm{x})+\mathbf{n} = f(\bm{x}')+\mathbf{n}' = \mathbf{s}\in\mathcal{S}$ and $\frac{||\mathbf{s}-f(\bm{x})||_2}{\sigma_M} \sim \chi_1$, $\frac{||\mathbf{s}-f(\bm{x}')||_2}{\sigma_M} \sim \chi_1$.

By defining the corresponding PDF notations, e.g., $f_{R\bm{h}}$, $f_{R'\bm{h'}}$, $f_{R}$, and $f_{\bm{h}}$,
the PLRV associated with our product noise in (\ref{eq:noise-generation}) is
\begin{equation}
\begin{aligned}
&\mathrm{\mathrm{\PLRV}}_{\mathcal{M}} (\mathbf{s})\\
\stackrel{(a)}=& \ln\left(\frac{f_{R\bm{h}}   [f(\bm{x})+ \sigma_M R\bm{h} = \mathbf{s} \big| \mathbf{s} \in \mathcal{S}]}{f_{R\bm{h}}   [f(\bm{x}')+ \sigma_M R'\bm{h}' = \mathbf{s}\big| \mathbf{s}  \in \mathcal{S}]}\right)  
\\
=& \ln\left(\frac{f_{R\bm{h}}   [R\bm{h} = \frac{\mathbf{s}- f(x)}{\sigma_M} \big| \mathbf{s} \in \mathcal{S}]}{f_{R'\bm{h}'}  [R'\bm{h}' = \frac{\mathbf{s} -f(\bm{x}')}{\sigma_M} \big| \mathbf{s} \in \mathcal{S}]}\right)\\
\stackrel{(b)}=& \ln\left(\frac{f_{R}  [R= \frac{||\mathbf{s}-f(\bm{x})||_2}{\sigma_M}] }{f_{R' }  [R'=\frac{||\mathbf{s}-f(\bm{x}')||_2}{\sigma_M}] } \cdot \frac{ f_{\bm{h}}  [\bm{h} \sim \mathbb{S}^{M-1}]}{f_{\bm{h}'}  [\bm{h}' \sim \mathbb{S}^{M-1}]}\right) 
\\ 
\stackrel{(c)}=& \ln\left(\frac{f_{R}  [R= \frac{||\mathbf{s}-f(\bm{x})||_2}{\sigma_M}] }{f_{R'}  [R'=\frac{||\mathbf{s}-f(\bm{x}')||_2}{\sigma_M}]}\right) \\
\stackrel{(d)}=& \ln\left(\frac{\frac{1}{2^{-1/2} \Gamma (1/2)} e^{-\left(\frac{||\mathbf{s}-f(\bm{x})||_2}{\sigma_M} \right)^{2}/2 } }{\frac{1}{2^{-1/2} \Gamma (1/2)} e^{-\left(\frac{||\mathbf{s}-f(\bm{x}')||_2}{\sigma_M} \right)^{2}/2 }}\right) 
\\
=& \frac{1}{2\sigma_M^{2}}(||\mathbf{s}-f(\bm{x}')||_2^{2}-||\mathbf{s}-f(\bm{x})||_2^{2}) \\
=& \frac{1}{2\sigma_M^{2}}(||\mathbf{s}-f(\bm{x}')||_2 + ||\mathbf{s}-f(\bm{x})||_2)(||\mathbf{s}-f(\bm{x}')||_2-||\mathbf{s}-f(\bm{x})||_2)\\ 
\leq& \frac{1}{2\sigma_M^{2}}(||\mathbf{n}'||_2+||\mathbf{n}||_2)||\bm{v}||_2,
 \label{eq:plrv-chi}
\end{aligned}
\end{equation}
where $(a)$ is due to Definition~\ref{def:plrv} and Remark~\ref{remark:plrv}, $(b)$ follows from Lemma~\ref{lemma:prod-measure}, and $(c)$ holds because both $\boldsymbol{h}$ and $\boldsymbol{h}'$ are uniformly distributed on $\mathbb{S}^{M-1}$ and their probability densities cancel with each other.  Finally, $(d)$ is obtained by 
substituting in $R\sim\chi_1$ and $R'\sim\chi_1$.

Next, we proceed to bound (\ref{eq:plrv-chi}) using the following geometric interpretation. For  any specific output  $\mathbf{s}\in\mathcal{S}$ and the noise vectors  $\mathbf{n}$ and $\mathbf{n}'$ used to obscure $f(\bm{x})$ and $f(\bm{x}')$ in $\R^M$, we have
\begin{equation*}
    \begin{aligned}
        \mathcal{M}(\bm{x}) = \mathcal{M}(\bm{x}') =\mathbf{s}\in\mathcal{S}  
        &\Leftrightarrow f(\bm{x})+\mathbf{n}= f(\bm{x}')+\mathbf{n}'=\mathbf{s}\in\mathcal{S} \\
        &\Leftrightarrow \mathbf{n}' - \mathbf{n} = f(\bm{x})-f(\bm{x}') = \bm{v},
    \end{aligned}
\end{equation*}
which implies that $\bm{v} \triangleq f(\bm{x})-f(\bm{x}')$, $\mathbf{n}$ and  $\mathbf{n}'$  all belong to a specific hyperplane. 
Such a hyperplane is non-deterministic and its randomness is determined by $\mathcal{M}$. As shown in Figure~\ref{fig:geo-dp}, under the definition of PLRV, $\bm{v}$, $\mathbf{n}$, and $\mathbf{n}'$   forms a   \textbf{random} triangle, i.e., $\mathbf{n}' - \mathbf{n} = \bm{v}$. For any arbitrary $\bm{v}$, the direction of $\mathbf{n}$ is independent of the direction of $\bm{v}$ due to the considered noise generation scheme (\ref{eq:noise-generation}). We use $\theta$ to represent the angle between $\bm{v}$ and  $\mathbf{n}$, and use $\frac{d}{2}$ to represent the radius of the circumcircle (cf.  Figure~\ref{fig:geo-dp} (right)). The upper bound of the length of each edge of a triangle is the diameter of the circumcircle of the triangle. According to the law of sines, we have $d = \frac{||\mathbf{n}' ||_2}{ \sin(\pi - \theta)} = \frac{||\mathbf{n} + \bm{v}||_2}{ \sin(\pi - \theta)}  =\frac{||\mathbf{n} + \bm{v}||_2 }{\sin\theta}$. Thus, 
\begin{equation*}
        \begin{aligned}
        (\ref{eq:plrv-chi}) &\leq \frac{1}{2\sigma_M^{2}} \cdot 2d \cdot ||\bm{v}||_2 \\
        &= \frac{||\bm{v}||_2}{ \sigma_M^{2}} \frac{||\mathbf{n}+\bm{v}||_2}{\sin(\pi - \theta)}\\
        &\leq \frac{||\bm{v}||_2}{ \sigma_M^{2}} \frac{||\mathbf{n}||_2+||\bm{v}||_2}{\sin \theta} \\
        &\leq  \frac{\Delta_2f}{\sigma_M}\left (\frac{||\mathbf{n}||_2}{\sigma_M} + \frac{\Delta_2f}{\sigma_M}\right) \frac{1}{\sin \theta} \\
        &= \frac{\Delta_2f}{\sigma_M}\left (R + \frac{\Delta_2f}{\sigma_M}\right) \frac{1}{\sin \theta}, \quad R\sim\chi_1,
        \end{aligned}
\end{equation*}
which concludes the proof of Proposition~\ref{prop: PLRV}.
\end{proof}

\subsection{Moment Bound of the New PLRV} \label{sec:moments-W-U}
Although Proposition~\ref{prop: PLRV} considers an upper bound of PLRV, it turns out to have tight measure concentration and enables a tight bound on the privacy loss. Based on (\ref{eq: plrv-pnpm}), we define two random variables
\begin{equation}\label{eq:rv-def}
\begin{aligned}
        W  \triangleq R + \lambda, \quad U \triangleq \frac{1}{\sin\theta}
\end{aligned},
\end{equation}
where $\lambda = \frac{\Delta_2f}{\sigma_M}$ and $\theta$ is the random angle formed by the random noise $\mathbf{n}$ and the difference $\bm{v}\triangleq f(\bm{x})-f(\bm{x}')$. Then, to achieve $(\epsilon,\delta)$-DP using the noise proposed in (\ref{eq:noise-generation}), our main technique focuses on investigating a tight measure concentration on the product of two random variables, i.e., 
\begin{equation}\label{eq:plrv-tail}
\Pr\left[WU\geq \frac{\sigma_M}{\Delta_2f}\epsilon\right]\leq \delta,
\end{equation}
which is a sufficient condition for the product noise to satisfy $\Pr[\mathrm{\mathrm{\PLRV}}_{\mathcal{M}} (\mathbf{s})\geq \epsilon]\leq \delta$. Usually, to study the tail bound of the product random variable, $WU$, one first needs to derive its PDF, $f_{WU}$, which involves the Mellin convolution~\cite{springer1979algebra} between two PDFs, i.e., $f_{W}$ and $f_{U}$. To avoid this computationally heavy reasoning, we instead derive the moment bound\footnote{Note that some other DP mechanisms, e.g.,  RDP~\cite{mironov2017renyi} and Moments Accountant~\cite{abadi2016deep}, essentially consider deriving a tail probability of PLRV using the Cram\'er-Chernoff bound, i.e.,  $\Pr[\Upsilon\geq t]\leq \inf_{\lambda>0}\mathbb{E}[e^{\lambda(\Upsilon-t)}]$, where $\Upsilon$ indicates a random variable. However, as proved in~\cite{philips1995moment}, moment bound for tail probabilities are always better than Cram\'er-Chernoff bound. More precisely, let $t>0$, we have the best moment bound for the tail probability  $\Pr[\Upsilon>t] = \min\limits_{q} \mathbb{E}[\Upsilon^q]t^{-q}\leq \inf_{\lambda>0}\mathbb{E}[e^{\lambda(\Upsilon-t)}]$ always hold. Thus, our technique provides tighter bounds on  PLRV.}~\cite{philips1995moment} on the product of $W$ and $U$, which will only involve the individual PDFs due to the independence between $W$ and $U$. To be more specific, we can set 
\begin{equation}\label{eq:our-tail}
\begin{aligned}
        \Pr[WU\geq t] & \leq \min_q \mathbb{E}[(WU)^q]t^{-q} = \min_q \mathbb{E}[W^q]\mathbb{E}[U^q]t^{-q} = \delta \ll 1,
\end{aligned}
    \end{equation}
where $q$ is positive and  the equality holds due to the independence of $W$ and $U$ (as discussed in Lemma~\ref{thm:ss_decomp}, the radius random variable is independent of the directional random variable). 

To analyze (\ref{eq:our-tail}), we first study $\mathbb{E}[W^q]$ and $\mathbb{E}[U^q]$, then derive the tail probability of $\Pr[WU\geq t]$. 

\subsubsection{\textbf{Moment of $W$}} \label{sec: moments-W}

According to (\ref{eq:noise-generation}), we have $\frac{||\mathbf{n}||_2}{\sigma_M} = R\sim\chi_1$. Then, the probability density function  (PDF) of $W$ can be obtained by a simple transformation of a random variable, i.e., $W = R+\lambda$, where $\lambda = \frac{\Delta_2f}{\sigma_M}$. Thus, by applying Lemma~\ref{thm:variable-trans}, and substituting $f(r;1) = \frac{\sqrt{2}}{\sqrt{\pi}}e^{-\frac{r^2}{2}}$ (Definition~\ref{def:chi}),  we have
\begin{equation}\label{eq:pdf_W}
     \begin{aligned}
         f_{W}(w)=\sqrt{\frac{2}{\pi}}e^{-\frac{(w-\lambda)^2}{2}} , w \ge \lambda.
     \end{aligned}
\end{equation}

Next, we can immediately bound the value of $E[W^q]$. 

\begin{proposition}
\label{prop:proof-q-th-moment-chi} An upper bound on the $q$-th moment of random variable $W$ is 
\begin{equation}
\begin{aligned}
\mathbb{E}[W^q] \leq  \ 2^{-\frac{q}{2}} e^{-\frac{\lambda^{2}}{2}} \Gamma(q+1) \Bigg(   &
\frac{1 }{\Gamma \left(\frac{q+2}{2} \right)} {}_1F_1\left(\frac{q+1}{2} ; \frac{1}{2} ; \frac{\lambda^2}{2} \right) 
\\
& +  \frac{\sqrt{2 }\lambda }{\Gamma \left(\frac{q+1}{2} \right)} {}_1F_1\left(\frac{q+2}{2} ; \frac{3}{2} ; \frac{\lambda^2}{2} \right)  \Bigg),
\end{aligned}
\end{equation}
where  $_1F_1(\cdot,\cdot,\cdot)$ is the confluent hypergeometric function of the first kind  defined in Definition~\ref{def:hypergeometric-function}.
\end{proposition}
\begin{proof}
Based on the definition of $q$-th moment, we have
\begin{equation*}\label{eq:moment-shifted-chi}
\begin{aligned}
\mathbb{E}[W^q] 
=& \int_{\lambda}^{\infty} w^{q} f_{W}(w) dw \\
\stackrel{(a)}=& \int_{\lambda}^{\infty} w^{q}\sqrt{\frac{2}{\pi}}e^{-\frac{(w-\lambda)^2}{2}}dw 
\\
=& \sqrt{\frac{2}{\pi}} \int_{\lambda}^{\infty} w^{q}e^{-\frac{1}{2}w^2}e^{\lambda w} e^{-\frac{\lambda^{2}}{2}} dw \\
=&  \sqrt{\frac{2}{\pi}} e^{-\frac{\lambda^{2}}{2}} \int_{\lambda}^{\infty} w^{q}e^{-\frac{1}{2}w^2+\lambda w}dw\\
\stackrel{(b)} \leq &  \sqrt{\frac{2}{\pi}} e^{-\frac{\lambda^{2}}{2}} \int_{0}^{\infty} w^{q}e^{-\frac{1}{2}w^2+\lambda w}dw \\
\stackrel{(c)}= & \sqrt{\frac{2}{\pi}} e^{-\frac{\lambda^{2}}{2}} \frac{D_{-q-1}(-\lambda)\Gamma(q+1)}{e^{-\frac{\lambda^2}{4}}},
\end{aligned}
\end{equation*}
where $(a)$ is achieved by substituting in (\ref{eq:pdf_W}), $(b)$ can be achieved by expanding the integration range for non-negative integrand, and $(c)$ is obtained by applying Definition~\ref{de: parabolic-cylinder-function} with $q = -p-1$ (i.e., $p=-q-1$), $z=-\lambda$. Then, by plugging in the expression of the parabolic cylinder function $D_{-q-1}(-\lambda)$  (also see Definition~\ref{de: parabolic-cylinder-function}), we can complete the proof. 
\end{proof}

\subsubsection{\textbf{Moment of $U$}} \label{sec: moments-U}
We first derive the PDF of $U\triangleq  \frac{1}{\sin\theta}$, where $\theta$ is the random angle between $\bm{v}$ and   $\mathbf{n}$ (see Figure~\ref{fig:geo-dp}).

\begin{proposition} \label{prop:proof-pdf-sin}
The PDF of  $U \stackrel{\triangle}= \frac{1}{\sin\theta}$  is  
\begin{equation} \label{eq:pdf-sin}
\begin{aligned}
f_{U}(u) =& \frac{2}{\sqrt{\pi}} \frac{\Gamma\left( \frac{M}{2} \right)}{\Gamma\left( \frac{M-1}{2} \right)} (u^2-1)^{-\frac{1}{2}} u^{1-M}, \quad u\geq 1. 
\end{aligned}
\end{equation}
\end{proposition}

\begin{proofsketch}
 This proposition can be proved by applying a sequence of  transformation of random variables and the connection between $\Beta$ and $\BP$ distributions in Lemma~\ref{lemma:beta-betaprime}. In particular, we   prove the following results
\begin{equation*}
    \begin{aligned}
    & X \triangleq \cos\theta, f_{X}(x) = \frac{1}{\sqrt{\pi}} \frac{\Gamma\left( \frac{M}{2} \right)}{\Gamma\left( \frac{M-1}{2} \right)} (1 - x^2)^{\frac{M-3}{2}}, x \in [-1,1], \\
    & Y \triangleq 1 - X^2 = \sin^2\theta, f_{Y}(y) = \frac{1}{\sqrt{\pi}} \frac{\Gamma\left( \frac{M}{2} \right)}{\Gamma\left( \frac{M-1}{2} \right)} y^{\frac{M-3}{2}} \frac{1}{\sqrt{1 - y}},  y \in [0,1], \\
    & Z \triangleq \frac{1}{Y} = \frac{1}{\sin^2\theta}, f_Z(z) = \frac{1}{\sqrt{\pi}} \frac{\Gamma\left( \frac{M}{2} \right)}{\Gamma\left( \frac{M-1}{2} \right)} (z - 1)^{-\frac{1}{2}} z^{-\frac{M}{2}},  z \in [1, \infty), \\
    & U \triangleq \sqrt{Z} = \frac{1}{\sin\theta}, f_U(u) = \frac{2}{\sqrt{\pi}} \frac{\Gamma\left( \frac{M}{2} \right)}{\Gamma\left( \frac{M-1}{2} \right)} (u^2 - 1)^{-\frac{1}{2}} u^{1 - M}, u \in [1, \infty).
\end{aligned}
\end{equation*}
All detailed steps are deferred to Appendix~\ref{app:proof-pdf-sin}.
\end{proofsketch}

Next, we arrive at the $q$-th moment of $U$.
\begin{proposition}
\label{prop:proof-q-th-moment-sin} The
$q$-th moment $\left(1<q<M-1 \right)$ of the  random variable $U\triangleq \frac{1}{\sin\theta}$ is 
\begin{equation}
\begin{aligned}
\mathbb{E}[U^q] =& \frac{\Gamma\left( \frac{M}{2} \right)}{\Gamma\left( \frac{M-1}{2} \right)} \frac{\Gamma \left ( \frac{M-q}{2}-\frac{1}{2} \right )}{\Gamma \left ( \frac{M-q}{2} \right )}. 
\end{aligned}
\end{equation}
\end{proposition}

\begin{proofsketch}
Based on the definition of moments, we have
\begin{equation*}
\begin{aligned}
    \mathbb{E}[U^q] 
    &= \int_{1}^{\infty} u^{q} f_{U}(u) du \\
    &= \frac{2}{\sqrt{\pi}} \cdot \frac{\Gamma\left( \frac{M}{2} \right)}{\Gamma\left( \frac{M-1}{2} \right)} 
    \int_{1}^{\infty} u^{q} (u^2 - 1)^{-\frac{1}{2}} u^{1 - M} \, du \\
    &\stackrel{(*)}{=} \frac{\Gamma\left( \frac{M}{2} \right)}{\Gamma\left( \frac{M-1}{2} \right)} \cdot 
    \frac{\Gamma \left( \frac{M - q}{2} - \frac{1}{2} \right)}{\Gamma \left( \frac{M - q}{2} \right)},
\ q \in[1,M-1].
\end{aligned}
\end{equation*}
Step $(*)$ can be evaluated by invoking the kernel function of the $\BP$ $(\frac{1}{2},\frac{M-q-1}{2})$ distribution, which requires $\frac{M-q-1}{2}> 0$. The detailed steps are  deferred to Appendix~\ref{app:proof-q-th-moment}.
\end{proofsketch}

\subsubsection{\textbf{Obtaining the tail probability}}\label{sec:Moment-of-WU}

Given $\mathbb{E}[W^q] $, $\mathbb{E}[U^q]$, and (\ref{eq:our-tail}), we   arrive at the following tail probability, $\Pr[WU\geq t]$. 
\begin{proposition} \label{prop:proof-min-q-result}
Given any $1<q<M-\frac{3}{2}$,  $k>1$,  and  $t^2 > 2 k^{\frac{2}{q}} \frac{ \left(\frac{q+3}{2} \right)^{\left(1+\frac{2}{q}\right)}}{e^{(1+\frac{1}{q})}}$, 
we have
\begin{equation}
\label{eq:min-q-result}
\resizebox{1\linewidth}{!}{$
\begin{aligned}
   \Pr[WU \geq t] 
   \leq \min_q \underbrace{\frac{1}{k} \cdot \frac{e^{-\frac{\lambda^{2}}{2}}}{\sqrt{\pi}} 
   \cdot \frac{{}_1F_1\left(\frac{q+1}{2} ; \frac{1}{2} ; \frac{\lambda^2}{2} \right) 
   + \sqrt{2} \lambda\  {}_1F_1\left(\frac{q+2}{2} ; \frac{3}{2} ; \frac{\lambda^2}{2} \right)}{\sqrt{q+\frac{3}{4}}}
   }_{\delta_1} 
    \cdot \underbrace{\frac{\sqrt{M-1}}{\sqrt{M - q - \frac{3}{2}}}}_{\delta_2}.
\end{aligned}
$}
\end{equation}
\end{proposition}
\begin{proofsketch}
According to (\ref{eq:our-tail}), the tail probability satisfies  $\Pr[WU\geq t]\leq \min_q \mathbb{E}[W^q]\mathbb{E}[U^q]t^{-q}$. Hence,  (\ref{eq:min-q-result}) can be proved by showing that  $\mathbb{E}[W^{q}] t^{-q} \leq \delta_1$ as long as $t^2 > 2 k^{\frac{2}{q}} \frac{ \left(\frac{q+3}{2} \right)^{\left(1+\frac{2}{q}\right)}}{e^{(1+\frac{1}{q})}}$ and $\mathbb{E}[U^{q}] \le \delta_2$ when  $q\in[1,M-\tfrac{3}{2}]$. Details are deferred to Appendix~\ref{app:tail-bound}.
\end{proofsketch}

\subsection{Privacy Guarantee} \label{sec:privacy-guarantee}
Using the tail probability provided in Proposition~\ref{prop:proof-min-q-result}, our main theoretical contribution is stated in Theorem~\ref{thm:mian-thm}.
\begin{theorem}\label{thm:mian-thm}
Given a function $f(\bm{x})\in\R^M$ ($M>2$) with $l_2$ sensitivity $\Delta_2f$ and any constant $k>1$, if we set 
\begin{equation*}
    \sigma_M^2 \geq \frac{(\Delta_2f)^2}{\epsilon^2}t^2  \quad\text{and} \quad t^2 = 2 k^{\frac{4}{M}} \frac{ \left(\frac{M}{4} +\frac{3}{2} \right)^{\left(1+\frac{4}{M}\right)}}{e^{(1+\frac{2}{M})}}, 
\end{equation*}
then the  product noise    in (\ref{eq:noise-generation}) achieves  $(\epsilon,\delta)$-DP on the perturbed $f(\bm{x})$, where $\epsilon>0$. By setting   $\lambda = \frac{\Delta_2f}{\sigma_M}$, the corresponding  $\delta$ is given by
\begin{equation}\label{eq:delta-accurate}
    \textstyle
    \delta= \frac{e^{-\frac{\lambda^{2}}{2}}}{k\sqrt{\pi}}  
    \left(  {}_1F_1\left(\frac{M}{4}+\frac{1}{2} ; \frac{1}{2} ; \frac{\lambda^2}{2} \right) 
    + \sqrt{2}\lambda {}_1F_1\left(\frac{M}{4}+1 ; \frac{3}{2} ; \frac{\lambda^2}{2} \right)  \right)
    \frac{\sqrt{M-1}}{\sqrt{\frac{M}{2}-\frac{3}{2}}\sqrt{\frac{M}{2}+\frac{3}{4}}},
    \end{equation}
where $_1F_1(\cdot,\cdot,\cdot)$ is the confluent hypergeometric function of the first kind (see Definition~\ref{def:hypergeometric-function}).
\end{theorem}

\begin{proofsketch}
    Combining the results of Proposition~\ref{prop:proof-min-q-result}, (\ref{eq:plrv-tail}) and (\ref{eq:our-tail}), it is clear that $(\epsilon,\delta)$-DP can be attained if one sets $\frac{\sigma_M}{\Delta_2f} \epsilon \geq t  $, i.e., $\sigma_M^2\geq \frac{(\Delta_2f)^2}{\epsilon^2}t^2$. The value of $\delta$ can be determined by solving the minimization problem formulated in (\ref{eq:min-q-result}), i.e., 
     \begin{equation*}
    \begin{aligned}
        \delta = \min_{q} \delta_1\delta_2, \quad \text{where}\quad q\in\left(1,M- \frac{3}{2}\right).
    \end{aligned}
\end{equation*}
The detailed steps to solve this optimization problem are deferred to Appendix~\ref{app:approximate-delta}. Our conclusion is that by selecting $q = \frac{M}{2}$, a negligible  $\delta$ (e.g., $<10^{-5}$) can be achieved. Hence, by letting $q = \frac{M}{2}$, we obtain $\delta$ as shown in (\ref{eq:delta-accurate}), which completes the proof.
\end{proofsketch}

In (\ref{eq:delta-accurate}), the constant $k$ serves as a tuning parameter; increasing $k$ results in a smaller $\delta$, as will be shown in Figure~\ref{fig:simulation} (a) in Section~\ref{sec:guidance-and-simulation}. In high-dimensional scenarios (which is the focus of this paper),   even with a large $k$, its impact on $\sigma_M^2$ remains minimal as the term $k^{\frac{4}{M}}$ approaches $1$   when $M$ is sufficiently large. 
We also provide the evaluation of the accurate $\delta$ (in (\ref{eq:delta-accurate})) and an approximate version of $\delta$  in Appendix \ref{app:approximate-delta}.

Theorem 1 gives an implicit recursive relationship between the noise parameter $\sigma_M$ and the privacy parameters ($\epsilon$, $\delta$). To make it practical, we show how to obtain $\sigma_M$ given the specific  privacy parameters $(\epsilon^*, \delta^*)$ in Algorithm~\ref{alg:pn-calibration}. 

\begin{algorithm}[htp]
\caption{Procedure to obtain $\sigma_M$ in Theorem~\ref{thm:mian-thm}}
    \label{alg:pn-calibration}
\begin{algorithmic}[1]
\Require Dimension $M$; sensitivity $\Delta_2 f$; initial tuning parameter $k$; privacy parameters $(\epsilon^*,\delta^*)$; step factor $\alpha>1$
\Ensure Noise parameter $\sigma_M$ that satisfies $(\epsilon^*,\delta^*)$-DP
\State Compute $\sigma_M$ and $\delta$ using Theorem~\ref{thm:mian-thm} with $\epsilon^*$ and $k$
\While{$\delta > \delta^*$}
  \State \textbf{Increase} $k$: $k \gets \alpha k$, and \textbf{Repeat} Line 1
\EndWhile
\State \Return $\sigma_M$
\end{algorithmic}
\end{algorithm}

\subsection{Guidance and Simulation} \label{sec:guidance-and-simulation}
In this section, we aim to provide an intuitive sense of how our proposed noise improves upon the classic Gaussian noise~\cite{dwork2006our} and analytic Gaussian noise~\cite{balle2018improving}. The comparison with the classic Gaussian noise is theoretical (considering both large and small dimension $M$), while that with the analytic Gaussian noise is numerical. A fully theoretically principled comparison between analytic Gaussian Mechanism and our product noise is computationally intractable. This is because our privacy parameter involves the confluent hypergeometric function, while that of the analytic Gaussian Mechanism relies on the standard error function ($\mathrm{erf}$) of the Gaussian CDF; neither admits a closed-form inverse in terms of elementary functions. Consequently, for the comparison with the analytic   Gaussian noise, we focused on numerical simulations and empirical evaluations. 

\subsubsection{Guidance on Mechanism Selection}
 To ensure fair comparison, we consider both our product noise and classic Gaussian noise achieve exact $(\epsilon,\delta)$-DP using the least noise scale, i.e., $\sigma_M =\frac{\Delta_2f}{\epsilon}t$ for our product noise and $\sigma = \frac{\Delta_2f}{\epsilon} \sqrt{2\log (\frac{1.25}{\delta})}$ for the classic Gaussian noise. 
We consider the ratio between the expected squared magnitude of the product noise and that of the classic Gaussian noise given a certain dimension $M$, i.e., 
\begin{equation*}
    \begin{aligned}
        f(M) = \frac{\mathbb{E}[||\mathbf{n}||_2^2]}{\mathbb{E}[||\mathbf{n}_{\mathrm{classic}}||_2^2]} = \frac{\sigma_M^2 \mathbb{E}[\chi_1^2]}{\sigma^2M} = \frac{\sigma_M^2}{\sigma^2M},
    \end{aligned}
\end{equation*}
 and study the behavior of $f(M)$ when $M$ is finite dimension and $M$ approaches infinity. The conclusions are summarized in Corollary~\ref{corollary:mechanism_guidance} and~\ref{corollary:Asymptotic_analysis_noise}, respectively, which are all proved in Appendix~\ref{proof:Asymptotic_analysis}.

\begin{corollary}\label{corollary:mechanism_guidance}
Let $\delta=10^{-5}$, for all $\epsilon$ permitted, 
$f(M)<1$ when $M\ge 14$.
\end{corollary}

\begin{remark}
    Although this work targets high-dimensional settings (the minimum dimension $M$ in our experiments presented in  Section~\ref{sec:experiments} is larger than 100), Corollary~\ref{corollary:mechanism_guidance} provides practical guidance for determining when our product noise should be applied. 
    Specifically, when dimension $M \ge 14$, it is preferable to adopt the proposed product noise over the classic Gaussian noise due to its lower noise magnitude. Furthermore, the classic Gaussian mechanism is only valid when $\epsilon \in (0,1)$, whereas our product noise can be applied for any $\epsilon>0$.
\end{remark}

\begin{corollary}\label{corollary:Asymptotic_analysis_noise}
For all $\delta \in (0,1)$   and all $\epsilon$ permitted, $f(M)$ converges to $\frac{1}{4e \log\left(\frac{1.25}{\delta}\right)}$ as $M$ approaches infinity.
\end{corollary}

\begin{remark} \label{Asymptotic_analysis_noise}
In real-world applications, we usually set $\delta = \frac{1}{n}$, where $n$ is the size of the dataset. Then, $\lim_{M\rightarrow \infty} f(M) = \bm{\Theta}\left(\frac{1}{\log n}\right)$;  the expected magnitude of our product noise is reduced by a factor of $\bm{\Theta}(\log n)$ relative to the classic Gaussian noise on average. For example, when $\epsilon =0.1$, $\delta = 10^{-5}$, and  $M$ is sufficiently large, the expected squared magnitude of our noise is approximately $\frac{1}{127}$ of that required by the classic  Gaussian noise. 
\end{remark}

\subsubsection{Simulation}
To corroborate Corollary~\ref{corollary:mechanism_guidance}, Figure~\ref{fig:noise_comparison} presents the mesh plot comparing the expected squared magnitude of our product noise against the classic Gaussian noise and analytic Gaussian noise. We vary $\epsilon$ from 0 to 1, increase dimension $M$ from 14 to 100, and fix $\delta = 10^{-5}$ and $k = 10^5$. Given the same $\epsilon$, both classic Gaussian noise and analytic Gaussian noise (blue and green surfaces) grow linearly with dimension $M$. In contrast, our product noise  (red surface)  consistently maintains lower magnitudes, 
reflecting a sublinear increase in $M$. 
\begin{figure}[htbp]
  \centering
  \begin{minipage}[t]{0.49\linewidth}
    \centering
    \includegraphics[width=0.85\linewidth]{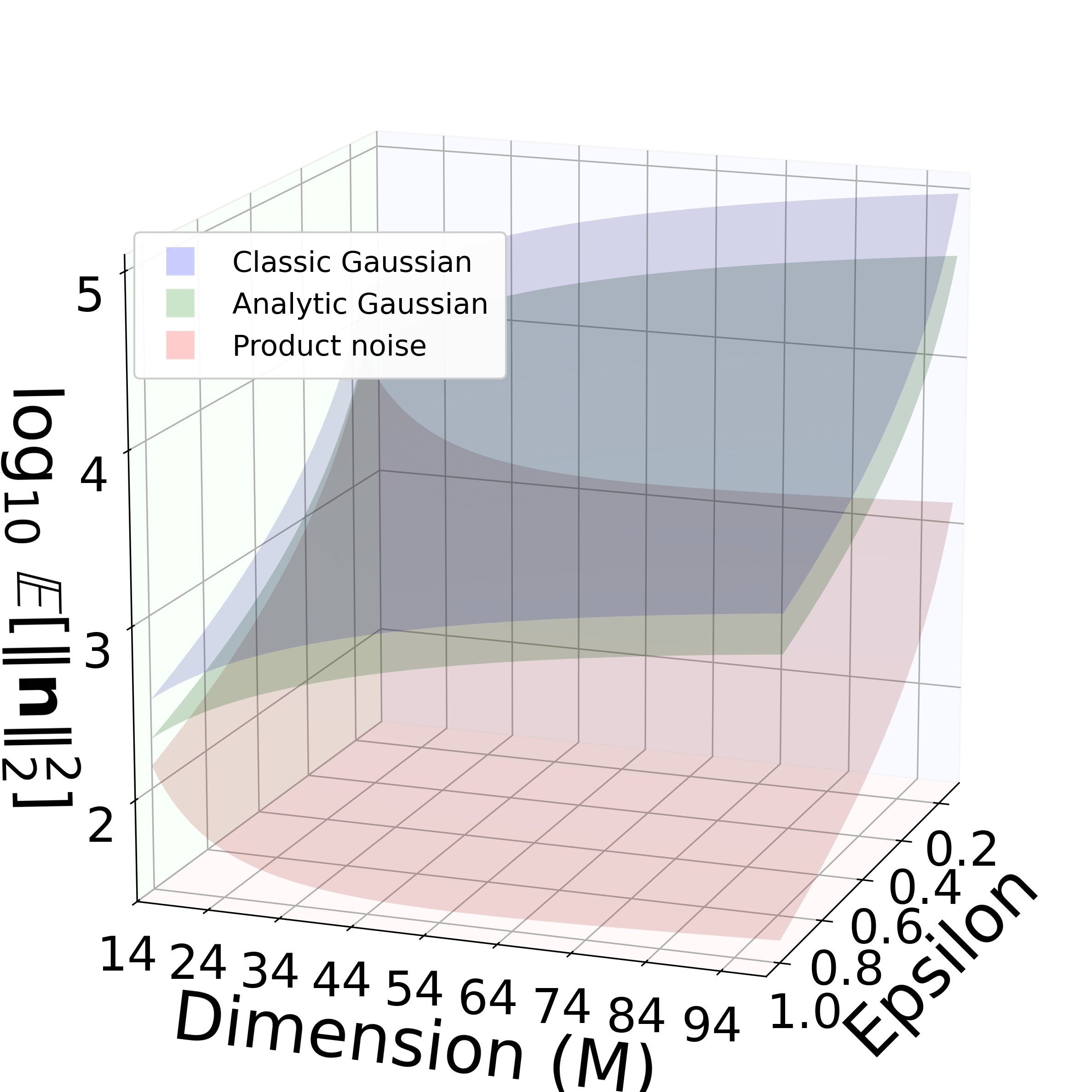}
    
    \captionof{figure}{$\log_{10} \mathbb{E}[||\mathbf{n}||_2^2]$ under varying $M$ and $\epsilon$. 
    }
    \Description{plot comparing product Gaussian noise and classical Gaussian noise as dimension increases.}
    \label{fig:noise_comparison}
  \end{minipage}
  \hfill
  \begin{minipage}[t]{0.49\linewidth}
  \centering
  \includegraphics[width=\linewidth]{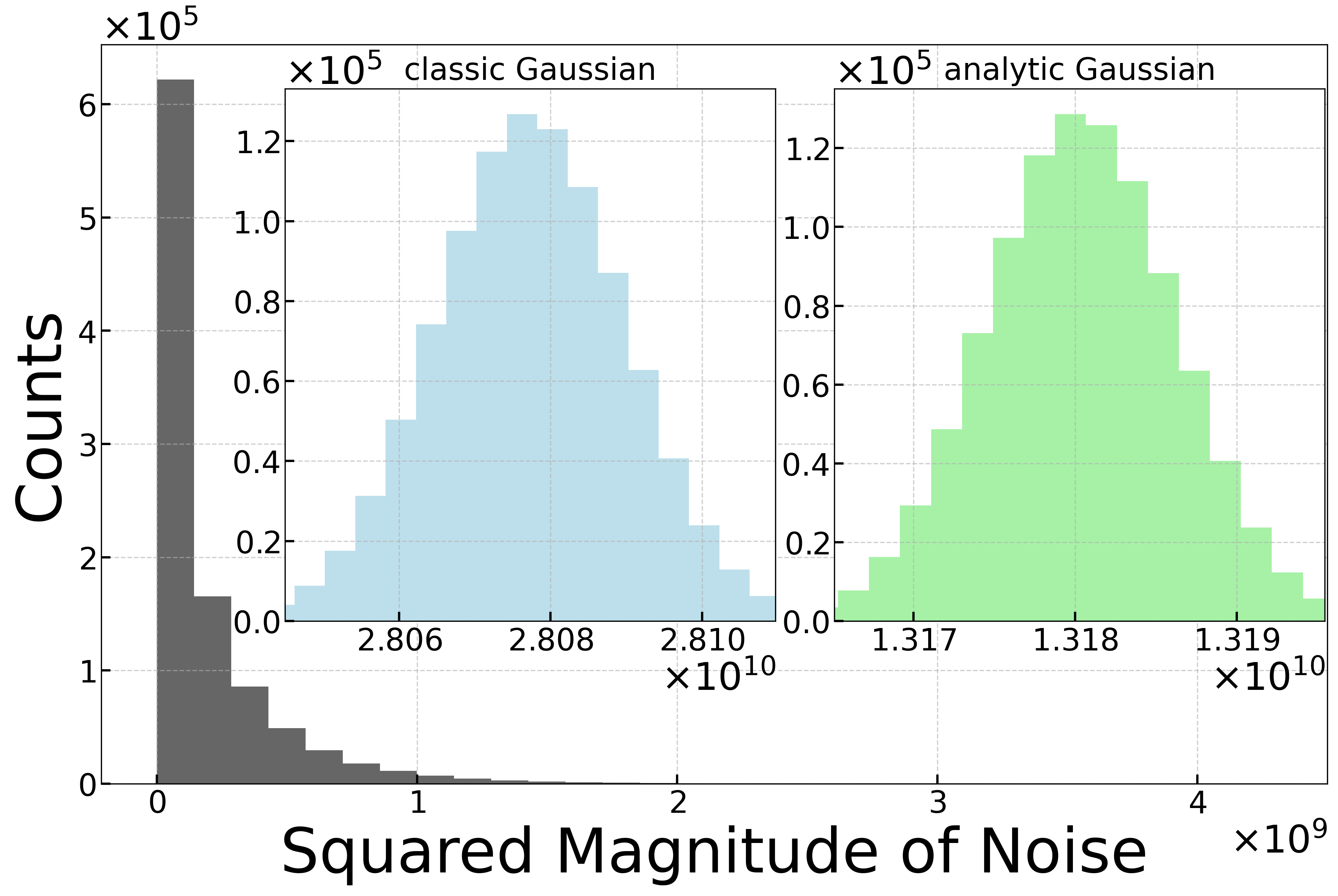}
  \captionof{figure}{Histogram of $\|\mathbf{n}\|_2^{2}$ for various noise.}
  \Description{Histogram showing the squared magnitude of noise for the proposed noise mechanism.}
  \label{fig:histogram-proposed-noise}
\end{minipage}
\end{figure}

To corroborate Corollary~\ref{corollary:Asymptotic_analysis_noise}, we also perform simulations in high-dimensional settings. In Figure~\ref{fig:histogram-proposed-noise}, we compare the histograms of squared noise magnitude ($||\mathbf{n}||_2^{2}$) achieved by our product noise and the Gaussian noises. Specifically, we set $\Delta_2f=1$, $\epsilon =0.1$, $\delta = 10^{-6}$, dimension $M=10^7$, and the constant $k=10^5$. For each noise, we generate a million samples. As illustrated in Figure \ref{fig:histogram-proposed-noise}, our product noise has lower squared magnitudes (e.g., on the order of $10^9$). In contrast, the classic Gaussian noise and analytic Gaussian noise have high squared magnitudes (e.g., on the order of $10^{10}$). Furthermore, the squared noise magnitude of our product noise is considerably more leptokurtic and right-skewed than that of the classic Gaussian noise and analytic Gaussian noise. This indicates that the kurtosis and skewness of the proposed mechanism are substantially increased compared to the classic Gaussian noise and analytic Gaussian noise, thereby enabling enhanced utility. 
\begin{figure}[htbp]
    \centering
    \begin{subfigure}{0.49\columnwidth}  
        \centering
        \includegraphics[width=\linewidth]{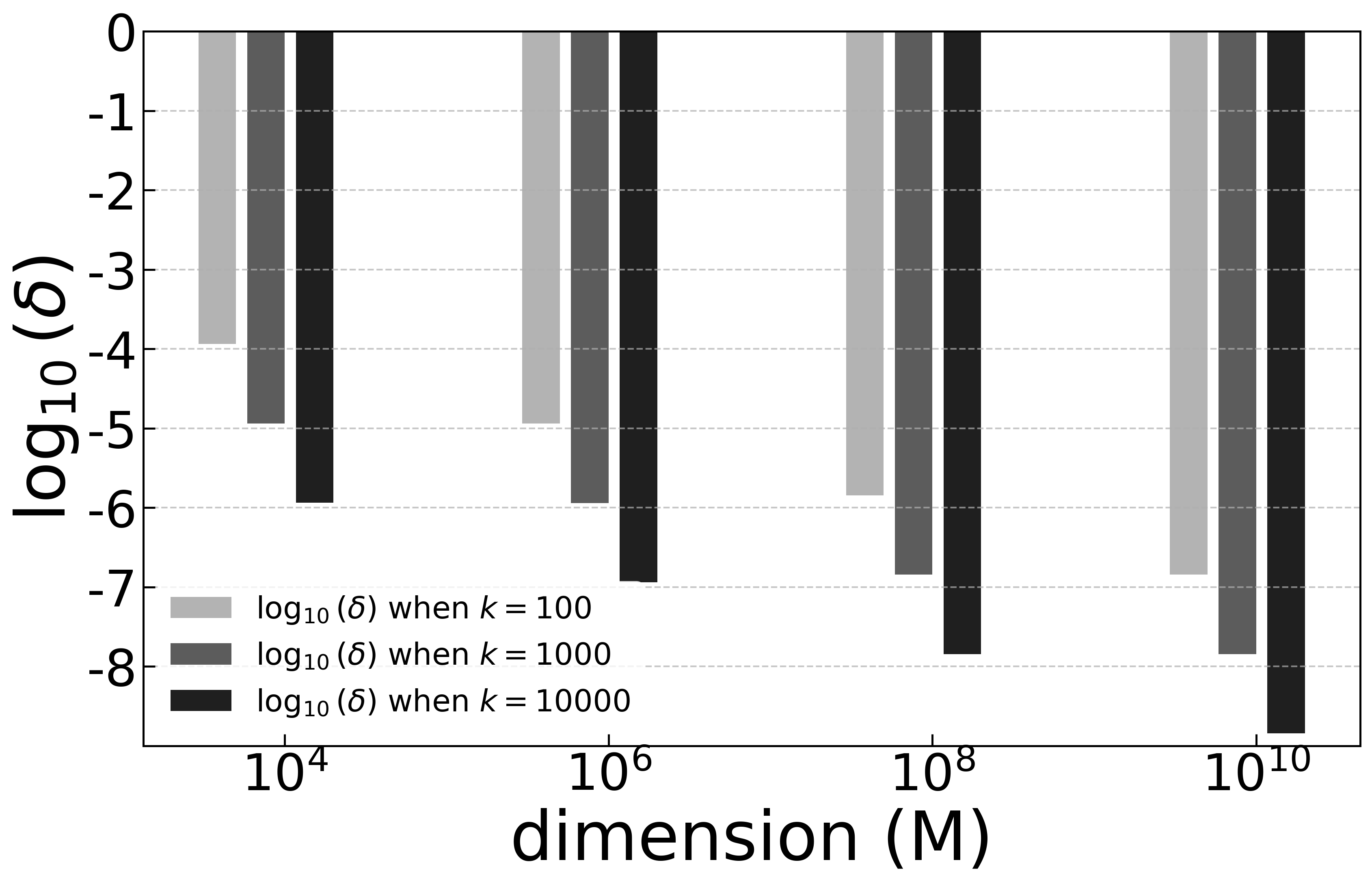}
        \Description{Plot of log10(delta) versus M.}
        \caption{$\log_{10}(\delta)$ v.s. $M$ for various tuning parameters k}
        \label{fig:simulation_delata_dimension}
    \end{subfigure}
    \hfill
    \begin{subfigure}{0.49\columnwidth}
        \centering
        \includegraphics[width=\linewidth]{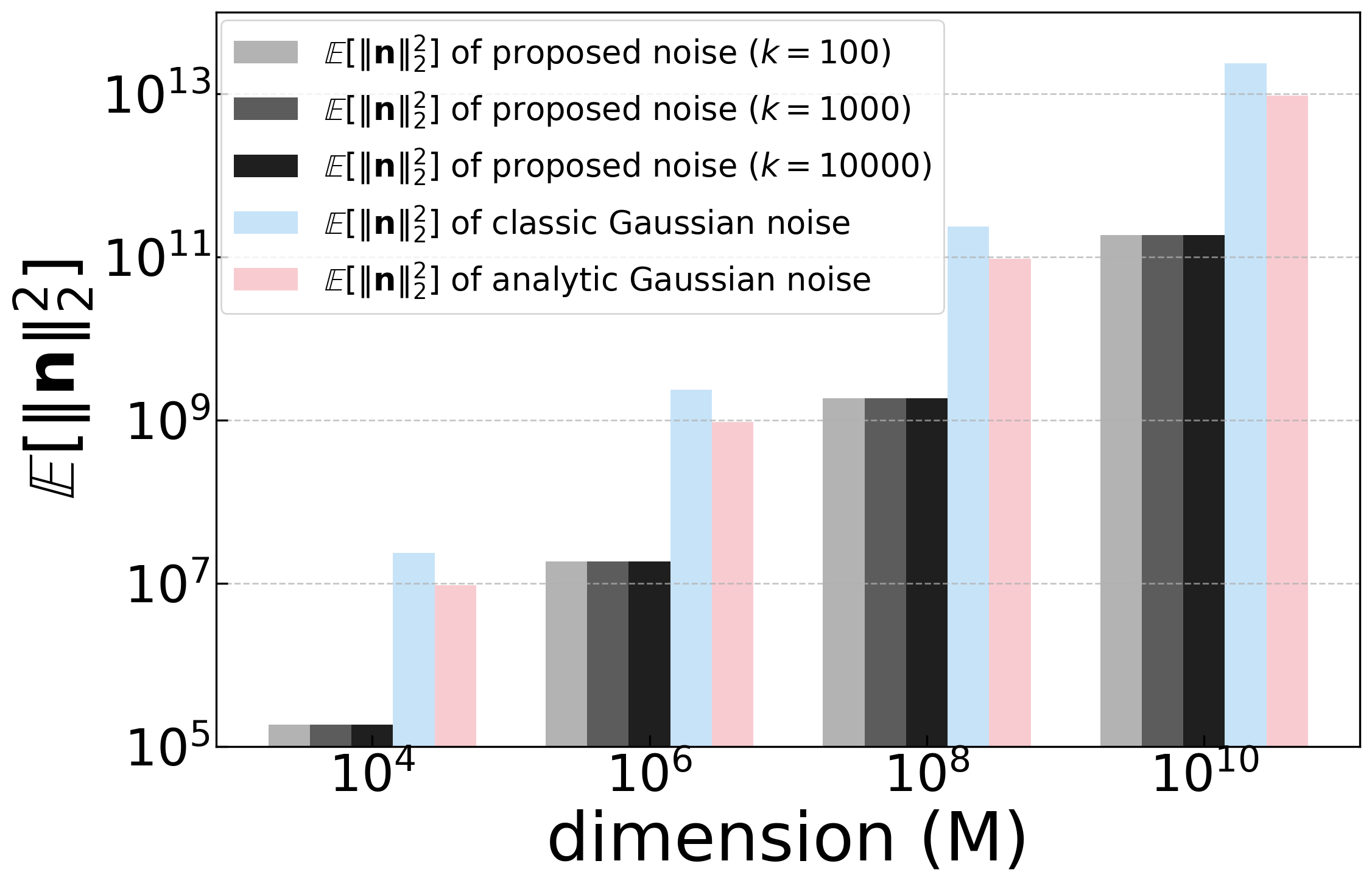}
        \Description{Plot of log10(expected squared noise magnitude) versus dimension M for various noises.}
        \caption{$\mathbb{E}[\|\mathbf{n}\|_2^2]$ v.s.  $M$ for various noises}
        \label{fig:expected squared noise magnitude}
    \end{subfigure}
    \caption{Evaluations of $\delta$ and $\mathbb{E}[||\mathbf{n}||_2^2]$.}
     \label{fig:simulation} 
\end{figure}

To verify the impact of the tuning parameter $k$, we also perform simulations by  considering $\Delta_2f = 1$, $\epsilon = 0.1$, and dimension $M \in \{ 10^4, 10^6, 10^8, 10^{10}\}$. For each $M$, we consider $k \in  \{ 100, 1000, 10000 \} $. First, we show the value of $\log_{10}(\delta)$ achieved by (\ref{eq:delta-accurate}) in Figure~\ref{fig:simulation}(a). Clearly, a small $\delta$ (i.e., $<10^{-5}$) can be obtained by choosing a large value of the constant $k$. In Figure~\ref{fig:simulation}(b),  by fixing $\Delta_2f = 1$ and $\epsilon=0.1$, we evaluate the value of $\mathbb{E}[||\mathbf{n}||_2^2]$, under various  $M$ and $k$. Besides, we also compare with the classic Gaussian noise and analytic Gaussian noise. It is obvious that our noise magnitude  (i) does not change significantly as $k$ increases, and (ii) is far smaller than that of the classic Gaussian noise and analytic Gaussian noise. 
\section{Applying to Privacy-Preserving Convex and Non-convex ERM}\label{sec:case-studies}

In this section, we discuss how our proposed product noise can be applied to solve empirical risk minimization (ERM) while ensuring differential privacy. Specifically, we consider both convex and non-convex ERM problems. For the convex case, we achieve DP guarantees using both output perturbation \cite{chaudhuri2011differentially,wu2017bolt} and objective perturbation \cite{iyengar2019towards}. For the non-convex case, we establish DP guarantees by leveraging gradient perturbation, specifically differentially private stochastic gradient descent (DPSGD) \cite{abadi2016deep, bu2020deep}.

\subsection{Background on ERM}\label{sec:erm_background}
Many classic machine learning or AI tasks, e.g., regression, classification, and deep learning, can be formulated as ERM problems. In general, the formulated problem minimizes a finite sum of loss functions over a training dataset $D$. Given a dataset $D = \left \{ d_1, d_2, \dots, d_n \right \} $ where $d_i \in \mathbb{R}^M$. The goal of an ERM problem is to get an optimal model $\bm{\hat{\omega}} \in \mathbb{R}^M$ via solving the following optimization problem, 
\begin{equation*}\label{eq:erm}
    \begin{aligned}
    \hat{\bm{\omega} } = \arg \min\limits_{\bm{\omega} \in \mathbb{R}^M} \mathcal{L} (\bm{\omega}; D).
\end{aligned}
\end{equation*}
$\mathcal{L}(\bm{\omega}; D) = \frac{1}{n} \sum_{i=1}^{n}  \ell(\bm{\omega}; d_i)$ is   the empirical risk function, and $\ell (\bm{\omega}; d_i)$ is the loss function evaluated on each data sample $d_i$.

\noindent\ul{\textbf{Convex Case.}}  Under this case, the considered $\ell (\bm{\omega}; d_i)$ is convex and usually assumed to be  $L$-Lipschitz (cf. Appendix~\ref{sec:common-assumptions}). Some studies also assume $\ell (\bm{\omega}; d_i)-\ell (\bm{\omega}; {d^{\prime}_i})$ is $L$-Lipschitz, such that $D$ and $D'$ only differ by the data of the $i$-th user, i.e., $d_i$ and $d_i'$. In this paper, we consider using output perturbation and objective perturbation to solve convex ERM with DP guarantees.

$\bullet$ \textbf{Output perturbation.} This approach first obtains the optimal solution $\hat{\bm{\omega} } $ to the convex ERM and then perturbs  $\hat{\bm{\omega} }$ using calibrated Gaussian~\cite{chaudhuri2008privacy, chaudhuri2011differentially} or Laplace noise~\cite{wu2017bolt}. The key idea to establish the privacy guarantee is to examine the gradient and the Hessian of functions $\mathcal{L}(\bm{\omega}; D)$ and $\ \ell(\bm{\omega}; d_i)$.

$\bullet$ \textbf{Objective perturbation.} This approach introduces a noise term (i.e., $\mathbf{n}^{\top}\boldsymbol{\bm{\omega}}$) into the ERM objective function and minimizes the perturbed ERM using privacy-preserving solvers, e.g., the DP Frank-Wolfe~\cite{talwar2015nearly} and stochastic gradient descent (SGD)~\cite{bassily2014private, wu2017bolt}. The privacy guarantee is established by leveraging the fact that both the original $\mathcal{L}(\bm{\omega}; D)$ and the noise term are differentiable everywhere; thus, there exists a unique $\bm{\omega}^*$ that optimizes both $\mathcal{L}(\bm{\omega}; D) + \mathbf{n}^{\top}\bm{\omega}$ and $\mathcal{L}(\bm{\omega};D')+\mathbf{n}^{\top}\bm{\omega}$. Traditional objective perturbation attains DP only at the exact minimum when the gradient becomes $\mathbf{0}$, which is impractical for large-scale optimization due to practical computational limitations. Approximate Minima Perturbation (AMP) \cite{iyengar2019towards} improves the traditional objective perturbation by ensuring DP guarantee even with approximate minima,  i.e., when the SGD algorithm does not converge. Thus, our work adopts the AMP approach in objective perturbation.

\noindent\ul{\textbf{Non-convex Case.}} Under this case, the considered loss function $\ell(\bm{\omega};d_i)$ is non-convex, e.g., image classification using deep neural networks. The most common approach to achieve DP in this case is to perturb the stochastic gradients using calibrated Gaussian noises~\cite{abadi2016deep, bu2020deep}. Since the gradients are perturbed in each iteration, it is important to accurately measure the cumulative privacy loss to make this approach practical. 

\subsection{Output Perturbation for Convex ERM} \label{sec:DP_Convex_Output}
For convex ERM, we consider the following optimization problem,
\begin{equation}\label{eq:obj_with_reg}
    \begin{aligned}
        \mathcal{L}(\bm{\omega}; D) = \frac{1}{n} \sum_{i=1}^{n} \ell(\bm{\omega}; d_i) + \frac{\Lambda}{2n} \|\bm{\omega}\|_2^2,
    \end{aligned}
\end{equation}
where $\Lambda$ is the regularization parameter.

In the classic output perturbation, a calibrated multivariate Gaussian noise $\mathbf{n} \sim \mathcal{N} (0, \sigma^2 \mathbf{I})$ is added to the optimal solution $\bm{\hat{\omega}}$ to guarantee ($\epsilon, \delta$)-DP. It is straightforward to replace the Gaussian noise with our product noise in the output perturbation to achieve a more favorable privacy and utility trade-off. The key steps are as follows.

\ul{First}, compute the optimal solution, i.e.,  $\hat{\bm{\omega}} = \arg \min \limits_{\bm{\omega} \in \mathbb{R}^M} \mathcal{L} (\bm{\omega}; D)$.
\ul{Next}, draw a product noise $\mathbf{n} = {\sigma_M} R \boldsymbol{h}$ ($R\sim\chi_1$ and $\boldsymbol{h} \sim \mathbb{S}^{M-1}$) to obfuscate $\hat{\bm{\omega}}$. The detailed steps are presented in Algorithm~\ref{alg:pn-output} (Product Noise-Based Output Perturbation) in Appendix~\ref{app:Pseudocode-pn-output-perturbation}.

In the following, we present the privacy and utility guarantees of output perturbation using product noise (i.e., Algorithm~\ref{alg:pn-output}). The proofs build upon the ideas in \cite{chaudhuri2011differentially}, with modifications on the analysis to account for the use of our product noise. We show the results in Corollary~\ref{privacy_guarantee_np_output} and ~\ref{utility_guarantee_np_output}. The details in Appendix~\ref{app:output-privacy-proof} and~\ref{app:output-utility-proof}.

\begin{corollary}[Privacy Guarantee of Product Noise-Based Output Perturbation (i.e., Algorithm \ref{alg:pn-output} in Appendix~\ref{app:Pseudocode-pn-output-perturbation})]\label{privacy_guarantee_np_output} 
   The product noise-based Output Perturbation solves (\ref{eq:obj_with_reg}) in a ($\epsilon, \delta$)-DP manner if product noise $\mathbf{n} = {\sigma_M} R \boldsymbol{h}$ is added to the  optimal solution $\bm{\hat{\omega}}$, where     ${\sigma_M}$ is $ \frac{2L}{\Lambda} \cdot \frac{\sqrt{2} k^{\frac{2}{M}} \left(\frac{M}{4} +\frac{3}{2} \right)^{ \left( \frac{1}{2} + \frac{2}{M} \right)}}{\epsilon e^{\left( \frac{1}{2} + \frac{1}{M}\right)}}$. Here,  $\delta$ is determined by a tuning parameter $k$ and the dimension of the problem $M$  (ref. (\ref{eq:delta-accurate})).
\end{corollary}

\begin{corollary} [Utility Guarantee of Product Noise-Based Output Perturbation (i.e., Algorithm \ref{alg:pn-output} in Appendix~\ref{app:Pseudocode-pn-output-perturbation})] \label{utility_guarantee_np_output}
In Algorithm \ref{alg:pn-output}, let $\hat{\bm{\omega}}=\arg \min \limits_{\bm{\omega} \in \mathbb{R}^M}  \frac{1}{n} \sum_{i=1}^{n} \ell(\bm{\omega}; d_i)$, and $\ell(\bm{\omega}; d_i)$ is $L$-Lipschitz, then we have the expected excess empirical risk
\begin{equation*}
    \begin{aligned}
        \mathbb{E} \left [ \mathcal{L} (\bm{\omega}_{\mathrm{priv}}; D) -\mathcal{L }(\bm{\hat{\omega}}; D) \right ] = L \mathbb{E} \left [ \left \| \mathbf{n} \right \|_2 \right ].
    \end{aligned}
\end{equation*}
\end{corollary}

\begin{remark} \label{remark-pn-output}
    According to the   results in Section \ref{sec:main_results}, product noise exhibits much lower expected squared magnitude and a more stable noise distribution (characterized by kurtosis and skewness). It  can significantly reduce the  loss, $\mathcal{L} (\bm{\omega}_{\mathrm{priv}}; D) -\mathcal{L }(\bm{\hat{\omega}}; D)$, under the same privacy parameter compared to classic Gaussian noise. We empirically evaluate this in Section~\ref{sec:case-study-output-perturbation} by considering differentially private logistic regression and support vector machine tasks on 6 benchmark datasets.
\end{remark}

\subsection{Objective Perturbation for Convex ERM}\label{sec:DP_Convex_Objective}
In this section, we consider the   product noise in objective perturbation to solve (\ref{eq:obj_with_reg}) in a differentially private manner. In particular, we consider the Approximate Minimum Perturbation (AMP) approach developed in~\cite{iyengar2019towards}, which is a variant of the classic objective perturbation~\cite{chaudhuri2011differentially,chaudhuri2008privacy} and can establish DP guarantee even if (\ref{eq:obj_with_reg}) does not achieve the global minimum in limited SGD steps.

Similar to~\cite{iyengar2019towards}, we consider (\ref{eq:obj_with_reg}) satisfies the following standard assumptions, i.e., the loss function of each data point, $\ell(\bm{\omega}; d_i)$ is $L$-Lipschitz, convex in $\bm{\omega}$, has a continuous Hessian, and is $\beta$-smooth with respect to both  $\bm{\omega}$ and $d_i$.  Details of the common assumptions for ERM problems are provided in Appendix~\ref{sec:common-assumptions}.

According to~\cite{iyengar2019towards}, AMP introduces noise at two stages to ensure $(\epsilon,\delta)$-DP even when the optimizer does not achieve the exact minimum. To be more specific,   a product noise $\mathbf{n}_{1} = {\sigma_M}_1 R_1 \boldsymbol{h}_1$ ($R_1\sim\chi_1$ and $\boldsymbol{h}_1 \sim \mathbb{S}^{M-1}$) is first added as a linear term to (\ref{eq:obj_with_reg}), i.e., $\mathcal{L}_{\mathrm{priv}}(\bm{\omega}; D) = \frac{1}{n} \sum_{i=1}^{n} \ell(\bm{\omega}; d_i) + \frac{\Lambda}{2n} \|\bm{\omega}\|_2^2 + \mathbf{n}_{1}^{\top} \bm{\omega}$. After obtaining an approximate minimizer of $\mathcal{L}_{\mathrm{priv}}(\bm{\omega}; D)$ denoted as $\bm{\omega}_{\mathrm{approx}}$, another product noise $\mathbf{n}_{2} = {\sigma_M}_2 R_2 \boldsymbol{h}_2$ ($R_2\sim \chi_1$ and $\boldsymbol{h}_2 \sim \mathbb{S}^{M-1}$) is added to obscure $\bm{\omega}_{\mathrm{approx}}$. The value of ${\sigma_M}_1$ and ${\sigma_M}_2$ are calibrated using privacy parameters, $\Lambda$, and the number of training data $n$. The detailed steps are  presented in Algorithm \ref{alg:PN-Based-AMP}  (Product Noise-Based AMP) in Appendix~\ref{app:Pseudocode-pn-amp}.

In what follows, we present the privacy and utility guarantee of Algorithm \ref{alg:PN-Based-AMP}. Since it replaces the  Gaussian noise with our proposed product noise, its privacy and utility guarantees can be established by adapting the proofs in~\cite{iyengar2019towards}. We show the results in Corollary \ref{privacy_guarantee_np_amp} and \ref{utility_guarantee_np_amp}. The proof details are in Appendix~\ref{app:obj-privacy-proof} and \ref{app:obj-utility-proof}.

\begin{corollary} [Privacy Guarantee of Product Noise-Based AMP (i.e., Algorithm \ref{alg:PN-Based-AMP})]\label{privacy_guarantee_np_amp}
    The product noise-based AMP solves (\ref{eq:obj_with_reg}) in a ($\epsilon, \delta$)-DP manner  if ${\sigma_M}_1$ and ${\sigma_M}_2$ are, respectively, set as ${\sigma_M}_1 = \frac{2L}{n} \cdot \frac{\sqrt{2} k^{\frac{2}{M}} \left(\frac{M}{4} +\frac{3}{2} \right)^{ \left( \frac{1}{2} + \frac{2}{M} \right)}}{\epsilon_3 e^{\left( \frac{1}{2} + \frac{1}{M}\right)}}$ and ${\sigma_M}_2 = \frac{n\gamma}{\Lambda} \cdot \frac{\sqrt{2} k^{\frac{2}{M}} \left(\frac{M}{4} +\frac{3}{2} \right)^{ \left( \frac{1}{2} + \frac{2}{M} \right)}}{\epsilon_2 e^{\left( \frac{1}{2} + \frac{1}{M}\right)}}$.  $\delta$ is determined by a tuning parameter $k$ and the dimension of the problem $M$  (ref.  (\ref{eq:delta-accurate})).
\end{corollary}

\begin{corollary} [Utility Guarantee of Product Noise-Based AMP (i.e., Algorithm \ref{alg:PN-Based-AMP})] \label{utility_guarantee_np_amp}
Let $\hat{\bm{\omega}}=\arg \min \limits_{\bm{\omega} \in \mathbb{R}^M}  \frac{1}{n} \sum_{i=1}^{n} \ell(\bm{\omega}; d_i)$, 
and $r = \min\{  M, 2 \cdot (\textit{upper bound on rank of } \ell \textit{'s }$ $\textit{Hessian}) \}$. In Algorithm \ref{alg:PN-Based-AMP}, if $\epsilon_i = \frac{\epsilon}{2}$ for $i \in \{1, 2  \} $, $\epsilon_3 = \max \{ \frac{\epsilon_1}{2}, \epsilon_1- 0.99\} $, and  the regularization parameter  satisfies $\Lambda =\bm{\Theta} \left( \frac{1}{ \left \| \hat{\bm{\omega}} \right \|_2} \left(  \frac{L \sqrt{rM}} { \epsilon} +  n \sqrt{\frac{L\gamma \sqrt{M}}{ \epsilon}}\right)  \right)$, then $\mathbb{E}\left[ \mathcal{L}(\bm{\omega}_{\mathrm{out}}; D) - \mathcal{L}(\hat{\bm{\omega}}; D) \right] =\bm{O} \left(  \frac{ \left \| \hat{\bm{\omega}} \right \|_2 L \sqrt{rM}} { n \epsilon} +  \left \| \hat{\bm{\omega}} \right \|_2 \sqrt{\frac{L\gamma \sqrt{M}}{ \epsilon}}\right)$.
\end{corollary}

\begin{remark} \label{remark-pn-amp}
According to \cite{iyengar2019towards},   $\delta = \frac{1}{n^2}$ ($n$ is the number of training samples), then the utility loss achieved by Gaussian noise-based AMP is  $\bm{O} \left(  \frac{ \left \| \hat{\bm{\omega}} \right \|_2 L \sqrt{rM \log n}} { n \epsilon} +  \left \| \hat{\bm{\omega}} \right \|_2 \sqrt{\frac{L\gamma \sqrt{M \log n}}{ \epsilon}}\right)$~\cite[Theorem 2]{iyengar2019towards}.
Clearly, our product noise-based AMP reduces the utility loss by a factor of $\Theta(\sqrt{\log n})$. It implies that our method can obtain a better ERM model while maintaining equivalent privacy guarantees, and this advantage becomes more significant as the number of samples increases. In Section~\ref{sec:case-study-objective-perturbation}, we empirically validate this theoretical result through experiments on differentially private logistic regression and support vector machine tasks on 4 benchmark datasets.
\end{remark}

\subsection{Gradient Perturbation for Non-Convex ERM}\label{sec:gradient-erm}
In this section, we use the product noise in gradient perturbation to solve ERM where the loss function for each data point, $\ell(\bm{\omega};d_i)$, is non-convex. In non-convex ERM applications (e.g., deep neural networks), Differentially Private SGD (DPSGD) \cite{abadi2016deep} has been widely accepted as the state-of-the-art technique. It extends the traditional SGD by enforcing gradient clipping followed by noise addition to guarantee privacy.
To be more specific, given a mini-batch of training data of size $I_t$, each sample gradient $\bm{g}_t(d_i)$ is first clipped to ensure that its $L_2$-norm does not exceed a predefined threshold $C$, i.e., $\bar{\bm{g}}_t (d_i) = \bm{g}_t(d_i) / \max \{1, \frac{\left \|\bm{g}_t(d_i)  \right \|_2 }{C} \}$. 
Next, the clipped gradients are aggregated, and calibrated Gaussian noise $\mathbf{n}_t \sim \mathcal{N} (0, \sigma^2 C^2 \mathbf{I})$ is added, yielding the noisy gradient $\tilde{\bm{g}}_t  = \frac{1}{I_t} \sum_{i \in I_t} \bar{\bm{g}_t} (d_i) +\mathbf{n}_t$.

According to our theoretical findings (Section~\ref{sec:main_results}), our product noise exhibits superior properties compared to classic Gaussian noise. Under the same privacy parameters, product noise significantly reduces noise magnitude. Consequently, we can replace $\mathbf{n}_t$
with the product noise in   DPSGD to achieve a more favorable privacy and utility trade-off. We defer the pseudocode to  Algorithm~\ref{alg:PN-Based-SGD}  (Product Noise-Based DPSGD) in Appendix~\ref{app:pn-dpsgd}. The privacy guarantee is shown in  Corollary~\ref{privacy_guarantee_np_sgd}  and proved in Appendix~\ref{app:prove-pn-dpsgd}.

 \begin{corollary}[Privacy Guarantee of Product Noise-Based DPSGD (i.e., Algorithm \ref{alg:PN-Based-SGD})]\label{privacy_guarantee_np_sgd} 
     Given a gradient norm clipping parameter $C$, the product noise-based DPSGD achieves ($\epsilon, \delta$)-DP for each gradient descend step, if ${\sigma_M} = \frac{  2 \sqrt{2} C} k^{\frac{2}{M}} \left(\frac{M}{4} +\frac{3}{2} \right)^{ \left( \frac{1}{2} + \frac{2}{M} \right)}{\epsilon e^{\left( \frac{1}{2} + \frac{1}{M}\right)}}$. Note $\delta$ is determined by $k$ and the dimension of the problem $M$  (ref.  (\ref{eq:delta-accurate})).
 \end{corollary} 

\begin{remark} \label{remark-pn-dpsgd}
Given a data sampling probability $p$ and training steps $T$, the cumulative privacy loss of our product noise-based DPSGD is evaluated using privacy amplification via subsampling (see Theorem~\ref{thm:Poisson-Subsampling-ampilification} in Appendix~\ref{app:dp}) followed by a distribution-independent composition~\cite{he2021tighter} (see Theorem~\ref{thm: composition_theorem_1} in Appendix~\ref{app:dp}).

The popular privacy composition techniques, e.g., the moments accountant (MA)~\cite{abadi2016deep} and central limit theorem (CLT)-based composition~\cite{dong2019gaussian}, cannot accommodate our product noise because they specifically require the perturbation noise to be a multivariate Gaussian.  The numerical composition technique,  PRV-based composition~\cite{gopi2021numerical},  is also not applicable because it requires the CDF of the PLRV. However, our PLRV involves the product of two random variables, each governed by a complex distribution, making the CDF analytically intractable.

As a result, the best composition technique that works for our product noise is the considered distribution-independent composition~\cite{he2021tighter}, which makes no assumptions on the noise distribution and remains applicable under arbitrary PLRV forms. In Section~\ref{sec:case-study-gradient-perturbation}, we will show that our method can provide higher privacy guarantees when achieving comparable utility (training and testing accuracy).
\end{remark}

\section{Experiments}\label{sec:experiments}
In this section, we conduct various case studies to evaluate the performance of the product noise in privacy-preserving convex and non-convex ERM. The dataset considered in each case study and the compared methods are summarized in Table \ref{table:dp_erm_summary}.
\begin{table}[htp]
    \caption{Datasets and Compared Methods Overview}
    \label{table:dp_erm_summary}
    \setlength{\tabcolsep}{6pt}
    \renewcommand{\arraystretch}{1.2}
    \small
    \resizebox{\columnwidth}{!}{  
    \begin{tabular}{|>{\centering\arraybackslash}m{2.6cm}|
                    >{\centering\arraybackslash}m{3.5cm}|
                    >{\centering\arraybackslash}m{2.3cm}|}
        \hline
        \textbf{DP-ERM} & \textbf{Compared Methods} & \textbf{Datasets} \\
        \hline
        \multirow{2}{=}{\centering Convex ERM via Output Perturbation} 
            & classic Gaussian noise and analytic Gaussian noise~\cite{chaudhuri2011differentially}
            & \multirow{2}{=}{\centering Adult, KDDCup99, MNIST, Synthetic-H, Real-sim, RCV1} \\
            & multivariate Laplace noise~\cite{wu2017bolt} & \\
        \hline
        Convex ERM via Objective Perturbation 
            & classic Gaussian noise and analytic Gaussian noise based AMP \cite{iyengar2019towards}
            & MNIST, Synthetic-H, Real-sim, RCV1 \\
        \hline
        \multirow{3}{=}{\centering Non-Convex ERM via Gradient Perturbation} 
            & classic DPSGD using MA~\cite{abadi2016deep}
            & \multirow{3}{=}{\centering Adult, IMDb, MovieLens, MNIST, CIFAR-10} \\
            & classic DPSGD using CLT~\cite{bu2020deep} & \\
            & classic DPSGD using PRV~\cite{gopi2021numerical} & \\
            \hline
    \end{tabular}
    \label{tab:method_overview}
    }
\end{table}

\noindent\textbf{Convex  ERM  Experiment Setups.} For the convex case, we consider Logistic Regression (LR) and Huber Support Vector Machine (SVM) as 
applications. The corresponding loss functions are provided in Appendix~\ref{sec:convex-setup}, and the dataset statistics are summarized in Table~\ref{tab:datasets_objective} therein. For output perturbation, we compare our product noise-based output perturbation with that using the classic Gaussian noise~\cite{chaudhuri2011differentially}, analytic Gaussian noise~\cite{balle2018improving}, and multivariate Laplace noise (i.e., Permutation-based SGD)~\cite{wu2017bolt}. Regarding objective perturbation, we adopt the same strategy for hyperparameter-free (H-F) optimization in~\cite{iyengar2019towards} and compare product noise-based AMP with H-F AMP that uses both classic Gaussian noise and analytic Gaussian noise. For all datasets used in the convex case studies, we randomly shuffled the data and split it into $80\%$ training and $20\%$ testing sets. We evaluate the performance using the averaged test accuracy and standard deviation, which reflect models' stability across various experiment instances. 

The goal of these convex experiment cases is to verify that, under the same privacy parameters $\epsilon$ and $\delta$, convex ERM using our product noise achieves significantly higher utility (i.e., test accuracy) than using alternative noise   (e.g., such as Gaussian and Laplace noises). We highlight the key observations of convex experiments as follows. 

\begin{observation}[Higher Utility under the Same or Stricter Privacy Guarantees] \label{observation-higher-accuracy}
     The model using product noise consistently achieves higher test accuracy in all considered tasks for convex ERM under the same privacy parameter. Notably, on specific datasets, the performance even surpasses that of the non-private baseline. For some datasets,  even under stricter privacy guarantees (i.e., smaller $\epsilon$), the product noise-based method achieves higher accuracy than the other noise in convex ERM    (ref. Section~\ref{sec:case-study-output-perturbation} and~\ref{sec:case-study-objective-perturbation}). 
\end{observation}
\begin{observation}[Higher Stability under the Same Privacy Guarantees] \label{observation-higher-stability}
    In the convex ERM, perturbation using product noise significantly reduces the model utility fluctuation (measured in terms of the standard deviation of test accuracy) under the same privacy guarantee  (ref. Section~\ref{sec:case-study-output-perturbation} and~\ref{sec:case-study-objective-perturbation}).
\end{observation}
\noindent\textbf{Non-convex  ERM  Experiment Setups.} For the non-convex case, we consider various neural networks as classification models. In this case, we compare product noise-based DPSGD with the classic DPSGD using Gaussian noise. As discussed in Remark~\ref{remark-pn-dpsgd}, given a specific data subsampling probability and training epoch, the cumulative privacy loss of our approach is evaluated using privacy amplification followed by a distribution-independent composition~\cite{he2021tighter}. In contrast, the cumulative privacy loss caused by standard DPSGD is evaluated using the MA~\cite{abadi2016deep}, PRV~\cite{gopi2021numerical}, and CLT~\cite{bu2020deep} approaches. 

The non-convex experiments aim to validate that, under comparable utility (i.e., test accuracy), our product noise-based DPSGD achieves stronger privacy guarantees (characterized by smaller privacy parameters $\epsilon$ and $\delta$) than the classic DPSGD. The experiment results are consistent with our analysis in Section~\ref{sec:guidance-and-simulation}, which shows that when introducing noise of comparable magnitude, our mechanism requires smaller privacy parameters, thereby offering tighter privacy preservation. We highlight the key observation in non-convex ERM as follows.

\begin{observation}[Stronger Privacy Guarantees under Comparable Utility] \label{observation-stronger-privacy}
To achieve comparable utility (i.e., testing and training accuracy across epochs)  in non-convex ERM, models trained with product noise require smaller privacy parameters than those using classic Gaussian noise, thus yielding stronger privacy preservation (ref. Section~\ref{sec:case-study-gradient-perturbation}).
\end{observation}

\subsection{Case Study \MakeUppercase{\romannumeral 1}: Output Perturbation} \label{sec:case-study-output-perturbation}
This case study evaluates the privacy and utility trade-off achieved by product noise in output perturbation for LR and Huber SVM tasks. As described in the experiment setups, we compare product noise-based output perturbation with  the classic Gaussian noise~\cite{chaudhuri2011differentially}, analytic Gaussian noise~\cite{balle2018improving}, and the multivariate Laplace noise~\cite{wu2017bolt}. The regularization parameter $\Lambda$ is set to $10^{-2}$ for Adult and KDDCup99 and $10^{-4}$ for MNIST, Synthetic-H, Real-sim, and RCV1. The Lipschitz constant $L$ is fixed at 1, and the tuning parameter $k$ is set to 1000 for all datasets. To cover a wide range of privacy parameters, we vary $\epsilon$ from $\{ 10^{-4}, 10^{-\frac{7}{2}}, 10^{-3}, 10^{-\frac{5}{2}}, 10^{-2}, 10^{-\frac{3}{2}}, 10^{-1} \}$. For privacy parameter $\delta$, we set it to $\frac{1}{n^2}$, where $n$ is the size of the dataset. Each experiment was independently repeated 10 times.
\begin{figure}[htp]
    \centering
    \begin{subfigure}{0.49\columnwidth}
        \centering
        \includegraphics[width=\linewidth]{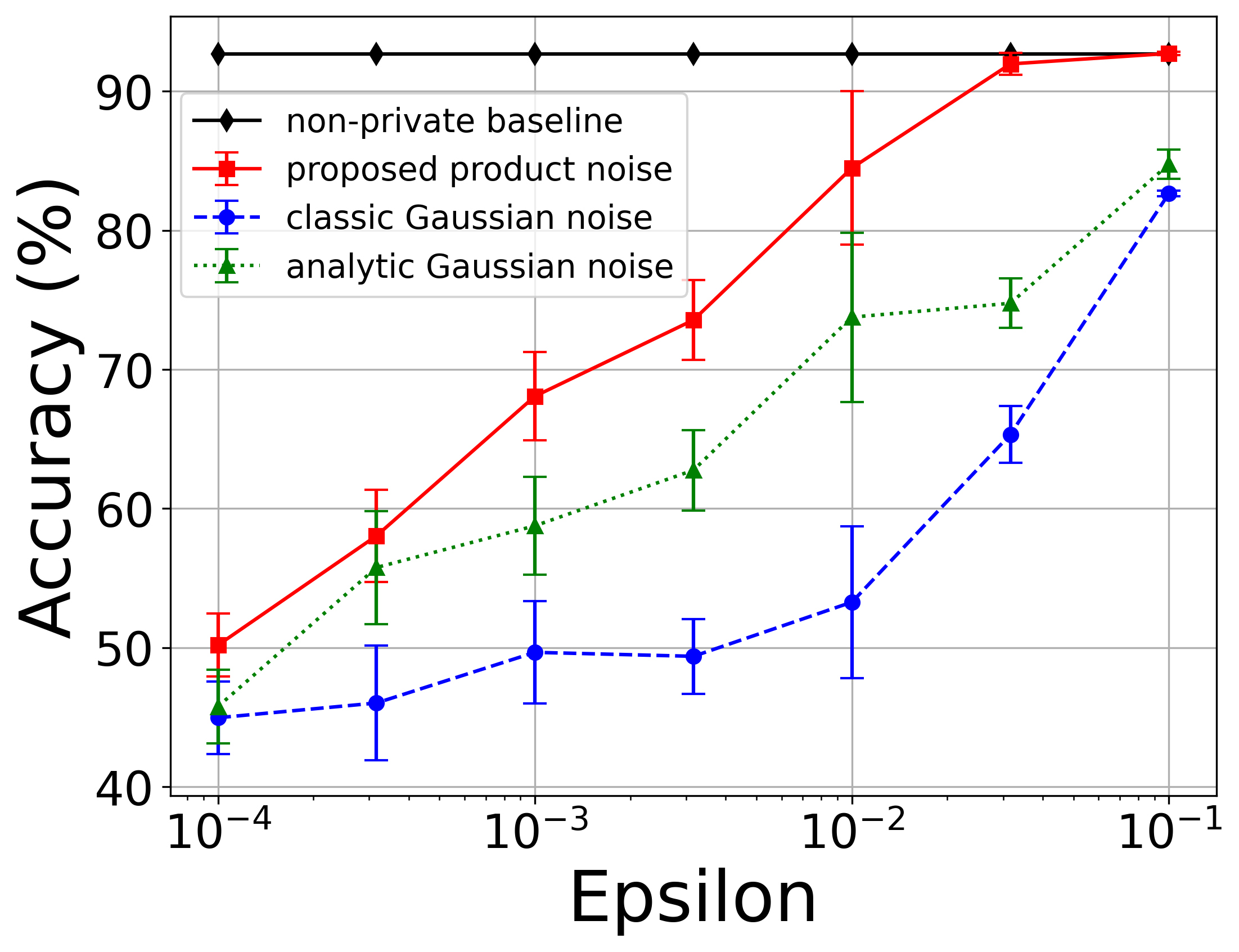}
        \Description{Accuracy results of output perturbation for MNIST dataset (Low-Dim).}
        \caption{\centering  MNIST }
        \label{fig:mnist_test_accuracy_output_lr}
    \end{subfigure}
    \hfill
    \begin{subfigure}{0.49\columnwidth}
        \centering
        \includegraphics[width=\linewidth]{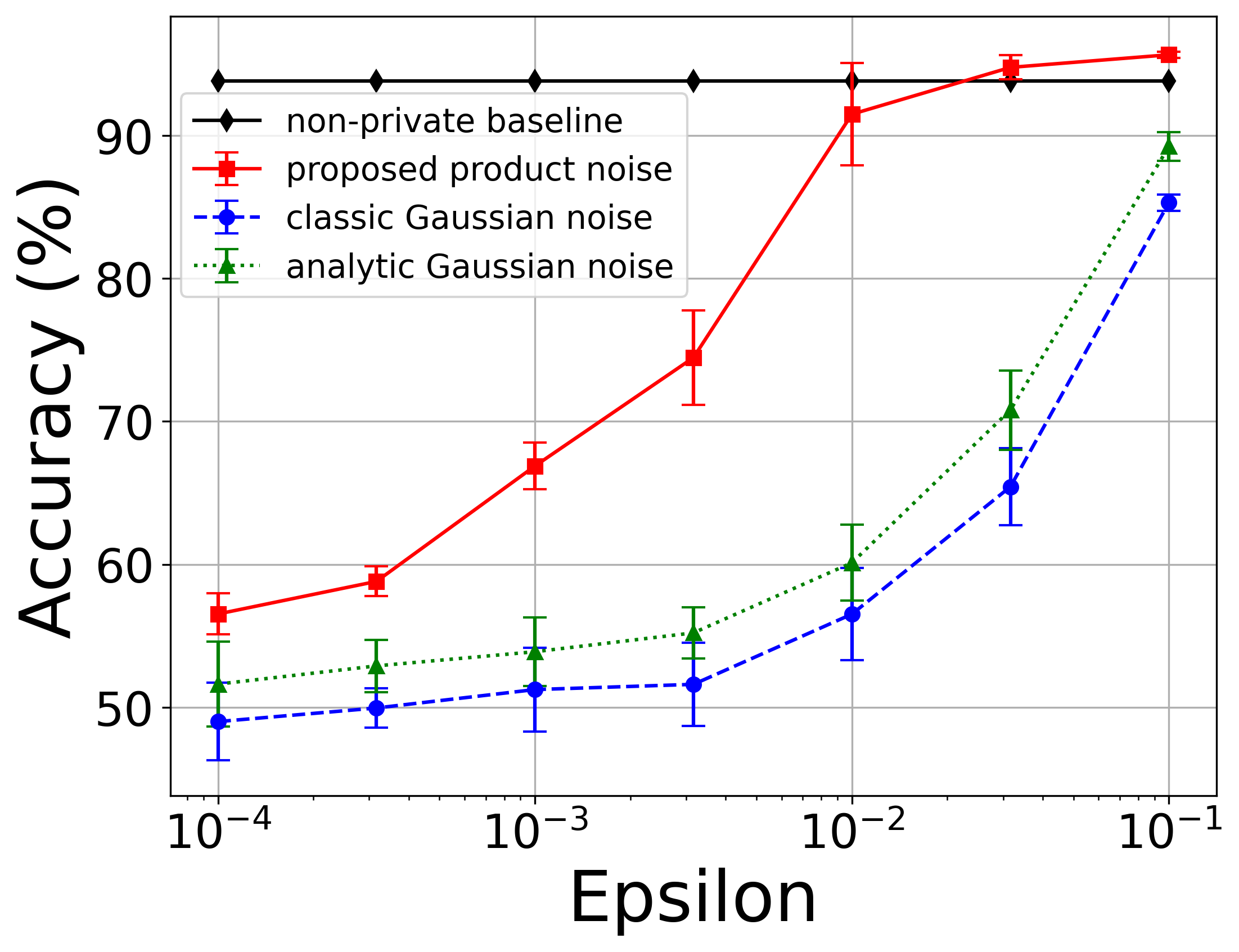}
        \Description{Accuracy results of output perturbation for RCV1 dataset (High-Dim).}
        \caption{\centering RCV1 }
        \label{fig:rcv1_test_accuracy_output_lr}
    \end{subfigure}
    \caption{ Test accuracy of Output Perturbation on LR. 
    Product noise v.s. classic Gaussian noise and analytic Gaussian noise. The experiment setups, e.g., iteration and solver,  follow ~\cite{chaudhuri2011differentially}.}
    \label{fig:output_lr}
\end{figure}

We first evaluate the performance of our product noise and Gaussian noise on  LR tasks~\cite{chaudhuri2011differentially}. As shown in Figure~\ref{fig:output_lr}, product noise-based output perturbation consistently outperforms the classic Gaussian mechanism and analytic Gaussian mechanism for all considered privacy guarantees. 
At higher privacy regimes, our method even approaches or exceeds the performance of non-private baselines, particularly on high-dimensional datasets. On RCV1 with $\epsilon = 10^{-1}$, we  achieve $95.64\%$ test accuracy, surpassing the non-private baseline ($93.79\%$). Even under a stringent privacy setting ($\epsilon = 10^{-\frac{5}{2}}$), our method attains $74.46\%$ test accuracy on RCV1, while the one using classic Gaussian noise and analytic Gaussian noise achieves only $65.43\%$ and $70.80\%$, respectively, under a looser privacy guarantee ($\epsilon = 10^{-\frac{3}{2}}$). These results confirm our finding in Observation~\ref{observation-higher-accuracy}.

In terms of utility stability, product noise-based output perturbation yields lower standard deviations of test accuracy on multiple random experiment trials across all datasets compared to the classic and analytic Gaussian mechanisms, indicating  a more stable utility. For example, on the RCV1 dataset with $\epsilon = 10^{-\frac{3}{2}}$, our method achieves a standard deviation of 0.00845, significantly lower than the classic Gaussian mechanism (0.02711) and analytic Gaussian mechanism (0.02772). This supports our finding in Observation~\ref{observation-higher-stability}. 

Besides test accuracy, we further quantify the  $\ell_2$ error between the private and non-private models. In output perturbation, $\ell_2$ error directly measures the magnitude of the noise added to the optimal parameters. Therefore, a smaller $\ell_2$ error indicates that our private parameters are closer in value to the non-private optimal parameters. Additionally, we also evaluate model robustness using the False Positive Rate (FPR). The experiment results of $\ell_2$ error and FPR are provided in Table~\ref{tab:l2_error_output_lr} and~\ref{tab:fpr_output_lr} in Appendix~\ref{app:supplemental_exp_output}, which clearly demonstrate that our product noise yields not only  the smallest $l_2$ error of the models, but also the lowest FPR. 

\noindent \textbf{More experiments on LR and Huber SVM. }
The experiment results considering other datasets and using multivariate Laplace noise are provided in Appendix~\ref{app:supplemental_exp_output} (see Figures~\ref{fig:output_lr_other_datasets},~\ref{fig:output_svm}, ~\ref{fig:ppsgd_lr}, and~\ref{fig:ppsgd_svm}).
\subsection{Case Study \MakeUppercase{\romannumeral 2}: Objective Perturbation} \label{sec:case-study-objective-perturbation}

This case study evaluates the privacy and utility trade-off achieved by the product noise in objective perturbation. As discussed in the experiment setups, we adopt the same hyperparameter-free (H-F) optimization strategy used in~\cite{iyengar2019towards}, which sets the Gaussian noise scale $\sigma$ according to ~\cite[Lemma 4]{nikolov2013geometry} and can accommodate all $\epsilon>0$. Additionally, we also consider the analytic Gaussian noise~\cite{balle2018improving} in  H-F-based AMP.
For fair comparison, we use the same optimization parameters (e.g., Lipschitz constant and stopping criterion) provided  in~\cite{iyengar2019towards}. Privacy parameter $\epsilon$ is selected from $\{ 10^{-2}, 10^{-\frac{3}{2}}, 10^{-1},10^{-\frac{1}{2}}, 10^{0}, 10^{\frac{1}{2}},  10^{1} \}$ and $\delta=\frac{1}{n^2}$, where $n$ is the number of training samples.
\begin{figure}[htp]
    \centering
    \begin{subfigure}{0.49\columnwidth}  
        \centering
        \includegraphics[width=\linewidth]{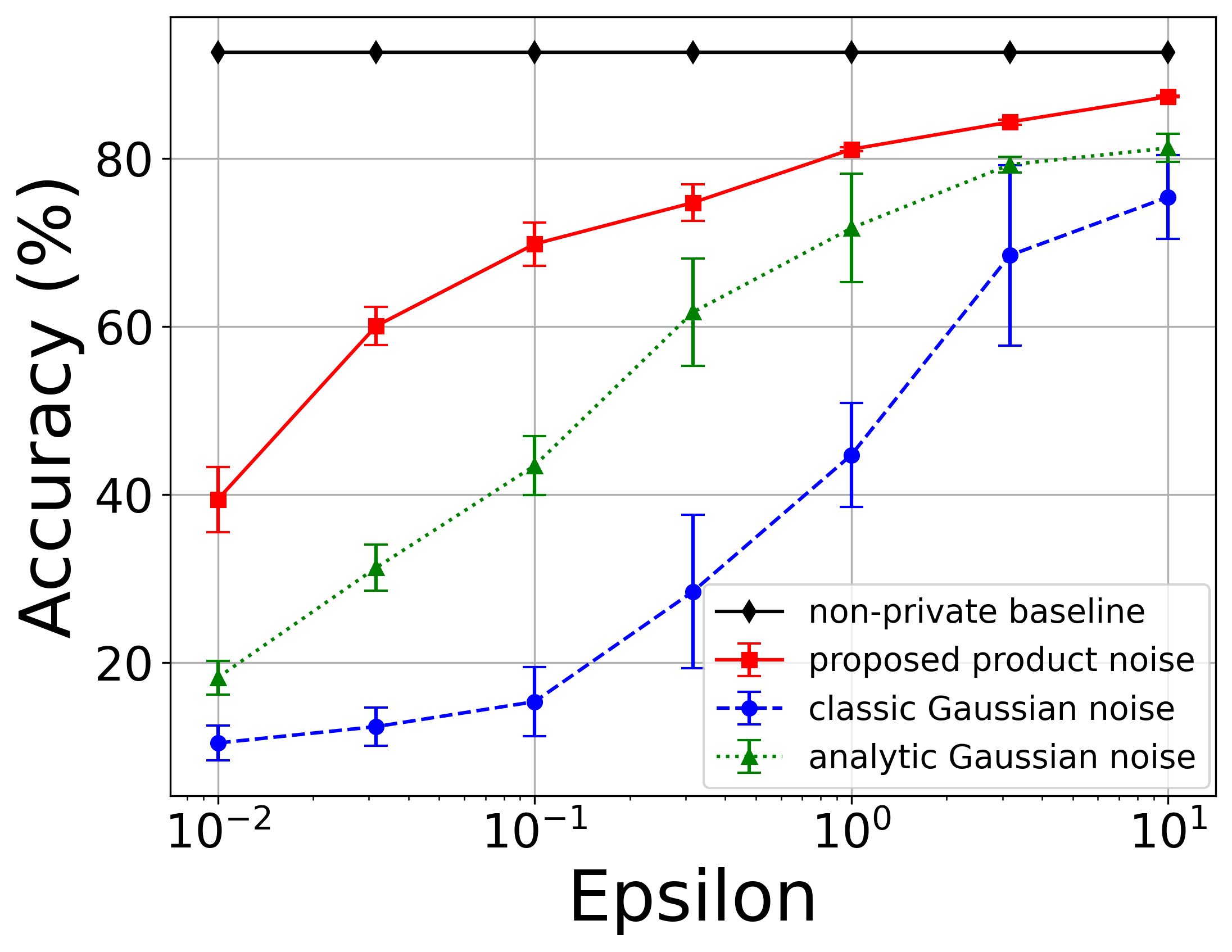}
        \Description{Graph showing accuracy versus epsilon for MNIST dataset (Low-Dim) under objective perturbation for LR.}
        \caption{\centering  MNIST }
        \label{fig:mnist_test_accuracy_objective_lr}
    \end{subfigure}
    \hfill
    \begin{subfigure}{0.49\columnwidth}
        \centering
        \includegraphics[width=\linewidth]{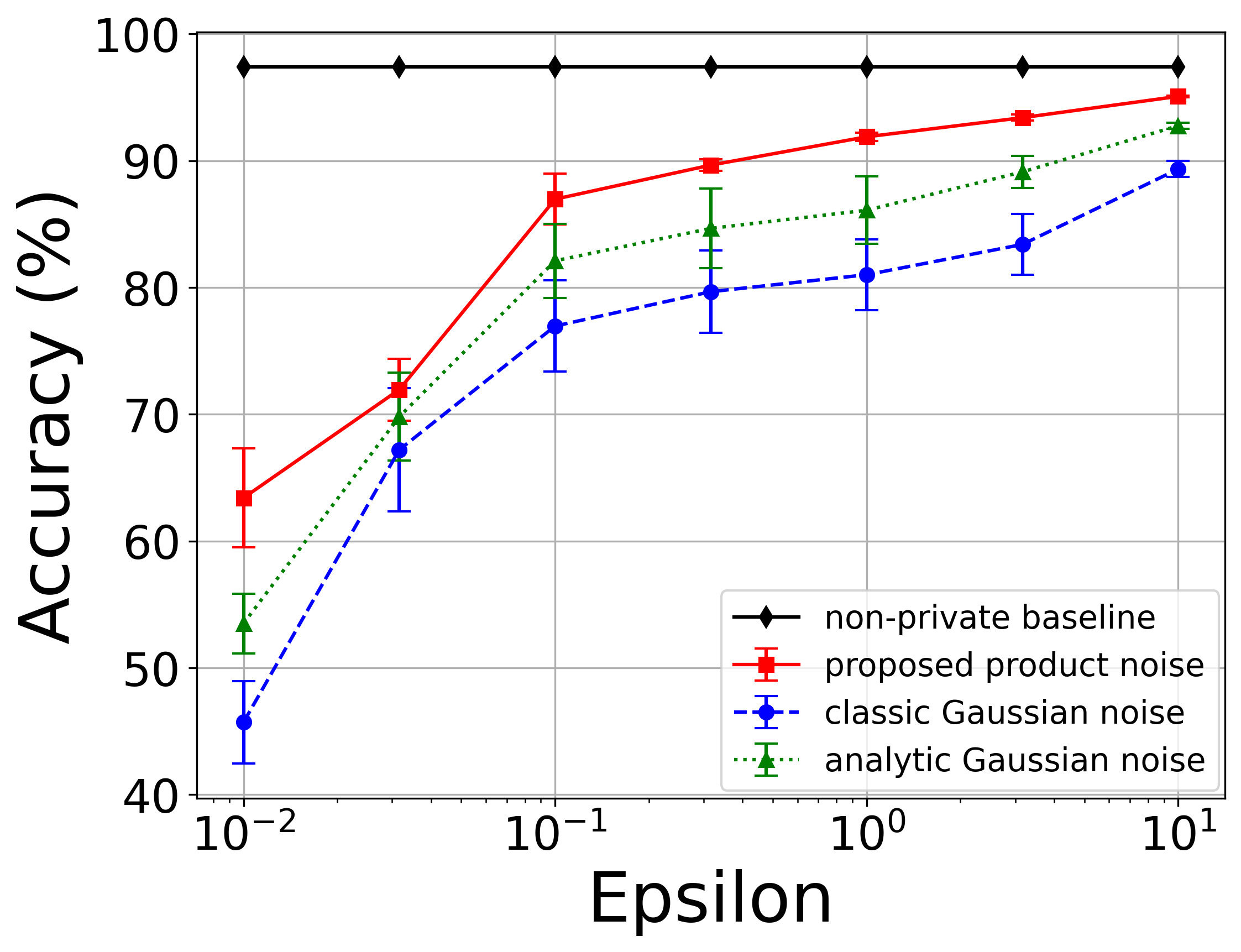}
        \Description{Graph showing accuracy versus epsilon for Synthetic-H dataset (High-Dim) under objective perturbation for LR.}
        \caption{\centering  Synthetic-H }
        \label{fig:syntheticH_test_accuracy_objective_lr}
    \end{subfigure}
    \hfill
    \begin{subfigure}{0.49\columnwidth}
        \centering
        \includegraphics[width=\linewidth]{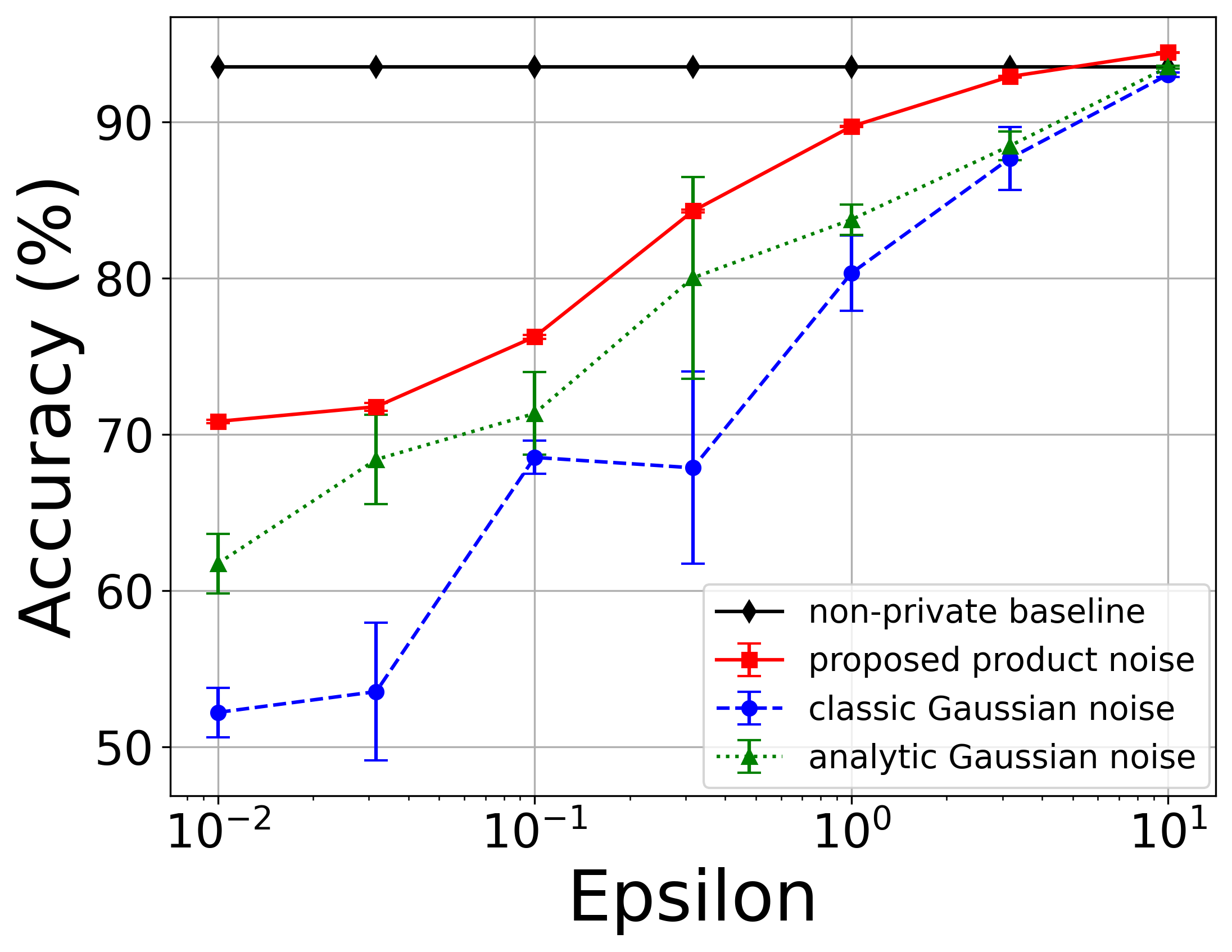}
        \Description{Graph showing accuracy versus epsilon for Real-sim dataset (High-Dim) under objective perturbation for LR.}
        \caption{\centering  Real-sim}
        \label{fig:realsim_test_accuracy_objective_lr}
    \end{subfigure}
    \hfill
    \begin{subfigure}{0.49\columnwidth}
        \centering
        \includegraphics[width=\linewidth]{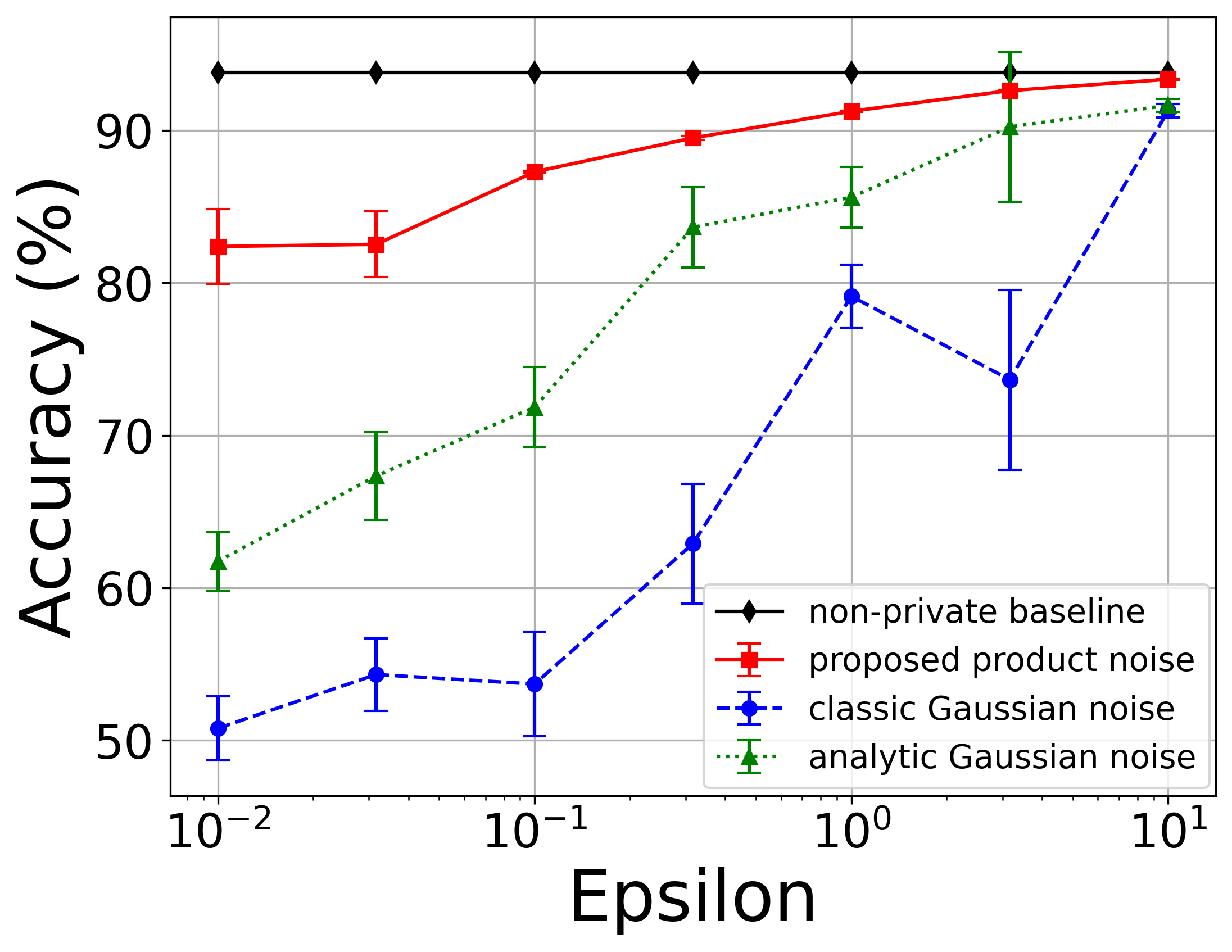}
        \Description{Graph showing accuracy versus epsilon for RCV1 dataset (High-Dim) under objective perturbation for LR.}
        \caption{\centering  RCV1 }
        \label{fig:rcv1_test_accuracy_objective_lr}
    \end{subfigure}
    \caption{ Test accuracy of Objective Perturbation on LR.}  
    \label{fig:objective_lr}
\end{figure}\

The experiment results on LR are  shown in Figure \ref{fig:objective_lr}. It demonstrates that product noise-based AMP consistently outperforms those using Gaussian noises on all datasets under the same privacy parameters. can even surpass the non-private baselines. For example, on the Real-sim dataset with $\epsilon = 10$, our method achieves test accuracy of $94.45\%$, exceeding the classic Gaussian noise ($93.03\%$), analytic Gaussian noise ($93.51\%$) and the non-private baseline ($93.52\%$). 

Moreover, our method retains strong utility even under stringent privacy constraints, especially on high-dimensional datasets. With $\epsilon=10^{-1}$, product noise-based AMP achieves $86.96\%$ and $87.29\%$ accuracy on Synthetic-H and RCV1, respectively, surpassing both classic and analytic Gaussian noise-based H-F AMP even when $\epsilon=1$. These results further support our claim in Observation~\ref{observation-higher-accuracy}.

Regarding utility stability, product noise-based AMP consistently achieves lower standard deviation of test accuracy than Gaussian noise-based H-F AMP across most datasets, indicating greater stability for objective perturbation with product noise. For instance, on RCV1 at $\epsilon=10^{\frac{1}{2}}$, the standard deviation is $0.0002$ for product-noise based AMP, compared to $0.05897$ for Gaussian noise-based H-F AMP and $0.04895$ for analytic Gaussian noise-based H-F AMP. This further substantiates Observation~\ref{observation-higher-stability}.
\begin{table}[htp]
  \centering
  \caption{ $\ell_2$ error of Objective Perturbation on LR.}
  \label{tab:l2_error_obj_LR}
  \footnotesize
  \setlength{\tabcolsep}{10pt}
  \renewcommand{\arraystretch}{1.0}
  \resizebox{\columnwidth}{!}{%
  \begin{tabular}{c|c|c|c|c|c}
    \toprule
    \multirow{2}{*}{\textbf{Dataset}} &
    \multirow{2}{*}{\textbf{Mechanism}} &
    \multicolumn{4}{c}{$\boldsymbol{\epsilon}$} \\
    \cline{3-6}
     &  & {\cellvcenter $\mathbf{10^{-2}}$}
        & {\cellvcenter $\mathbf{10^{-1}}$}
        & {\cellvcenter $\mathbf{10^{0}}$}
        & {\cellvcenter $\mathbf{10^{1}}$} \\
    \midrule
    \multirow[c]{3}{*}{MNIST}
      & classic  & 92.1 & 44.6 & 44.2 & 59.5 \\
      & analytic & 22.4 & 24.1 & 26.00 & 27.5 \\
      & ours     & \textbf{0.2} & \textbf{6.0} & \textbf{1.3} & \textbf{4.5} \\
    \midrule
    \multirow[c]{3}{*}{Synthetic-H}
      & classic  & 70.5 & 63.6 & 70.4 & 49.3 \\
      & analytic & 32.8 & 35.8 & 41.2 & 46.7 \\
      & ours     & \textbf{3.0} & \textbf{1.6} & \textbf{5.6} & \textbf{1.2} \\
    \midrule
    \multirow[c]{3}{*}{ Real-sim}
      & classic  & 79.2 & 63.5 & 77.9 & 119.1 \\
      & analytic & 36.6 & 39.8 & 42.9 & 49.3 \\
      & ours     & \textbf{4.9} & \textbf{0.7} & \textbf{1.3} & \textbf{1.2} \\
    \midrule
    \multirow[c]{3}{*}{RCV1}
      & classic  & 102.6 & 141.9 & 102.8 & 85.7 \\
      & analytic & 53.4 & 58.3 & 63.3 & 73.6 \\
      & ours     & \textbf{2.6} & \textbf{0.1} & \textbf{0.9} & \textbf{2.3} \\
    \bottomrule
  \end{tabular}}
\end{table}

\begin{table}[htp]
  \centering
  \caption{FPR of Objective Perturbation on LR.}
  \label{tab:fpr_obj_lr}
  \footnotesize
  \setlength{\tabcolsep}{2pt}
  \renewcommand{\arraystretch}{1.0}
  \resizebox{\columnwidth}{!}{%
  \begin{tabular}{c|c|c|c|c|c}
    \toprule
    \multirow{2}{*}{\textbf{Dataset}} &
    \multirow{2}{*}{\textbf{Mechanism}} &
    \multicolumn{3}{c}{\textbf{$\boldsymbol{\epsilon}$}} \\ 
    \cline{3-6}
     &  & {\cellvcenter $\mathbf{10^{-2}}$}
        & {\cellvcenter $\mathbf{10^{-1}}$}
        & {\cellvcenter $\mathbf{10^{0}}$}
        & {\cellvcenter $\mathbf{10^{1}}$} \\
    \midrule
    \multirow[c]{3}{*}{MNIST} 
      & classic  & 0.099$\pm$0.016 & 0.084$\pm$0.015 & 0.041$\pm$0.008 & 0.022$\pm$0.004 \\
      & analytic & 0.099$\pm$0.014 & 0.081$\pm$0.015 & 0.027$\pm$0.004 & 0.016$\pm$0.002 \\
      & ours     & \textbf{0.031$\pm$0.004} & \textbf{0.036$\pm$0.006} & \textbf{0.022$\pm$0.002} & \textbf{0.015$\pm$0.003} \\
    \midrule
    \multirow[c]{3}{*}{Synthetic-H} 
      & classic  & 0.492$\pm$0.046 & 0.420$\pm$0.015 & 0.207$\pm$0.008 & 0.120$\pm$0.004 \\
      & analytic & 0.228$\pm$0.035 & 0.170$\pm$0.023 & 0.103$\pm$0.013 & 0.073$\pm$0.007 \\
      & ours     & \textbf{0.143$\pm$0.011} & \textbf{0.077$\pm$0.009} & \textbf{0.049$\pm$0.005} & \textbf{0.029$\pm$0.002} \\
    \midrule
    \multirow[c]{3}{*}{Real-sim} 
      & classic  & 0.463$\pm$0.075 & 0.199$\pm$0.022 & 0.150$\pm$0.024 & 0.069$\pm$0.010 \\
      & analytic & 0.440$\pm$0.068 & 0.292$\pm$0.036 & 0.205$\pm$0.025 & 0.105$\pm$0.002 \\
      & ours     & \textbf{0.444$\pm$0.074} & \textbf{0.091$\pm$0.030} & \textbf{0.020$\pm$0.009} & \textbf{0.007$\pm$0.000} \\
    \midrule
    \multirow[c]{3}{*}{RCV1} 
      & classic  & 0.657$\pm$0.094 & 0.477$\pm$0.082 & 0.137$\pm$0.022 & 0.065$\pm$0.008 \\
      & analytic & 0.322$\pm$0.057 & 0.235$\pm$0.031 & 0.075$\pm$0.013 & 0.064$\pm$0.012 \\
      & ours     & \textbf{0.183$\pm$0.036} & \textbf{0.027$\pm$0.004} & \textbf{0.048$\pm$0.010} & \textbf{0.057$\pm$0.006} \\
    \bottomrule
  \end{tabular}}
\end{table}

\begin{figure*}[htp]
    \centering
    \begin{subfigure}{0.33\textwidth}
        \centering
        \includegraphics[width=\linewidth]{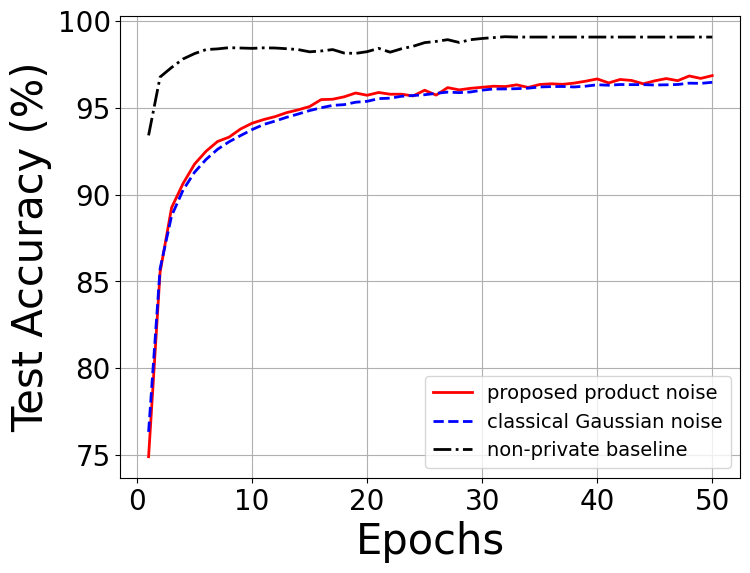}
        \Description{Graph showing test accuracy results on the MNIST dataset using DPSGD.}
        \caption{\centering Test Accuracy vs Epochs}
        \label{fig:MNIST_test_accuracy}
    \end{subfigure}
    \hfill
    \begin{subfigure}{0.33\textwidth}
        \centering
        \includegraphics[width=\linewidth]{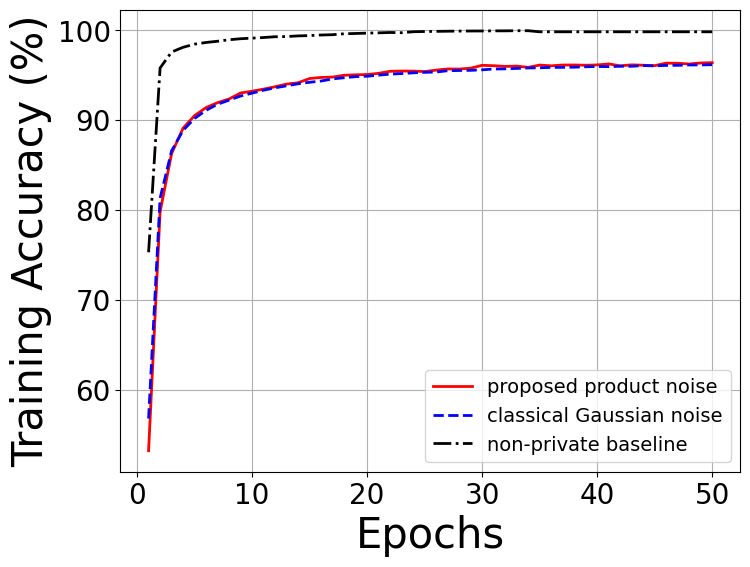}
        \Description{Graph showing Training accuracy results on the MNIST dataset using DPSGD.}
        \caption{\centering Training Accuracy vs Epochs}
        \label{fig:MNIST_Training_accuracy}
    \end{subfigure}
    \hfill
    \begin{subfigure}{0.33\textwidth}
        \centering
        \includegraphics[width=\linewidth]{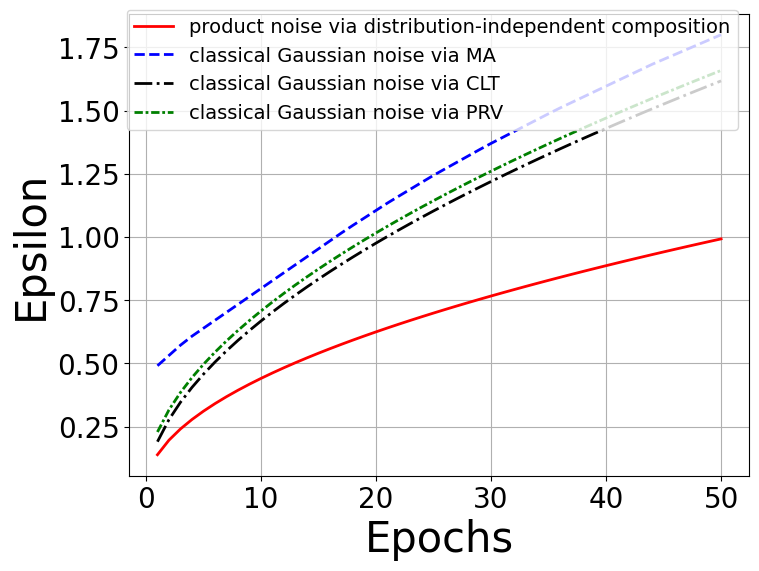}
        \Description{Graph comparing epsilon values on the MNIST dataset under DPSGD.}
        \caption{\centering Epsilon vs Epochs}
        \label{fig:MNIST_epsilon}
    \end{subfigure}
    \caption{DPSGD results on the MNIST dataset: Test accuracy, training accuracy, and privacy parameter ($\epsilon$) over epochs.}
    \label{fig:MNIST_dpsgd}
\end{figure*}

\begin{figure*}[htp]
    \centering
    \begin{subfigure}{0.33\textwidth}  
        \centering
        \includegraphics[width=\linewidth]{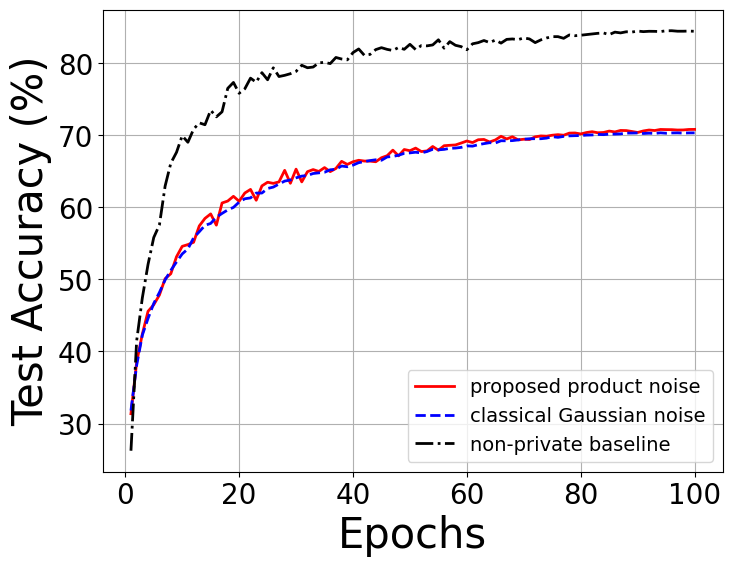}
        \Description{Graph showing test accuracy results on the CIFAR-10 dataset using DPSGD.}
        \caption{Test Accuracy vs Epochs}
        \label{fig:CIFAR-10_test_accuracy}
    \end{subfigure}
    \hfill
    \begin{subfigure}{0.33\textwidth}
        \centering
        \includegraphics[width=\linewidth]{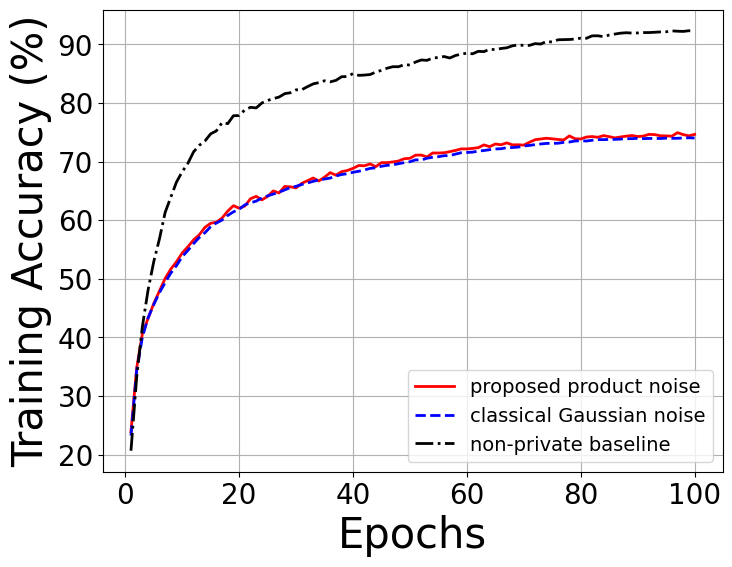}
        \caption{Training Accuracy vs Epochs}
        \Description{Training accuracy curve of the CIFAR-10 dataset with DPSGD.}
        \label{fig:CIFAR-10_Training_accuracy}
    \end{subfigure}
        \hfill
      \begin{subfigure}{0.33\textwidth}
        \centering
        \includegraphics[width=\linewidth]{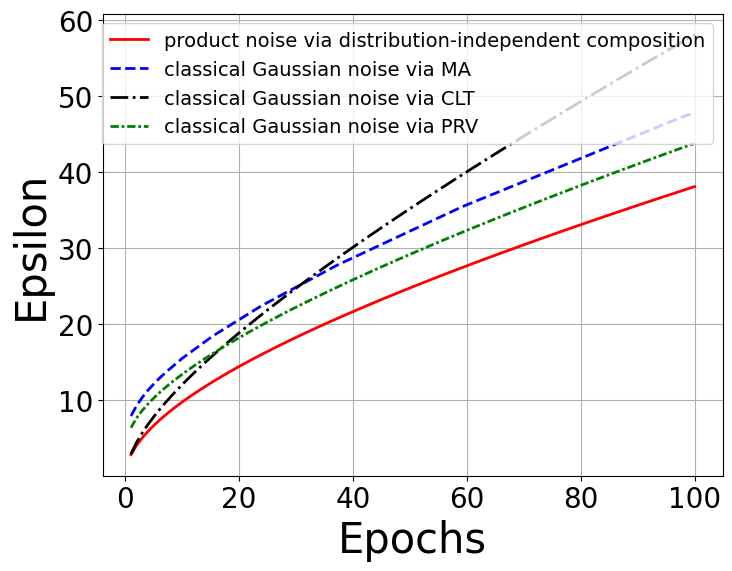}
        \Description{Graph comparing epsilon values on the CIFAR-10 dataset under DPSGD.}
        \caption{Epsilon vs Epochs}
        \label{fig:CIFAR-10_epsilon}
    \end{subfigure}
    \caption{ DPSGD results on the CIFAR-10 dataset: Test accuracy,  training accuracy, and privacy parameter ($\epsilon$) over epochs.}
    \label{fig:CIFAR-10_dpsgd}
\end{figure*}

In Table~\ref{tab:l2_error_obj_LR}, we also  present the $l_2$ error between the private and non-private models for this case study. The product noise-based Objective Perturbation consistently achieves significantly lower $\ell_2$ error than both the classic and analytical Gaussian noises across all privacy parameters. Again, our method effectively mitigates   learning model  distortion, ensuring that the obtained ERM models stay  close to the non-privacy ones.

In Table~\ref{tab:fpr_obj_lr}, we present the FPR obtained by various models given different $\epsilon$. Again, our product noise achieves the lowest FPRs among all mechanisms. These results confirm that product noise maintains stable decision boundaries, reduces misclassifications, and at the same time, preserves strong DP guarantees.
Combining Figure~\ref{fig:objective_lr} and  Table~\ref{tab:fpr_obj_lr}, we observe that the increased accuracy of our method is accompanied by lower FPR, suggesting that the observed utility gains in certain experiments (e.g., Figure~\ref{fig:objective_lr} (c)) can be attributed to improved generalization.

\begin{remark}
    Notably, although differential privacy is traditionally perceived to incur performance degradation, our experiments reveal a contrasting observation, i.e.,  in specific datasets (e.g., Real-sim in Figure~\ref{fig:objective_lr} (c) and  Synthetic-H in Figure~\ref{fig:objective_svm} (b)), the product noise-based private models avoid accuracy loss and even outperform the non-private baselines in test accuracy. 
    This occurs because the non-private baseline may overfit the training data, particularly when the number of samples is of the same order as the data dimension (e.g., the Synthetic-H, Real-sim, and RCV1 datasets).  The addition of DP noise through output or objective perturbation acts as a form of regularization, which suppresses overfitting and thereby improves generalization performance on unseen test data. In general,  when the number of features (dimensions) is comparable to or exceeds the number of training samples, the model has sufficient flexibility to “memorize” the training data, leading to overfitting. In such cases, DP noise can act as an effective regularizer to mitigate overfitting. This phenomenon has also been observed in prior works, e.g., ~\cite{bassily2014private,dwork2015generalization,iyengar2019towards}. 
\end{remark}

The experiments   on Huber SVM  using objective perturbation are similar to the LR tasks, i.e., product noise-based AMP can achieve higher test accuracy and utility stability. The detailed plots and tables  are shown in  Appendix~\ref{app:svm-exp-objective} (see Figure~\ref{fig:objective_svm}, Table~\ref{tab:l2_error_obj_SVM} and~\ref{tab:fpr_obj_svm}). 

\subsection{Case Study \MakeUppercase{\romannumeral 3}: Gradient Perturbation}\label{sec:case-study-gradient-perturbation}
In this case study, we first corroborate that, with less privacy leakage (measured in terms of cumulative $\epsilon$), our proposed product-noise-based-DPSGD can achieve   utility that is comparable with classic DPSGD~\cite{abadi2016deep}. Then, we show that we can provide the same level of robustness against membership inference attacks~\cite{wei2023dpmlbench}. 

\subsubsection{Comparable Utility, Yet Higher Privacy} \label{sec: acc_dpsgd}
First, we compare our product noise method with the classic Gaussian noise under identical training settings. More specifically, we aim to compare the privacy guarantees offered by each method when achieving comparable utility. We compute the cumulative privacy loss using the distribution-independent composition (Theorem~\ref{thm: composition_theorem_1}), and compare the results against those obtained using the MA~\cite{abadi2016deep}, PRV~\cite{gopi2021numerical}, and CLT~\cite{bu2020deep} composition approaches.

The experiments are conducted on 5 datasets: Adult,  MovieLens, IMDb, MNIST, and CIFAR-10. 
Detailed setups for various datasets, including batch size, learning rate, gradient clipping threshold, noise multiplier (only used by classic DPSGD), initial $\epsilon$, and tuning parameter $k$ (initial $\epsilon$ and $k$ only used by our product noise-based DPSGD)  are summarized in Table~\ref{tab:dpsgd_parameters} (Appendix~\ref{app:non-convex-setup}). Each experiment is independently repeated 5 times.

\noindent \textbf{MNIST.}
This dataset~\cite{lecun1998gradient} contains $60,000$ training images and $10,000$ test images. We construct a convolutional neural network with the same architecture as in \cite{bu2020deep}. To ensure that different methods achieve comparable utility, we select the key training parameters as follows: the learning rate $\eta_t$ is set to 0.15. For the Gaussian noise-based DPSGD, we set the noise multiplier $\sigma = 1.3$, and the gradient clipping norm is fixed at $C = 1.0$. For our product noise-based DPSGD, we set the initial privacy parameter $\epsilon=0.3$ and the tuning parameter $k=40,000$.

As shown in Figure~\ref{fig:MNIST_dpsgd} (a) and (b), our product noise-based DPSGD and the Gaussian noise-based DPSGD achieve comparable utility. However, our product noise-based DPSGD   offers significantly stronger privacy preservation. As shown in Figure~\ref{fig:MNIST_dpsgd} (c), for 50 epochs, our method achieves $(0.99, 9.87 \times 10^{-6})$-DP (evaluated using privacy amplification followed by distribution-independent composition). This is notably smaller than the privacy guarantees obtained by the Gaussian noise-based DPSGD under MA ($(1.80, 10^{-5})$-DP), PRV ($(1.66, 10^{-5})$-DP), and CLT ($(1.62, 10^{-5})$-DP) composition approaches. Furthermore, the smaller $\delta$ value indicates a lower failure probability for the DP guarantee. These results support our argument in Observation~\ref{observation-stronger-privacy} that our method can achieve tighter privacy guarantees under comparable utility.

\noindent \textbf{CIFAR-10. }
This dataset \cite{cifar10} consists of images belonging to  10 classes, with $50,000$ training and $10,000$ test examples. We also construct a convolutional neural network with the same architecture  as in \cite{bu2020deep}. To ensure that different methods achieve comparable utility, we select the key training parameters as follows: the learning rate $\eta_t$ is set to 0.25. For the Gaussian noise-based DPSGD, we set the noise multiplier $\sigma = 0.50$, and the gradient clipping norm is fixed at $C = 1.5$. For our product noise-based DPSGD, we set the initial privacy parameter $\epsilon=1.8$ and the tuning parameter $k=300,000$.

As shown in Figure~\ref{fig:CIFAR-10_dpsgd} (a) and (b), our product noise-based DPSGD exhibits utility performance comparable to that of the Gaussian noise-based DPSGD throughout the training process, i.e., both methods also show a similar trend regarding training and testing accuracy. However, under a comparable utility, our product noise method demonstrates a notable advantage regarding privacy preservation. As illustrated in Figure~\ref{fig:CIFAR-10_dpsgd} (c), by the 100th epoch, our method achieves a privacy guarantee of $(38.11,9.94 \times 10^{-6})$-DP, which is substantially better than the guarantees obtained by the classical DPSGD under MA ($(47.83,10^{-5})$-DP), PRV ($(43.80,10^{-5})$-DP), and CLT ($(58.09,10^{-5})$-DP) composition approaches. This result further supports the claim made in Observation~\ref{observation-stronger-privacy}. 

\begin{remark}
    Our theoretical analysis (e.g., Theorem~\ref{thm:mian-thm} and Corollary~\ref{corollary:mechanism_guidance} and ~\ref{corollary:Asymptotic_analysis_noise} ) focuses on the single-release setting, where the noisy result is only released once (as in Case Studies \MakeUppercase{\romannumeral 1} and \MakeUppercase{\romannumeral 2}). In terms of iterative settings (as in Case Study \MakeUppercase{\romannumeral 3}) where the noisy results are repeatedly released, the final privacy cost depends on the considered privacy amplification and composition methods.  
    
    In our current DPSGD experiments, the privacy loss of the baseline methods are accounted using MA/PRV/CLT approaches, which are unfortunately infeasible for our product noise. Since we have not yet developed a composition approach tailored to our noise, we use a distribution-independent, albeit loose, composition approach, which yields a pessimistic cumulative $\epsilon$ in iterative settings. This accounts for the modest privacy improvement observed in some experiments. 
\end{remark}

The experiment results on other datasets are shown in    Appendix~\ref{app:dpsgd-more-experiments}, i.e., cf. Figure~\ref{fig:adult_dpsgd} (for adult), Figure~\ref{fig:IMDb_dpsgd} (for IMDB) and Figure~\ref{fig:MovieLens_dpsgd} (for MovieLens). Similarly, DPSGD using our product noise also yields smaller privacy parameters under  comparable utility.  

\subsubsection{Robustness against Membership Inference Attacks (MIAs)} 
Now, we evaluate the robustness against  MIAs in a black-box setting by comparing models trained with our product noise against those trained with the classic Gaussian noise. Following the experiment setups in a prior work~\cite{wei2023dpmlbench}, we use the tailored area under the ROC curve (AUC) as the robustness metric. In particular,  \cite{wei2023dpmlbench} defines  $\widetilde{\mathrm{AUC}}\triangleq\max\{\mathrm{AUC},0.5\}$, and the closer it is to 0.5, the more robust the model is against MIAs. Still, for the 5 datasets considered in Table~\ref{tab:method_overview}, we train each corresponding model  for 50 epochs, and  repeat 5 times. All other settings follow Subsection~\ref{sec: acc_dpsgd}.
\begin{table}[htp]
  \centering
  \small
  \setlength{\tabcolsep}{6pt}
  \renewcommand{\arraystretch}{1.0}
  \caption{ Tailored AUC ($\widetilde{\mathrm{AUC}}$) in black-box MIAs~\cite{wei2023dpmlbench}.}
  \label{tab:noise_results_all_transposed}
  \begin{tabular}{|c|c|c|}
    \hline
    \textbf{Dataset} & \textbf{Classic Gaussian noise} & \textbf{Product noise} \\
    \hline
    Adult     & 0.51$\pm$0.004 & \textbf{0.50}$\pm$0.004 \\
    \hline
    IMDb      & 0.51$\pm$0.005 & \textbf{0.50}$\pm$\textbf{0.000} \\
    \hline
    MovieLens & 0.50$\pm$0.000 & 0.50$\pm$0.000 \\
    \hline
    MNIST     & 0.50$\pm$0.000  & 0.50$\pm$0.000  \\
    \hline
    CIFAR-10  & 0.51$\pm$0.005 & 0.51$\pm$\textbf{0.004} \\
    \hline
  \end{tabular}
  \label{tab:mia}
\end{table}

The results of tailored AUC ($\widetilde{\mathrm{AUC}}$) are summarized in Table~\ref{tab:mia}, which shows that the product noise achieves a level of resistance to MIAs comparable to that of the classic Gaussian noise; the obtained $\widetilde{\mathrm{AUC}}$ are  all around $0.5$, which indicates that  the black-box MIAs perform no better than random guessing the presence of a single record in the training dataset. The experiment results verify that when considering real-world privacy threats, even with a smaller noise scale, our product noise does not weaken privacy compared to the Gaussian baseline in DPSGD. 

\section{Conclusion}\label{sec:Conclusion}

In this work, we propose a novel spherically symmetric noise to reduce utility loss and improve the accuracy of differentially private query results. The new noise leads to a new PLRV that can be presented as the product between two random variables, providing tighter measure concentration analysis and leading to small tail bound probability. In contrast, existing classic Gaussian mechanism and its variants usually characterize their  PLRVs using   Gaussian distributions. Simulation results show that when perturbing    $f(\bm{x})\in\R^M$ in high dimension, our proposed noise has expected squared magnitude far smaller than that required by the Gaussian mechanism and its variant.  
To validate the effectiveness of the developed noise, we apply it to privacy-preserving convex and non-convex  ERM. Experiments on multiple datasets demonstrate substantial utility gains in diverse convex ERM models and notable privacy improvements in non-convex ERM models.

\begin{acks}
The work of Shuainan Liu, Tianxi Ji, and Zhongshuo Fang was supported in part by the U.S. Department of Agriculture under Grant AP25VSSP0000C026.
\end{acks}


\appendix

\section{APPENDIX}\label{sec:APPENDIX}

\subsection{Additional Preliminaries on Probability and Special Function}\label{app:additional-app}
\begin{lemma} 
[Transformation of Random Variables{~\cite[p. 51]{casella2002statistical}}] Let $X$ have $\textit{PDF}$ $f_{X}(x)$ and let $Y=g(X)$, where $g$ is a monotone function. $\mathcal{X}=\left\{x: f_{X}(x)>0 \right\}$ and $\mathcal{Y}=\left\{y: y=g(x)\ \text{for some}\ x \in \mathcal{X} \right\}$. Suppose that $f_{X}(x)$ is continuous on $\mathcal{X}$ and that $g^{-1}(y)$ has a continuous derivative on $\mathcal{Y}$. Then the probability density function of $Y$ is given by
\begin{equation*}
\begin{aligned}
    f_Y(y) = 
\begin{cases} 
f_X(g^{-1}(y)) \left| \frac{d}{dy} g^{-1}(y) \right|  & y \in \mathcal{Y} \\
0  & \text{otherwise}
\end{cases}.
\end{aligned}
\end{equation*}
\label{thm:variable-trans}
\end{lemma}

\begin{lemma}[Properties of Gamma functions]
\label{lemma:gamma-property}
\begin{align} 
     &  \sqrt{n+\frac{1}{4}}< \frac{\Gamma \left ( n+ 1 \right )}{\Gamma \left ( n+\frac{1}{2} \right )} <\sqrt{n+\frac{1}{2}}, \quad {~\cite[p. 425]{mortici2010new} \label{eq:ratio-gamma-1}}\\
    &\Gamma(2z)= \pi^{-\frac{1}{2}}2^{2z-1}\Gamma(z)\Gamma(z+\frac{1}{2}), \forall 2z \ne 0, -1, -2, \cdots,  \quad {~\cite[p. 138]{abramowitz1948handbook}\label{eq:duplication-formula}}\\
    &  \Gamma(x)<\frac{x^{x-\frac{1}{2}}}{e^{x-1}}, \quad x \in (1, \infty), \quad {~\cite[p. 258]{joshi1996inequalities}} \label{eq: gamma-1}
\end{align}
\end{lemma}

\begin{lemma}[Properties of confluent hypergeometric function] 
    \begin{equation*} \label{eq: hypergeometric-ineq}
 {}_1F_1(a; c; x)< 1+\frac{2xa}{c}, \quad c>0, a>0, 0<x<1, \qquad {~\cite[p. 258]{johnson1995continuous}}
\end{equation*}
      \label{thm: hypergeometric-ineq}
\end{lemma}

\subsection{Additional preliminaries on DP}\label{app:dp} 
\begin{theorem}[Privacy Amplification via Poisson Sampling~\cite{balle2018privacy}]
\label{thm:Poisson-Subsampling-ampilification}
Suppose a randomized mechanism $\mathcal{M}$ satisfies $(\epsilon,\delta)$-DP. If the dataset is pre-processed using Poisson subsampling with sampling probability $p$, then the subsampled mechanism satisfies: 
\begin{equation*}
    \begin{aligned}
        \left( \log \left(1 + p(e^{\epsilon} - 1)\right),\ p\delta \right)-DP.
    \end{aligned}
\end{equation*}
\end{theorem}

\begin{theorem} [Distribution-Independent Composition~\cite{he2021tighter}] \label{thm: composition_theorem_1}
Suppose an iterative algorithm $\mathcal{A}$ has $T$ steps: $\{W_i (s)\}_{i=0}^{\top}$, where $W_i$ is the learned hypothesis after the $i$-th iteration. Suppose all the iterators 
are ($\epsilon, \delta$)-differentially private. Then, the algorithm $\mathcal{A}$ is ($\epsilon^{\prime}, \delta^{\prime}$)-differentially private that
\small{
\begin{equation*}
        \begin{aligned}
            & \epsilon^{\prime} = T \frac{(e^{\epsilon}-1)\epsilon}{e^{\epsilon} +1 } + \sqrt{2 \log \left( \frac{1}{\tilde{\delta}}  \right)  T \epsilon^2}, \\
            & \delta^{\prime} = 2 - \left( 1-e^\epsilon \frac{\delta}{1 + e^\epsilon} \right)^{\left \lceil \frac{\epsilon^{\prime}}{\epsilon} \right \rceil } \left ( 1- \frac{\delta}{1 + e^{\epsilon}} \right)^{T - \left \lceil \frac{\epsilon^{\prime}}{\epsilon} \right \rceil } - \left( 1 - \frac{\delta}{1 + e^\epsilon} \right)^{T} + \delta^{\prime \prime},
        \end{aligned}
    \end{equation*}
}\ignorespacesafterend
where $\tilde{\delta}$ is an arbitrary positive real constant, and $\delta^{\prime \prime}$ is defined as 
\small{
 \begin{equation*}
        \begin{aligned}
            \delta^{\prime \prime} = e^{- \frac{\epsilon^{\prime} + T\epsilon}{2}} \left ( \frac{1}{1 + e^\epsilon} \left( \frac{2T\epsilon}{T\epsilon - \epsilon^{\prime}}  \right) \right)^{\top} \left ( \frac{T\epsilon + \epsilon^{\prime}}{T\epsilon-\epsilon^{\prime}} \right)^{-\frac{\epsilon^{\prime} +T\epsilon}{2\epsilon}}.
        \end{aligned}
    \end{equation*}
}\ignorespacesafterend
\end{theorem}

\subsection{Proof of Proposition~\ref{prop:proof-pdf-sin}} \label{app:proof-pdf-sin}
\begin{proof}
Proposition~\ref{prop:proof-pdf-sin} can be proved by using transformation of random variables. First, we recall   the following Lemma  from \cite[p. 1851]{cai2013distributions}.

\begin{lemma}\label{lemma:cos}
If $X_{1},X_{2},...$ are random points sampled from $\mathbb{S}^{M-1}$ (the unit sphere embedded in in $\mathbb{R}^{M}$) uniformly at random, $M>2$. Let  $\Theta_{ij}$ be the angle between $\overrightarrow{OX_{i}}$ and $\overrightarrow{OX_{j}}$ and $x_{ij}=cos\Theta_{ij}$ for any $i \neq j$. Then $\left\{ x_{ij}; 1\leq i  < j \leq n \right\}$ are pairwise independent and identically distributed with density function
\begin{equation*}
f_{X}(x) = \frac{{1}}{\sqrt \pi} \frac{\Gamma\left( \frac{M}{2} \right)}{\Gamma\left( \frac{M-1}{2} \right)} (1 - x^2)^{\frac{M-3}{2}} \quad \text{for} \quad |x| < 1, 
\end{equation*} 
where $x$ denotes $\cos \theta$.
\end{lemma}

Lemma~\ref{lemma:cos} applies to two random unit vectors and also holds when one vector is fixed while the other is randomly sampled from the unit sphere~\cite{cai2013distributions,cai2012phase}. In our study, given an arbitrary  pair of $f(\bm{x})$ and $f(\bm{x}')$, it determines a $\bm{v}\triangleq f(\bm{x})-f(\bm{x}')\in\R^M$, while the direction of the 
perturbation noise $\mathbf{n}$ is randomly sampled from the unit sphere (i.e., $\bm{h}\sim\mathbb{S}^{M-1}$). Thus, the cosine value of the  random angle between $\bm{v}$ and $\mathbf{n}$ (see Figure~\ref{fig:geo-dp}) also satisfies Lemma~\ref{lemma:cos}. Then, we will obtain the PDF of $\frac{1}{\sin \theta}$ via a sequence of random variable transformations. 

For the first transformation, we denote $y = g(x)=1-x^{2}=\sin^{2}\theta$. When $x \in (-1,0)$, $g^{-1}(y)=- \sqrt{1-y}$. Applying the transformation of random variables in Lemma~\ref{thm:variable-trans}, we have
\begin{equation*}\label{eq:rho-part1}
\begin{aligned}
f_{Y}(y)  
&= f_{X}(g^{-1}(y)) \left|\frac{d}{dy} g^{-1}(y) \right| \\
&= f_{X}\left( -\sqrt{1-y} \right) \left| \frac{1}{-2\sqrt{1-y}} \right|  \\
&= \frac{1}{2\sqrt{\pi}} \frac{\Gamma\left( \frac{M}{2} \right)}{\Gamma\left( \frac{M-1}{2} \right)} y^{\frac{M-3}{2}} \frac{1}{\sqrt{1-y}}.
\end{aligned}
\end{equation*}
Similarly, when $x \in [0, 1)$, $g^{-1}(y)=\sqrt{1-y}$, we also have
\begin{equation*}\label{eq:rho-part2}
\begin{aligned}
\frac{1}{2\sqrt{\pi}} \frac{\Gamma\left( \frac{M}{2} \right)}{\Gamma\left( \frac{M-1}{2} \right)} y^{\frac{M-3}{2}} \frac{1}{\sqrt{1-y}}.  
\end{aligned}
\end{equation*}
Combine the two results above, we can get
\begin{align*}\label{eq:rho-part}
f_{Y}(y)&= \frac{1}{\sqrt{\pi}}\frac{\Gamma\left( \frac{M}{2} \right)}{\Gamma\left( \frac{M-1}{2} \right)} y^{\frac{M-3}{2}} \frac{1}{\sqrt{1-y}}.
\end{align*}

From Definition~\ref{app:def-beta-and-beta-prime}, it is clear that $f_{Y}(y)$ is  a $\Beta$ distribution with parameters $\alpha= \frac{M-1}{2} $, $\beta= \frac{1}{2}$. 
According to Lemma~\ref{lemma:beta-betaprime}, we have $\frac{Y}{1-Y}\sim\BP (\frac{M-1}{2}, \frac{1}{2} )$ and  $\frac{1-Y}{Y}=\frac{1}{Y}-1 \sim $ \resizebox{0.135\textwidth}{!}{$\BP (\frac{1}{2}, \frac{M-1}{2} )$}. Plug into the PDF of $\BP (\frac{1}{2}, \frac{M-1}{2} )$, we have
\begin{equation*}
\begin{aligned}
f_{\frac{1}{Y}-1}\left(\frac{1}{y}-1 \right)
&= \frac{1}{\sqrt{\pi}}\frac{\Gamma\left( \frac{M}{2} \right)}{\Gamma\left( \frac{M-1}{2} \right)}\left(\frac{1}{y}-1 \right)^{-\frac{1}{2}}\left(\frac{1}{y} \right)^{-\frac{M}{2}}, \frac{1}{y}-1\geq 0.
\end{aligned}
\end{equation*}

Next, we define $ z =h\left ( \frac{1}{y}-1  \right )=\frac{1}{y}$, hence $h^{-1}(z)=z-1\geq 0$. By applying theorem~\ref{thm:variable-trans} again, we arrive at 
\begin{equation*}
\begin{aligned}
f_{Z}(z) &= f_{\frac{1}{Y}-1}(h^{-1}(z)) \left|\frac{d}{dz} h^{-1}(z) \right| \\
&= {\frac{1}{\sqrt{\pi}}}\frac{\Gamma\left( \frac{M}{2} \right)}{\Gamma\left( \frac{M-1}{2} \right)}(z-1)^{-\frac{1}{2}}(z)^{-\frac{M}{2}} \left|1 \right|  \\
&= \frac{1}{\sqrt{\pi}} \frac{\Gamma\left( \frac{M}{2} \right)}{\Gamma\left( \frac{M-1}{2} \right)} (z-1)^{-\frac{1}{2}} z^{-\frac{M}{2}},\quad  z \geq 1,
\end{aligned}
\end{equation*}
which is the PDF of $\frac{1}{\sin^2\theta}$, where $\theta$ is the angle between $\bm{v}$ and $\mathbf{n}$ (or $\bm{h}\sim\mathbb{S}^{M-1}$).

Finally, we define $u=\upsilon\left ( z  \right )= \sqrt{z}$, hence $\upsilon^{-1}(u)=u^2\geq 1$. By applying theorem~\ref{thm:variable-trans} again, we arrive at
\begin{equation*}
\begin{aligned}
f_{U}(u) &= f_{Z}(\upsilon^{-1}(u)) \left|\frac{d}{du} \upsilon^{-1}(u) \right|  \\
&= {\frac{1}{\sqrt{\pi}}}\frac{\Gamma\left( \frac{M}{2} \right)}{\Gamma\left( \frac{M-1}{2} \right)}(u^2-1)^{-\frac{1}{2}}(u^2)^{-\frac{M}{2}} 2u \\
&= \frac{2}{\sqrt{\pi}} \frac{\Gamma\left( \frac{M}{2} \right)}{\Gamma\left( \frac{M-1}{2} \right)} (u^2-1)^{-\frac{1}{2}} u^{1-M},\quad u\geq 1.
\end{aligned}
\end{equation*}
which is the PDF of $\frac{1}{\sin \theta}$. 
Thus, we conclude the proof of Proposition~\ref{prop:proof-pdf-sin}. 
\end{proof}

\subsection{Proof of Proposition~\ref{prop:proof-q-th-moment-sin}} \label{app:proof-q-th-moment} 
\begin{proof}
The $q$-th moment of random variable $U$ is
\begin{equation*}
\begin{aligned}
&\mathbb{E}[U^q] \\
=&\int_{1}^{\infty}u^{q}f_{u}(u)du  \stackrel{(a)}= \frac{2}{\sqrt{\pi}} \frac{\Gamma\left( \frac{M}{2} \right)}{\Gamma\left( \frac{M-1}{2} \right)}\int_{1}^{\infty}u^{q}(u^2-1)^{-\frac{1}{2}} u^{1-M}du  \\
=& \frac{2}{\sqrt{\pi}} \frac{\Gamma\left( \frac{M}{2} \right)}{\Gamma\left( \frac{M-1}{2} \right)} \int_{1}^{\infty} u ^{q-M+1}\left(u^2-1\right)^{-\frac{1}{2}} \frac{1}{2}u^{-1} d u^2 \notag \\
=& \frac{1}{\sqrt{\pi}} \frac{\Gamma\left( \frac{M}{2} \right)}{\Gamma\left( \frac{M-1}{2} \right)} \int_{1}^{\infty} u ^{q-M}\left(u^2-1\right)^{-\frac{1}{2}} d u^2  \\
=& \frac{1}{\sqrt{\pi}} \frac{\Gamma\left( \frac{M}{2} \right)}{\Gamma\left( \frac{M-1}{2} \right)} \int_{1}^{\infty} (t+1) ^{\frac{q-M}{2}}t^{-\frac{1}{2}} d t \qquad \text{let $t=u^2-1$}\\
=& \frac{1}{\sqrt{\pi}} \frac{\Gamma\left( \frac{M}{2} \right)}{\Gamma\left( \frac{M-1}{2} \right)} B\left(\frac{1}{2},\frac{M-q-1}{2} \right) \underbrace{\int_{0}^{\infty} \frac{1}{B\left(\frac{1}{2},\frac{M-q-1}{2} \right)} t^{-\frac{1}{2}} (t+1)^{\frac{q-M}{2}}dt}_{\substack{=1, \text{\ integrand is the kernel of } \\ t\sim\BP(\frac{1}{2},\frac{M-q-1}{2})}}  \\
=& \frac{1}{\sqrt{\pi}} \frac{\Gamma\left( \frac{M}{2} \right)}{\Gamma\left( \frac{M-1}{2} \right)} B\left(\frac{1}{2},\frac{M-q-1}{2} \right) \\ 
=& \frac{1}{\sqrt{\pi}} \frac{\Gamma\left( \frac{M}{2} \right)}{\Gamma\left( \frac{M-1}{2} \right)} \frac{\Gamma \left ( \frac{1}{2} \right )\Gamma \left ( \frac{M-q}{2}-\frac{1}{2}\right )}{\Gamma \left ( \frac{M-q}{2} \right )} \\
=& \frac{\Gamma\left( \frac{M}{2} \right)}{\Gamma\left( \frac{M-1}{2} \right)} \frac{\Gamma \left ( \frac{M-q}{2}-\frac{1}{2} \right )}{\Gamma \left ( \frac{M-q}{2} \right )} \quad(\text{$1<q< M-1$}),
\end{aligned}
\end{equation*}
where $(a)$ is obtained by plugging the PDF of Beta Prime distribution in Proposition~\ref{prop:proof-pdf-sin}. Thus, we complete the proof.
\end{proof}

\subsection{Proof of Proposition~\ref{prop:proof-min-q-result}}\label{app:tail-bound} 
\begin{proof}
Proposition~\ref{prop:proof-q-th-moment-sin} gives $\mathbb{E}[U^q] = \frac{\Gamma\left( \frac{M}{2} \right)}{\Gamma\left( \frac{M-1}{2} \right)} \frac{\Gamma \left ( \frac{M-q}{2}-\frac{1}{2} \right )}{\Gamma \left ( \frac{M-q}{2} \right )}$, where $1<q<M-1$.
By applying (\ref{eq:ratio-gamma-1}), we have
\begin{equation*}
    \frac{\Gamma\left( \frac{M}{2} \right)}{\Gamma\left( \frac{M-1}{2} \right)} < \sqrt{\frac{M-1}{2}}, \quad \frac{\Gamma \left ( \frac{M-q}{2}-\frac{1}{2} \right )}{\Gamma \left ( \frac{M-q}{2} \right )} < \frac{2}{\sqrt{2(M-q)-3}}.
\end{equation*}
Thus,  $\mathbb{E}[U^q] <\sqrt{\frac{M-1}{2}} \frac{2}{\sqrt{2(M-q)-3}}= \frac{\sqrt{M-1}}{\sqrt{M-q-\frac{3}{2}}}$, where $1<q < M-\frac{3}{2}$.

By applying the result of Proposition~\ref{prop:proof-q-th-moment-chi}, we have 
\begin{equation*}
\begin{aligned}
    \mathbb{E}[W^q]\leq\ & 2^{-\frac{q}{2}} e^{-\frac{\lambda^{2}}{2}} \Gamma(q+1) \left(
    \frac{1 }{\Gamma \left(\frac{q+2}{2} \right)}\ 
    {}_1F_1\left(\frac{q+1}{2} ; \frac{1}{2} ; \frac{\lambda^2}{2} \right) \right. \\
    & \left. + \frac{\sqrt{2 }\lambda }{\Gamma \left(\frac{q+1}{2} \right)}\ 
    {}_1F_1\left(\frac{q+2}{2} ; \frac{3}{2} ; \frac{\lambda^2}{2} \right) \right).
\end{aligned}
\end{equation*}

By applying (\ref{eq:duplication-formula}) and setting $z = \frac{q+1}{2}$, we have 
$\Gamma(q+1)= \pi^{-\frac{1}{2}}2^{q}\Gamma(\frac{q+1}{2})\Gamma(\frac{q+2}{2}).$
As a consequence,
\begin{equation*}
    \begin{aligned}
    \mathbb{E}[W^q] \leq & \frac{e^{-\frac{\lambda^{2}}{2}}}{\sqrt{\pi}} 2^{\frac{q}{2}}
    \bigg( \Gamma\left( \frac{q+1}{2}\right)  {}_1F_1\left(\frac{q+1}{2} ; \frac{1}{2} ; \frac{\lambda^2}{2} \right) \\
    &\quad + \sqrt{2}\lambda\Gamma\left( \frac{q+2}{2}\right)  {}_1F_1\left(\frac{q+2}{2} ; \frac{3}{2} ; \frac{\lambda^2}{2} \right) \bigg)\\
    \leq & \frac{e^{-\frac{\lambda^{2}}{2}}}{\sqrt{\pi}} 2^{\frac{q}{2}}  \Gamma\left( \frac{q+2}{2}\right)
    \bigg( {}_1F_1\left(\frac{q+1}{2} ; \frac{1}{2} ; \frac{\lambda^2}{2} \right) \\
    &\quad + \sqrt{2}\lambda\ {}_1F_1\left(\frac{q+2}{2} ; \frac{3}{2} ; \frac{\lambda^2}{2} \right) \bigg) .
    \end{aligned}
\end{equation*}

Next, to make the tail probability negligible, we consider $\mathbb{E}[W^q]t^{-q}$ together. 
In particular, we wish to  bound $2^ {\frac{q}{2}}\Gamma \left(\frac{q+2}{2} \right) t^{-q}$ using a monotonically decreasing function in $q$. We construct such function as the ratio of two gamma functions, i.e., $\frac{1}{k}\frac{\Gamma \left(\frac{q+1}{2}+\frac{1}{2} \right)}{\Gamma \left(\frac{q+1}{2}+1 \right)}$ and $k > 1$. 
To this end, we first find a sufficient condition on $t$ to make the following hold,
\begin{equation*}
\begin{aligned}
2^ {\frac{q}{2}} \Gamma \left(\frac{q+2}{2} \right) t^{-q}&= \frac{\Gamma \left(\frac{q+1}{2}+\frac{1}{2} \right)}{\left(  \frac{t^2}{2} \right)^{\frac{q}{2}}}
&< \frac{1}{k} \frac{\Gamma \left(\frac{q+1}{2}+\frac{1}{2} \right)}{\Gamma \left(\frac{q+1}{2}+1 \right)}, 
\end{aligned}
\end{equation*}
i.e.,
\begin{equation}\label{eq:needing-bound}
    \begin{aligned}
        \left(  \frac{t^2}{2} \right)^{\frac{q}{2}} >k \Gamma \left(\frac{q+1}{2}+1 \right).
    \end{aligned}
\end{equation}

According to  (\ref{eq: gamma-1}), we have
\begin{equation*}
    \begin{aligned}
        \frac{\left(\frac{q+1}{2}+1 \right)^{\left(\frac{q+1}{2}+1 \right)-\frac{1}{2}}}{e^{\left(\frac{q+1}{2}+1 \right)-1}}>\Gamma \left(\frac{q+1}{2}+1 \right).
    \end{aligned}
\end{equation*}

Hence, a sufficient condition for (\ref{eq:needing-bound}) to hold is
\begin{equation*}
   \begin{aligned}
       \left(  \frac{t^2}{2} \right)^{\frac{q}{2}}> k \frac{\left(\frac{q+1}{2}+1 \right)^{\left(\frac{q+1}{2}+1 \right)-\frac{1}{2}}}{e^{\left(\frac{q+1}{2}+1 \right)-1}},
   \end{aligned}
\end{equation*}
which suggests
\begin{equation*}
    \begin{aligned}
        t^2 >  2 k^{\frac{2}{q}} \frac{ \left(\frac{q+3}{2} \right)^{\left(1+\frac{2}{q}\right)}}{e^{(1+\frac{1}{q})}}.
    \end{aligned}
\end{equation*}

As a result, $2^ {\frac{q}{2}} \Gamma \left(\frac{q+2}{2} \right) t^{-q}< \frac{1}{k}\frac{\Gamma \left(\frac{q+1}{2}+\frac{1}{2} \right)}{\Gamma \left(\frac{q+1}{2}+1 \right)}< \frac{1}{k} \frac{1}{\sqrt{q+\frac{3}{4}}}$ (where the last inequality is due to (\ref{eq:ratio-gamma-1})), 
when $t^2 > 2 k^{\frac{2}{q}}  \frac{ \left(\frac{q+3}{2} \right)^{\left(1+\frac{2}{q}\right)}}{e^{(1+\frac{1}{q})}}$. Finally, we   arrive at the tail probability of $\Pr[WU\geq t]$ as shown in (\ref{eq:min-q-result}). which concludes the proof. 
\end{proof}

\subsection{Solving for \texorpdfstring{$\delta = \min \limits_{q} \delta_1\delta_2$}{delta = min delta1 delta2}}\label{app:approximate-delta}
Since $\delta = \delta_1\delta_2$ involves  the confluent hypergeometric function, which is cumbersome to work with, we first consider an approximation of $\delta$ by deriving it upper bound. See Proposition~\ref{prop: approximate-delta}.

\begin{proposition} \label{prop: approximate-delta} 
   Given any $1<q<M-\frac{3}{2}$,  $k>1$,  and  $t^2 = 2 k^{\frac{2}{q}} \frac{ \left(\frac{q+3}{2} \right)^{\left(1+\frac{2}{q}\right)}}{e^{(1+\frac{1}{q})}}$,
we have
\begin{equation}\label{eq:delta-approx}
    \begin{aligned}
        \delta <& \frac{2e^{-\frac{\lambda^{2}}{2}}\sqrt{M-1}}{k\sqrt{\pi}} \frac{1+\lambda^2(q+1)}{\sqrt{(M-q-\frac{3}{2})(q+\frac{3}{4})}},
    \end{aligned}
\end{equation}
which is minimized when $q = q^{\ast} = \frac{ 4M-2\lambda^2M-9}{4\lambda^2M+8-\lambda^2}$.
\end{proposition}

\begin{proof}
We set     $\frac{\sigma_M }{\Delta_2f} \epsilon > t$ according to (\ref{eq:plrv-tail}). 
With some foresight, we consider $\lambda< \frac{1}{\sqrt{2}}$, which in hindsight is trivially true. This condition turns out to be a justified posteriori, since $\lambda = \frac{\Delta_2f}{\sigma_M} = \Theta(1/M)$ and our final result indeed implies $\lambda \ll 1$ for high dimensional applications.  As a result, we have $\lambda^2 < \frac{1}{2}$. Then, Lemma~\ref{thm: hypergeometric-ineq} can be invoked to bound the value of the confluent hypergeometric function of the first kind. As a result, an upper bound of  $\delta$ in (\ref{eq:min-q-result}) takes the following form 
{\footnotesize
\begin{equation}\label{eq:approximate-delta-1}
     \begin{aligned}
         \delta < \frac{e^{-\frac{\lambda^{2}}{2}}}{k\sqrt{\pi}} \left( 1+ \lambda^2(q+1) + \sqrt{2}\lambda\left( 1+\frac{\lambda^2(q+2)}{3} \right) \right)
     \frac{\sqrt{M-1}}{\sqrt{(M-q-\frac{3}{2})(q+\frac{3}{4})}}.
     \end{aligned}
\end{equation}
}\ignorespacesafterend
Since for any $1<q<M-\frac{3}{2}$, we have $1>\sqrt{2}\lambda$, $\lambda^2q>\frac{\sqrt{2}\lambda^3}{3}q$, $\lambda^2>\frac{2\sqrt{2}}{3}\lambda^3$. Hence,   we   have $\delta<\text{r.h.s\ of\ }  (\ref{eq:approximate-delta-1})<\text{r.h.s\ of\ }(\ref{eq:delta-approx})$.

Next, we proceed to minimize the r.h.s   of (\ref{eq:delta-approx}) for $1<q<M-\frac{3}{2}$.
The term $\frac{2e^{-\frac{\lambda^{2}}{2}} \sqrt{M-1}}{k\sqrt{\pi}}$ is a positive constant, so we only need to consider the minimum value of $\frac{1+ \lambda^2(q+1) }{\sqrt{(M-q-\frac{3}{2})(q+\frac{3}{4})}}$ with respect to $q$.
Define the function $f(q)= \ln\left( \frac{1+ \lambda^2(q+1) }{\sqrt{(M-q-\frac{3}{2})(q+\frac{3}{4})}} \right)$, then we have first-order derivative of $f(q)$ as 
\begin{equation*} \label{first-order-derivative}
\begin{aligned}
    f^\prime(q) &= \frac{\lambda^2}{1+\lambda^2q+\lambda^2} + \frac{8q-4M+9}{-8q^2+8Mq-18q+6M-9} \\
    &= \frac{\left(4M \lambda^2 + 8 -\lambda^2\right)q +2 M \lambda^2 -4M -9}{\left( 1+\lambda^2q+\lambda^2 \right) \left( -8q^2+8Mq-18q+6M-9 \right)}, 
\end{aligned}
\end{equation*}
Let $f^\prime(q) =0$, we   calculate $q^{\ast}=\frac{ 4M-2M\lambda^2-9}{4M\lambda^2+8-\lambda^2}$.

Next, we study the monotonicity of $f(q)$. 
First, for $-8q^2+8Mq-18q+6M-9$, one can verify that the two roots 
are $q = -\frac{3}{4}$ and $q=M-\frac{3}{2}$, which do not belong to our considered range $1<q<M-\frac{3}{2}$. Thus,  $-8q^2+8Mq-18q+6M-9$ is always positive when $1<q<M-\frac{3}{2}$, which suggests that $\left( 1+\lambda^2q+\lambda^2 \right) \left( -8q^2+8Mq-18q+6M-9 \right)$ is always positive in our considered range.

Second, observe that $4M\lambda^2+8-\lambda^2=\lambda^2(4M-1)+8>0$, since \(M>2\) and \(0<\lambda^2<\frac{1}{2}\). Define $g(q)=\left(4M\lambda^2+8-\lambda^2\right)q+2M\lambda^2-4M-9$. Then $g(q)$ is a strictly increasing linear function in $q$. It follows that $g(q)<0$ for $q<\frac{4M-2M\lambda^2-9}{4M\lambda^2+8-\lambda^2}$, and $g(q)>0$ for $q>\frac{4M-2M\lambda^2-9}{4M\lambda^2+8-\lambda^2}$. Therefore, when $q$ is $\frac{ 4M-2M\lambda^2-9}{4M\lambda^2+8-\lambda^2}\triangleq q^{\ast}$, $f(q)$ is minimized, so does the approximated $\delta$ in (\ref{eq:delta-approx}). This concludes the proof of Proposition~\ref{prop: approximate-delta}.
\end{proof}

Given the fact that, in high-dimensional settings,   $\lambda<\frac{1}{\sqrt{2}}$; thus, we have $q^{\ast} \approx \frac{M}{2}$. As a result, in Theorem~\ref{thm:mian-thm}, we consider a  $\delta$ that is obtained by plugging in $q = \frac{M}{2}$ into (\ref{eq:min-q-result}).
\begin{figure}[htp]
    \centering
    \begin{subfigure}{0.49\columnwidth}  
        \centering
        \includegraphics[width=\linewidth]{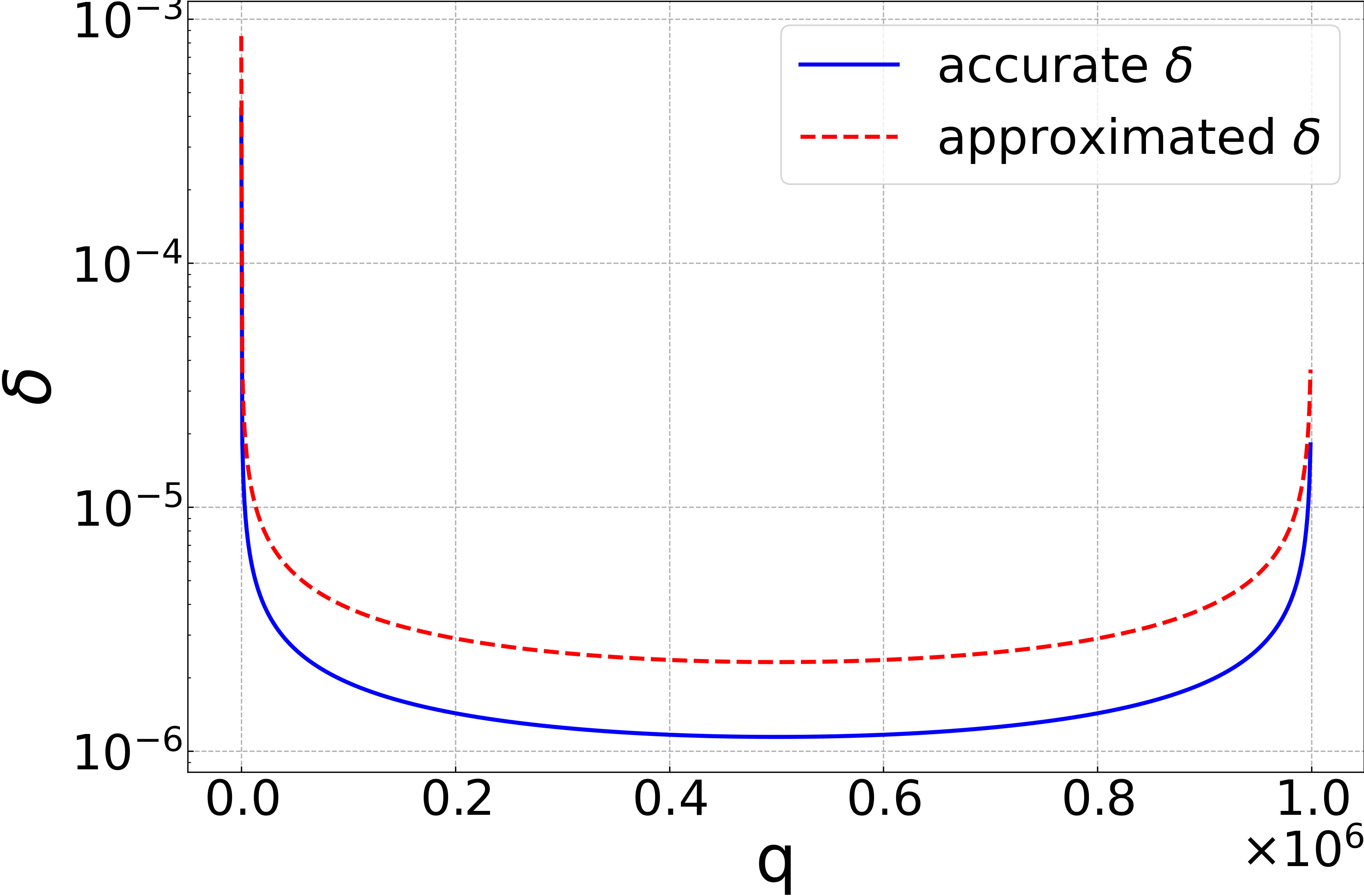}
        \Description{Graph showing delta versus q for M=10^6.}
        \caption{$\delta$ versus $q$ ($M=10^6$)}
        \label{fig:q-delta-10e6}
    \end{subfigure}
    \hfill
    \begin{subfigure}{0.49\columnwidth}
         \centering
        \includegraphics[width=\linewidth]{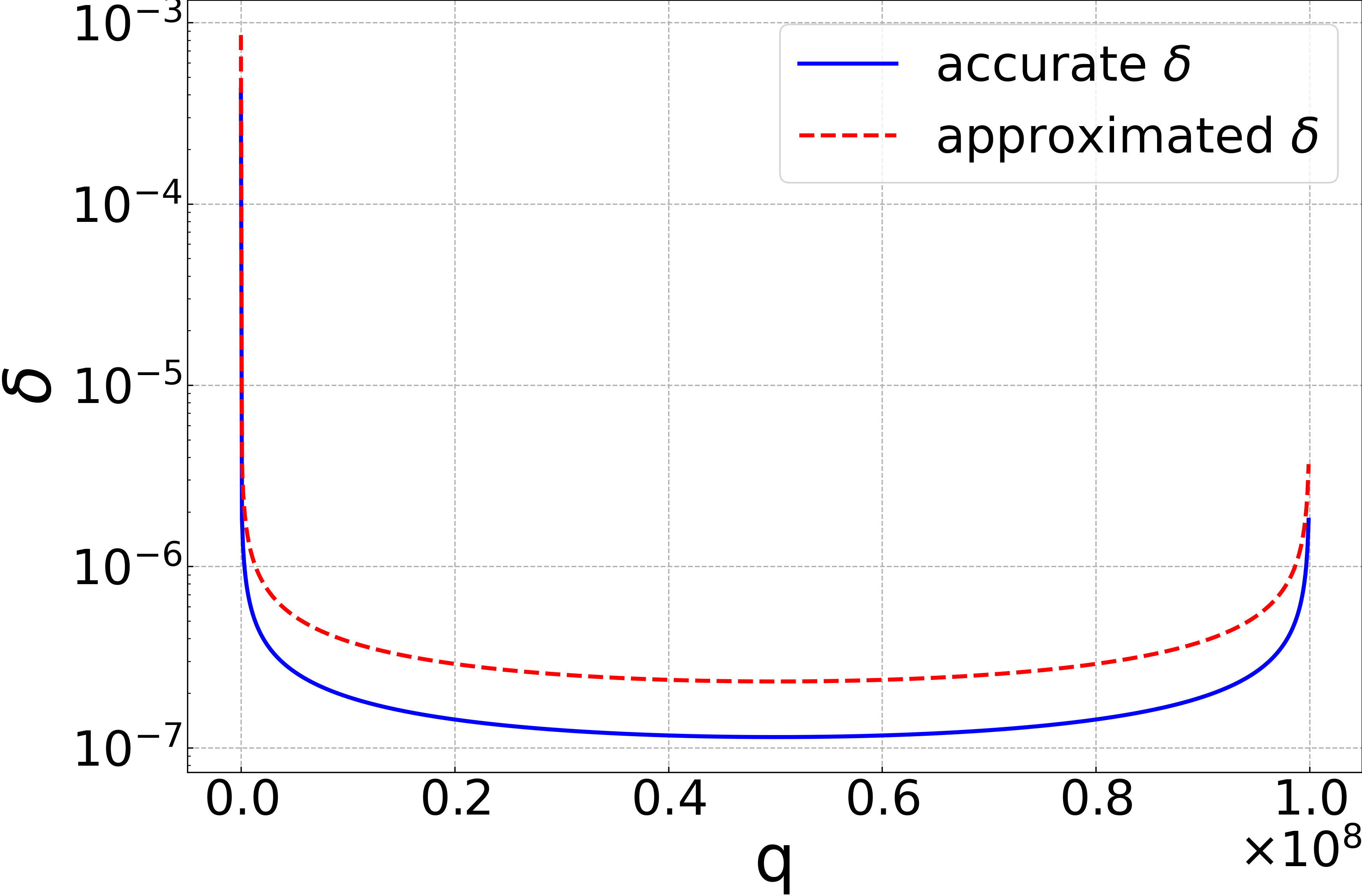}
        \Description{Graph showing delta versus q for M=10^8.}
        \caption{$\delta$ versus $q$ ($M=10^8$)}
        \label{fig:q-delta-10e8}
    \end{subfigure}
    \caption{Evaluations of $\delta$ and $q$ for dimension $M$.}
    \label{fig:simulation-delta}
\end{figure}

To verify the optimality of having $q = \frac{M}{2}$ when evaluating $\delta$, we     plot  the values of  accurate $\delta$ (in (\ref{eq:delta-accurate})) and approximated $\delta$ (in the r.h.s. of   (\ref{eq:delta-approx})) by increasing $q$ from 2 to $M-2$. The plots are visualized in Figure~\ref{fig:simulation-delta}.  In particular, we set $\epsilon=0.1$, $k=1000$, and $M\in\{10^6,10^8\}$. From Figure \ref{fig:simulation-delta}, we observe that when $q$ is about $\frac{M}{2}$, both accurate and approximated $\delta$ are close to the   minimum. For example, when $M = 10^6$, $q^{\ast}$ that minimizes (\ref{eq:delta-accurate}) and the r.h.s. of  (\ref{eq:delta-approx}) are $500,500$ and  $499,499$, respectively, and $\delta$ valuated using (\ref{eq:delta-accurate}) and (\ref{eq:delta-approx}) are   $1.144 \times 10^{-6}$ and $2.318 \times 10^{-6}$, respectively. When $M = 10^8$, $q^{\ast}$ that minimizes (\ref{eq:delta-accurate}) and the r.h.s. of  (\ref{eq:delta-approx}) are $50,050,050$ and  $49,949,950$, respectively, and $\delta$ valuated using (\ref{eq:delta-accurate}) and (\ref{eq:delta-approx}) are   $1.144 \times 10^{-7}$ and $2.318 \times 10^{-7}$, respectively. 

\subsection{Proof of Corollary~\ref{corollary:mechanism_guidance} and Corollary~\ref{corollary:Asymptotic_analysis_noise}}\label{proof:Asymptotic_analysis}

In this section, we prove both Corollary~\ref{corollary:mechanism_guidance} and~\ref{corollary:Asymptotic_analysis_noise}.
\begin{proof}
Consider the product noise in (\ref{eq:noise-generation}) and $\sigma_M$ in Theorem~\ref{thm:mian-thm}. The  expected value of its squared magnitude is:
\begin{equation*} \label{eq:mag-ours}
      \begin{aligned}
          \mathbb{E}[||\mathbf{n}||_2^2] = \sigma_M^2\mathbb{E}[\chi_1^2] = \sigma_M^2 \cdot 1 = \frac{(\Delta_2f)^2}{ \epsilon^2} \cdot  2 k^{\frac{4}{M}} \frac{ \left(\frac{M}{4} +\frac{3}{2} \right)^{\left(1+\frac{4}{M}\right)}}{e^{(1+\frac{2}{M})}}.
      \end{aligned}
\end{equation*}
In contrast, the least expected squared magnitude of the classic Gaussian mechanism~\cite{dwork2006our} is
\begin{equation*} \label{eq:mag-classic}
    \begin{aligned}
        \mathbb{E}[||\mathbf{n}_{\mathrm{classic}}||_2^2] = \sigma^2 \mathbb{E}[\chi_M^2] = \sigma^2 M = \frac{2 \log\left( \frac{1.25}{\delta} \right) (\Delta_2 f)^2}{\epsilon^2} \cdot M.
    \end{aligned}
\end{equation*}
Given same privacy parameters $\epsilon$ for both noises, we have
\begin{equation} \label{eq:ratio_noise}
    \begin{aligned}
        f(M) =&  \frac{\mathbb{E}[||\mathbf{n}||_2^2]}{\mathbb{E}\left[||\mathbf{n}_{\mathrm{classic}}||_2^2\right]} \\
        =&  \frac{\frac{(\Delta_2f)^2}{ \epsilon^2} \cdot  2 k^{\frac{4}{M}} \frac{ \left(\frac{M}{4} +\frac{3}{2} \right)^{\left(1+\frac{4}{M}\right)}}{e^{(1+\frac{2}{M})}}}{ \frac{2 \log\left( \frac{1.25}{\delta} \right) (\Delta_2 f)^2}{\epsilon^2} \cdot M}\\
        =& \frac{k^{\frac{4}{M}}\frac{\left(\frac{M}{4} + \frac{3}{2}\right)^{\left(1+\frac{4}{M}\right)}}{e^{(1+\frac{2}{M})}}}{\log\left(\frac{1.25}{\delta}\right) M }.
    \end{aligned}
\end{equation}

Conservatively, we set $k = 10^5$ to ensure that our mechanism can achieve $\delta \le 10^{-5}$. 
    By requiring $f(M) < 1$, we obtain
    \begin{equation*}
        \begin{aligned}
            \frac{k^{\frac{4}{M}}\left(\frac{M}{4} + \frac{3}{2}\right)^{\left(1+\frac{4}{M}\right)}}{e^{(1+\frac{2}{M})}\log\left(\frac{1.25}{\delta}\right) M} < 1.
        \end{aligned}
    \end{equation*}
A sufficient condition for the above requirement to hold is
    \begin{equation*}
        \begin{aligned}
            k^{\frac{4}{M}} \left(\frac{M}{4} + \frac{3}{2}\right)^{\left(1+\frac{4}{M}\right)}< e \log \left(\frac{1.25}{\delta}\right) M.
        \end{aligned}
    \end{equation*}
Taking the logarithm on both sides yields
    \begin{equation}\label{eq:after-log}
        \begin{aligned}
            \frac{4}{M}\log k + \left(1+\frac{4}{M}\right) \log\left(\frac{M}{4} + \frac{3}{2}\right) < 1 + \log\left(\log\left(\frac{1.25}{\delta}\right)\right) + \log M.
        \end{aligned}
    \end{equation}

Define $g(M) = \left(1+\frac{4}{M}\right)\log\left(\frac{M}{4} + \frac{3}{2}\right) - \log M$. By direct differentiation, one can verify
    \begin{equation*}
        \begin{aligned}
            g^{\prime}(M) =& \left(1+\frac{4}{M}\right)\cdot \frac{1}{\frac{M}{4} + \frac{3}{2}}\cdot \frac{1}{4} - \frac{4}{M^{2}} \log\left(\frac{M}{4} + \frac{3}{2}\right) - \frac{1}{M} \\
            =& \frac{-2M - 4(M+6)\log\left(\frac{M+6}{4}\right)}{M^{2}(M+6)}<0, \forall M>0.
        \end{aligned}
    \end{equation*}
Therefore, $g(M)$ is monotonically decreasing for all $M > 0$.  Since $g(7) = -0.09374$, we have  $\left(1+\frac{4}{M}\right)\log\left(\frac{M}{4} + \frac{3}{2}\right) < \log M$   when $M\geq 7$. Thus, it suffices to require   $\frac{4}{M}\log k < 1 + \log(\log(\frac{1.25}{\delta}))$, which suggests $M>13.2995$ when $k = 10^5, \delta = 10^{-5}$. 

Therefore, when the dimension $M \ge 14$ , we have $f(M)<1$, which means that our product noise has smaller expected noise squared magnitude than the classic Gaussian noise. This concludes Corollary~\ref{corollary:mechanism_guidance}.

Next, we give the proof of Corollary~\ref{corollary:Asymptotic_analysis_noise}.
It is   obtained  by analyzing the asymptotic behavior of $f(M)$ in (\ref{eq:ratio_noise}). Clearly, as $M$ approaches  infinity, $k^{4/M}$ approaches 1, $(\frac{M}{4} + \frac{3}{2})^{\left( 1 + \frac{4}{M} \right)} \approx \frac{M}{4} $, and $e^{\left( 1 + \frac{2}{M} \right)} \approx e$. Consequently, (\ref{eq:ratio_noise}) becomes 
$\frac{1}{4e \log\left(\frac{1.25}{\delta}\right)}$.
\end{proof}

\subsection{Common Assumptions in ERM} \label{sec:common-assumptions}
\begin{definition}\cite{iyengar2019towards}
    A function $f: \mathbb{R}^M \to \mathbb{R}: $
    \begin{itemize}
        \item is a convex function if for all $\omega_1$, $\omega_2 \in \mathbb{R}^M$, $f(\omega_1)-f(\omega_2) \ge \left \langle \nabla f(\omega_2), \omega_1 - \omega_2 \right \rangle $.
        \item is a $\xi$-strongly convex function if for all $\omega_1$, $\omega_2 \in \mathbb{R}^M$, $f(\omega_1) \ge f(\omega_2) + \left \langle \nabla f(\omega_2), \omega_1 - \omega_2 \right \rangle + \frac{\xi}{2} \left \| \omega_1 -\omega_2 \right \|_2^2 $  or equivalently,  $\left \langle \nabla f(\omega_1)- \nabla f(\omega_2), (\omega_1 - \omega_2) \right \rangle \ge \xi \left \| \omega_1 -\omega_2 \right \|_2^2 $.
        \item has $L$-Lipschitz constant $L$ if for all $\omega_1, \omega_2 \in \mathbb{R}^M$, $\ | f(\omega_1)  $ $ - f(\omega_2) \ | \leq  \left \| \omega_1 -\omega_2 \right \|_2$.
        \item is $\beta$-smooth if for all $\omega_1$, $\omega_2 \in \mathbb{R}^M$, $\left \|f(\omega_1) - f(\omega_2) \right \| \leq \beta \cdot \left \|\omega_1 - \omega_2 \right \|_2$.
    \end{itemize}
\end{definition}

\section{Omitted Details in Section~\ref{sec:DP_Convex_Output}}
\subsection{Pseudocode of Product Noise-Based Output Perturbation}\label{app:Pseudocode-pn-output-perturbation}
Algorithm~\ref{alg:pn-output} presents the pseudocode of product noise-based output perturbation. 
\begin{algorithm}[htp]
	\caption{Product Noise-Based Output Perturbation}\label{alg:pn-output}
	\begin{algorithmic}[0]
		\State{\bf Input:} Dataset: $D = \{d_1, \dots, d_n\}$; loss function: $\ell(\bm{\omega}; d_i)$ that has $L$-Lipschitz constant $L$, is convex in $\bm{\omega}$; regularization parameter $\Lambda$; privacy parameter $(\epsilon, \delta)$
        \begin{algorithmic}[1]
        \State Compute the ERM optimal solution $\hat{\bm{\omega} } = \arg \min\limits_{\bm{\omega} \in \mathbb{R}^M} \mathcal{L} (\bm{\omega}; D)$
        \State Draw a noise vector $\mathbf{n} = {\sigma_M} R \boldsymbol{h}, \text{where } {\sigma_M} = \frac{2L}{\Lambda} \cdot \frac{\sqrt{2} k^{\frac{2}{M}} \left(\frac{M}{4} +\frac{3}{2} \right)^{ \left( \frac{1}{2} + \frac{2}{M} \right)}}{\epsilon e^{\left( \frac{1}{2} + \frac{1}{M}\right)}}, R\sim\chi_1\ \text{and } \boldsymbol{h} \sim \mathbb{S}^{M-1}$
        \State Perturb the output $\bm{\omega}_{\mathrm{priv}} = \hat{\bm{\omega} } +\mathbf{n}$
        \end{algorithmic}
{\bf Output: }$\bm{\omega}_{\mathrm{priv}}$
\end{algorithmic}
\end{algorithm}

\subsection{Proof of Corollary~\ref{privacy_guarantee_np_output}} \label{app:output-privacy-proof}
\begin{proof} This proof is adapted from~\cite{chaudhuri2011differentially}. 
    Let $D$ and $D^\prime$ be neighboring datasets that differ in the value of the $i$-th individual, and let $\hat{\bm{\omega}}$ and $\hat{\bm{\omega}}^{\prime}$ be the optimal solutions corresponding to $D$ and $D^\prime$, respectively.
Since the optimal solution satisfies the first-order optimality condition: $\nabla \mathcal{L}(\hat{\bm{\omega}}; D) =  \frac{1}{n} \sum_{i=1}^{n} \nabla \ell(\hat{\bm{\omega}}; d_i) + \frac{\Lambda}{n} \hat{\bm{\omega}} = 0$ and $\nabla \mathcal{L}(\hat{\bm{\omega}}^{\prime}; D^{\prime}) =  \frac{1}{n} \sum_{i=1}^{n} \nabla \ell(\hat{\bm{\omega}}^{\prime}; d_i) + \frac{\Lambda}{n}\hat{\bm{\omega}}^{\prime} = 0$, subtracting these two equations, we have
\begin{equation*}
    \begin{aligned}
        \frac{1}{n} \left( \nabla \ell(\hat{\bm{\omega}}; d_i) -  \nabla \ell(\hat{\bm{\omega}}^{\prime}; d_i) \right) +  \frac{\Lambda}{n} (\hat{\bm{\omega}} - \hat{\bm{\omega}}^{\prime}) = 0.
    \end{aligned}
\end{equation*}
Since $\ell(\bm{\omega}; d_i)$ is $L$-Lipschitz, we have $\left \| \nabla {\ell(\alpha; d^{\prime}_i)} - \nabla \ell(\alpha; d_i) \right \|_2 \leq 2L$,
thus $\frac{\Lambda}{n} \left \| \hat{\bm{\omega}} - \hat{\bm{\omega}}^{\prime} \right \|_2 \leq \frac{2L}{n}$.
Then, we get $ \left \| \hat{\bm{\omega}} - \hat{\bm{\omega}}^{\prime} \right \|_2 \leq \frac{2L}{\Lambda}$, i.e., the $L_2$-sensitivity of $ \mathcal{L}(\bm{\omega}; D)$ is at most $ \frac{2L}{\Lambda}$. Setting ${\sigma_M} \ge $ $\frac{2L}{\Lambda} \cdot \frac{\sqrt{2} k^{\frac{2}{M}} \left(\frac{M}{4} +\frac{3}{2} \right)^{ \left( \frac{1}{2} + \frac{2}{M} \right)}}{\epsilon e^{\left( \frac{1}{2} + \frac{1}{M}\right)}}$ 
according to Theorem \ref{thm:mian-thm}, we complete the proof.
\end{proof}

\subsection{Proof of Corollary~\ref{utility_guarantee_np_output}} \label{app:output-utility-proof}
We provide a brief analysis of the utility guarantee for output perturbation using product noise. 
\begin{proof}
Output perturbation only adds noise to the optimization solution $\bm{\hat{\omega}}$. We have
\begin{equation*}
    \begin{aligned}
        \mathbb{E} \left [ \mathcal{L} (\bm{\omega}_{\mathrm{priv}}; D) -\mathcal{L }(\bm{\hat{\omega}}; D) \right ]  \stackrel{(a)} \leq \mathbb{E}\left [ L \left \| \bm{\omega}_{\mathrm{priv}} - \bm{\hat{\omega}} \right \|_2 \right ] = L  \mathbb{E} \left [ \left \| \mathbf{n} \right \|_2 \right ],
    \end{aligned}
\end{equation*}
where $(a)$ is due to  the definition of $L$-Lipschitz.
\end{proof}

\section{Omitted Details in Section~\ref{sec:DP_Convex_Objective}}

\subsection{Pseudocode of product noise-based AMP}\label{app:Pseudocode-pn-amp}
Algorithm~\ref{alg:PN-Based-AMP} presents the pseudocode of product noise-based AMP. 
\begin{algorithm}[htp]
	\caption{Product Noise-Based AMP (adapted from~\cite{iyengar2019towards})}\label{alg:PN-Based-AMP}
	\begin{algorithmic}[0]
		\State{\bf Input:} Dataset: $D = \{d_1, \dots, d_n\}$; loss function: $\ell(\bm{\omega}; d_i)$ that has $L$-Lipschitz constant $L$, is convex in $\bm{\omega}$, has a continuous Hessian, and is $\beta$-smooth for all $\bm{\omega} \in \mathbb{R}^M$ and all $d_i$; Hessian rank bound parameter: $r$ (the minimum of $M$ and twice the upper bound on the rank of $\ell$'s Hessian); privacy parameters: $(\epsilon, \delta)$; gradient norm bound: $\gamma$.
        
        \begin{algorithmic}[1]
        \State Set $\epsilon_1, \epsilon_2,\epsilon_3, \delta_1, \delta_2 > 0$ such that $\epsilon = \epsilon_1 + \epsilon_2$, $\delta = \delta_1 + \delta_2$, $0 < \epsilon_1 -\epsilon_3 < 1$, and tuning parameter $k > 1$.
        
        \State Set the regularization parameter: $\Lambda \geq \frac{r\beta}{\epsilon_1 -\epsilon_3}$
    
        \State 
        Set $\mathbf{n}_{1} = {\sigma}_{M_1} R \boldsymbol{h}, \text{where } {\sigma}_{M_1} = \frac{2L}{n} \cdot \frac{\sqrt{2} k^{\frac{2}{M}} \left(\frac{M}{4} +\frac{3}{2} \right)^{ \left( \frac{1}{2} + \frac{2}{M} \right)}}{\epsilon_3 e^{\left( \frac{1}{2} + \frac{1}{M}\right)}}, R\sim\chi_1\ \text{and } \boldsymbol{h} \sim \mathbb{S}^{M-1}$

        \State Define 
        $ \mathcal{L}_{\mathrm{priv}}(\bm{\omega}; D) = \frac{1}{n} \sum_{i=1}^{n} \ell(\bm{\omega}; d_i) + \frac{\Lambda}{2n} \|\bm{\omega}\|^2 + \mathbf{n}_1^{\top} \bm{\omega}$
    
        \State $\bm{\omega}_{\mathrm{approx}} \gets  \bm{\omega} \text{ such that } \|\nabla  \mathcal{L}_{\mathrm{priv}}(\bm{\omega}_{\mathrm{approx}}; D)\| \leq \gamma$

        \State 
        Set $\mathbf{n}_{2} = {\sigma}_{M_2} R \boldsymbol{h}, \text{where } {\sigma}_{M_2} = \frac{n\gamma}{\Lambda} \cdot \frac{\sqrt{2} k^{\frac{2}{M}} \left(\frac{M}{4} +\frac{3}{2} \right)^{ \left( \frac{1}{2} + \frac{2}{M} \right)}}{\epsilon_2 e^{\left( \frac{1}{2} + \frac{1}{M}\right)}}, R\sim\chi_1\ \text{and } \boldsymbol{h} \sim \mathbb{S}^{M-1}$
        \end{algorithmic}
                {\bf Output: }
                $ \bm{\omega}_{\mathrm{out}} = \bm{\omega}_{\mathrm{approx}} + \mathbf{n}_2$ 
\end{algorithmic}
\end{algorithm}

\subsection{Proof of Corollary~\ref{privacy_guarantee_np_amp}}\label{app:obj-privacy-proof}
 To prove Algorithm \ref{alg:PN-Based-AMP} (product noise-based AMP) is ($\epsilon, \delta$)-DP, we follow the proof procedure of AMP~\cite{iyengar2019towards}  while modifying the analysis to accommodate the use of product noise.

We define the optimal  solution $\bm{\omega}_{\mathrm{min}} = \arg \min\limits_{\bm{\omega} \in \mathbb{R}^M} $ $\mathcal{L}_{\mathrm{priv}} (\bm{\omega}, D)$, and the adjusted gradient direction $b_{\mathrm{PN}} (\bm{\omega}; D) = -\nabla \mathcal{L}(\bm{\omega}; D) - \frac{\Lambda }{n} \bm{\omega}$ for $D \in \mathcal{D}^n$ and $\bm{\omega} \in \mathbb{R}^M$. For neightporing datasets $D, D^{\prime} \in \mathcal{D}^{n}$, we require
    \begin{equation*}
        \begin{aligned}
            \frac{\mathrm{pdf}(\bm{\omega}_{\mathrm{min}} = \alpha \mid D)}{\mathrm{pdf}(\bm{\omega}_{\mathrm{min}} = \alpha \mid D^{\prime})} \leq e^{{\epsilon}_1} \quad \text{w.p. } \ge 1-\delta_1.
        \end{aligned}
    \end{equation*}
Based on the  Function Inverse theorem \cite{billingsley2017probability}, the above probability density ratio  can be decomposed into two terms: one related to the adjusted gradient direction of the perturbed objective function and the other related to the determinant of its Jacobian, i.e., 
    \begin{equation*}
    \begin{aligned}
        \frac{\mathrm{pdf}(\bm{\omega}_{\mathrm{min}} = \alpha \mid D)}{\mathrm{pdf}(\bm{\omega}_{\mathrm{min}} = \alpha \mid D^{\prime})} =& \frac{\mathrm{pdf}\left( b_{\mathrm{PN}} (\alpha; D); \epsilon_1, \delta_1, L \right)}{\mathrm{pdf}\left( b_{\mathrm{PN}} (\alpha; D^{\prime}); \epsilon_1, \delta_1, L \right)} \cdot \frac{\left | \text{det}\left( \nabla b_{\mathrm{PN}}(\alpha; D^{\prime}) \right)  \right |}{ \left | \text{det}\left( \nabla b_{\mathrm{PN}}(\alpha; D) \right)  \right |} 
    \end{aligned}.
    \end{equation*}
  Next, we establish the following two inequalities step by step, i.e.,  bounding the ratio of the densities: $\frac{\mathrm{pdf}\left( b_{\mathrm{PN}} (\alpha; D); \epsilon_1, \delta_1, L \right)}{\mathrm{pdf}\left( b_{\mathrm{PN}} (\alpha; D^{\prime}); \epsilon_1, \delta_1, L \right)}  \leq e^{\epsilon_3}$ w.p. at least $1-\delta_1$, and then Bounding $ \frac{\left | \text{det}\left( \nabla b_{\mathrm{PN}}(\alpha; D^{\prime}) \right)  \right |}{ \left | \text{det}\left( \nabla b_{\mathrm{PN}}(\alpha; D) \right)  \right |} \leq e^{\epsilon_1-\epsilon_3}$ if $\epsilon_1 - \epsilon_3 <1$.

Similar to \cite{iyengar2019towards}, we give Lemma~\ref{Lemma: Bound_ratio_of_densities} and \ref{Lemma: Bound_determinant_of_Jacobian}. 
\begin{lemma} \label{Lemma: Bound_ratio_of_densities}
    For any pair of neighboring datasets \( D, D^{\prime} \in \mathcal{D}^n\) and \( \epsilon_1 - \epsilon_3 < 1 \), we have
    \begin{equation*}
        \begin{aligned}
            \frac{\mathrm{pdf}\left( b_{\mathrm{PN}} (\alpha; D); \epsilon_1, \delta_1, L \right)}{\mathrm{pdf}\left( b_{\mathrm{PN}} (\alpha; D^{\prime}); \epsilon_1, \delta_1, L \right)} \quad \mathrm{w.p.} \ge 1 - \delta_1.
        \end{aligned}
    \end{equation*}
\end{lemma}
 \begin{proof}
     Assume w.l.o.g that $d_i \in D$ has been replaced by ${d^{\prime}_i} \in D^{\prime}$, based on Lemma IV.1. in \cite{iyengar2019towards}, we can bound the $L_2$-sensitivity of $b(\alpha; )$ as $\left \| b(\alpha; D) - b(\alpha; D^{\prime}) \right \|_2 \leq  \frac{\left \| \nabla {\ell(\alpha; d^{\prime}_i)} - \ell(\alpha; d_i) \right \|_2}{n} \leq \frac{2L}{n}$. Then we set ${\sigma}_{M_1} \ge \frac{2L}{n} \cdot \frac{\sqrt{2} k^{\frac{2}{M}} \left(\frac{M}{4} +\frac{3}{2} \right)^{ \left( \frac{1}{2} + \frac{2}{M} \right)}}{\epsilon_3 e^{\left( \frac{1}{2} + \frac{1}{M}\right)}}$ for $\epsilon_3 < \epsilon_1$, we can get the proof of the Lemma from the privacy guarantee of product noise in Theorem \ref{thm:mian-thm}.
 \end{proof}

\begin{lemma}[\cite{iyengar2019towards}]\label{Lemma: Bound_determinant_of_Jacobian}
    For any pair of neightporing datasets \( D, D^{\prime} \in \mathcal{D}^n\), and \( \epsilon_1 - \epsilon_3 < 1 \), we have
    \begin{equation*}
       \begin{aligned}
           \frac{\left | \mathrm{det}\left( \nabla b_{\mathrm{PN}}(\alpha; D^{\prime}) \right)  \right |}{ \left | \mathrm{det}\left( \nabla b_{\mathrm{PN}}(\alpha; D) \right)  \right |}  \leq e^{\epsilon_1-\epsilon_3}.
       \end{aligned}
    \end{equation*}
\end{lemma}

\begin{proof}
W.l.o.g, let $d_i \in D$ is replaced by ${d^{\prime}_i} \in D^{\prime}$. \cite[Lemma IV.2]{iyengar2019towards} gives $\frac{\left | \text{det}\left( \nabla b_{\mathrm{PN}}(\alpha; D^{\prime}) \right)  \right |}{ \left | \text{det}\left( \nabla b_{\mathrm{PN}}(\alpha; D) \right)  \right |}\leq \frac{\Lambda}{\Lambda - r \beta}$, where $\Lambda$ is the regularization parameter, $\beta$ is the smoothness constant, and $r$ is the minimum of $M$ and twice the upper bound on the rank of $\ell$'s Hessian.   Setting $\Lambda= \frac{r \beta}{\epsilon_1 - \epsilon_3}$ attains the reslut.
\end{proof}
Combining  Lemma \ref{Lemma: Bound_ratio_of_densities} and Lemma \ref{Lemma: Bound_determinant_of_Jacobian}, we have $\frac{\mathrm{pdf}(\bm{\omega}_{\mathrm{min}} = \alpha \mid D)}{\mathrm{pdf}(\bm{\omega}_{\mathrm{min}} = \alpha \mid D^{\prime})} \leq e^{{\epsilon}_1}, \mathrm{w.p. } \ge 1-\delta_1$, i.e., $\bm{\omega}_{\mathrm{min}}$ is $(\epsilon_1, \delta_1)$-differentially private.

We can write $\bm{\omega}_{\mathrm{out}} = \bm{\omega}_{\mathrm{approx}} + \mathbf{n}_2 = \bm{\omega}_{\mathrm{approx}} - \bm{\omega}_{\mathrm{min}} + \mathbf{n}_2 + \bm{\omega}_{\mathrm{min}}$, and we have previously proved that $\bm{\omega}_{\mathrm{min}}$ is $(\epsilon_1, \delta_1)$-differentially private, so we just need to prove that $(\bm{\omega}_{\mathrm{approx}} - \bm{\omega}_{\mathrm{min}} + \mathbf{n}_2)$ is $(\epsilon_2, \delta_2)$-DP, shown in the next Lemma.
\begin{lemma}
    For $D \in \mathcal{D}^n$, let $\gamma \ge 0$ be chosen independently of $D$. If $\bm{\omega}_{\mathrm{approx}} \in \mathbb{R}^M$, such that $\left \| \nabla \mathcal{L}(\bm{\omega}_{\mathrm{approx}; D}) \right \|_2 \leq \gamma $, then releasing $(\bm{\omega}_{\mathrm{approx}} - \bm{\omega}_{\mathrm{min}} + \mathbf{n}_2)$ is $(\epsilon_2, \delta_2)$-DP, where $\mathbf{n}_2 ={\sigma}_{M_2} R \boldsymbol{h}$ and ${\sigma}_{M_2} = \frac{n\gamma}{\Lambda} \cdot \frac{\sqrt{2} k^{\frac{2}{M}} \left(\frac{M}{4} +\frac{3}{2} \right)^{ \left( \frac{1}{2} + \frac{2}{M} \right)}}{\epsilon_2 e^{\left( \frac{1}{2} + \frac{1}{M}\right)}}$.
\end{lemma}
\begin{proof}
    Based on   \cite[Lemma IV.3]{iyengar2019towards}, we can   bound the $L_2$-norm of $\bm{\omega}_{\mathrm{approx}} - \bm{\omega}_{\mathrm{min}}$ as
\begin{equation*}
    \begin{aligned}
         \left \|\bm{\omega}_{\mathrm{approx}} - \bm{\omega}_{\mathrm{min}}  \right \|_2 
    \leq \frac{n \left \| \nabla \mathcal{L}_{\mathrm{priv}(\bm{\omega}_{\mathrm{approx}}; D)} 
    - \nabla \mathcal{L}_{\mathrm{priv}(\bm{\omega}_{\mathrm{min}}; D)} \right \|_2}{\Lambda} 
    \leq \frac{n \gamma}{\Lambda}. 
    \end{aligned}
\end{equation*}
Setting ${\sigma}_{M_2} = \frac{n\gamma}{\Lambda} \cdot \frac{\sqrt{2} k^{\frac{2}{M}} \left(\frac{M}{4} +\frac{3}{2} \right)^{ \left( \frac{1}{2} + \frac{2}{M} \right)}}{\epsilon_2 e^{\left( \frac{1}{2} + \frac{1}{M}\right)}}$ gives the result.

Based on the linear composition of DP \cite{dwork2014algorithmic}, we   get the privacy guarantee of Algorithm~\ref{alg:PN-Based-AMP} as $\epsilon = \epsilon_1 + \epsilon_2$ and $\delta = \delta_1 + \delta_2$. This concludes the proof of Corollary~\ref{privacy_guarantee_np_amp}.
\end{proof} 

\subsection{Proof of Corollary~\ref{utility_guarantee_np_amp}}\label{app:obj-utility-proof}

The proof is adapted from the  proof of    \cite[Theorem 4]{kifer2012private} and    \cite[Lemma A.1]{iyengar2019towards}. 

According to    \cite[Lemma A.1]{iyengar2019towards}, we have $\mathcal{L}(\bm{\omega}_{\mathrm{out}}; D) - \mathcal{L}(\hat{\bm{\omega}}; D) \leq L\left( \frac{n\gamma}{\Lambda} + \left \| \mathbf{n}_2 \right \|_2  \right) + \frac{\Lambda {\left \| \bm{\omega} \right \|_2^{2}} }{2n} + \frac{2n {\left \|\mathbf{n}_1 \right \|_2^{2}}}{\Lambda}$.   We first bound $\left \| \mathbf{n}_s \right \|_2 $, $s\in\{1,2\}$, and $n_{s} = {\sigma_M}_s R \boldsymbol{h}, \text{where } R\sim\chi_1\ \text{and } \boldsymbol{h} \sim \mathbb{S}^{M-1}$.

In particular, according to    \cite[Lemma 2]{dasgupta2007probabilistic}, when  $R\sim\chi_1$, for any $\mu \ge 1$, it satisfies  
\begin{equation*}
    \begin{aligned}
        \Pr \left[ \left \| R   \right \|_2^2 \ge \mu \right] \leq \left( e^{\mu - 1 - \ln \mu} \right)^{-\frac{1}{2}}.
    \end{aligned}
\end{equation*}
Then we   get
\begin{equation*}
\begin{aligned}
    \Pr \left[ \left \| R   \right \|_2^2 {\sigma_M}_s^2 \ge \mu {\sigma_M}_s^2 \right] \leq \left( e^{\mu - 1 - \ln \mu} \right)^{-\frac{1}{2}},
\end{aligned}
\end{equation*}
which implies
\begin{equation*}
    \begin{aligned}
        \Pr \left[ \left \| \mathbf{n}_s   \right \|_2^2 \ge \mu {\sigma_M}_s^2 \right] \leq \left( e^{\mu - 1 - \ln \mu} \right)^{-\frac{1}{2}}.
    \end{aligned}
\end{equation*}
By letting $\left( e^{\mu - 1 - \ln \mu} \right)^{-\frac{1}{2}}\leq \frac{\delta}{2}$, we have $\mu \ge \frac{1 + \sqrt{1 + 4\left( 2 \log (\frac{2}{\delta} ) + 1 \right)}}{2}$. As a result,   $\left \| \mathbf{n}_s   \right \|_2 \leq \sqrt{\mu} {\sigma_M}_s \leq {\sigma_M}_s \cdot \frac{1 + \sqrt{1 + 4\left( 2 \log (\frac{2}{\delta} ) + 1 \right)}}{2}$ hold w.p. at least $1- \frac{\delta}{2}$. 

Substituting this into $\mathcal{L}(\bm{\omega}_{\mathrm{out}}; D) - \mathcal{L}(\hat{\bm{\omega}}; D)$, we get that w.p. $\ge 1-\delta$,

\begin{equation*}
    \begin{aligned}
        \mathcal{L}(\bm{\omega}_{\mathrm{out}}; D) - \mathcal{L}(\hat{\bm{\omega}}; D) \leq&\ L \left( \frac{n\gamma}{\Lambda} + {\sigma}_{M_2} \cdot \frac{1 + \sqrt{1 + 4\left( 2 \log \left(\frac{2}{\delta} \right) + 1 \right)}}{2} \right) \\
        &+ \frac{\Lambda {\left \| \bm{\omega} \right \|_2^{2}} }{2n} + \frac{n^2 {\sigma}_{M_1}^2 \left( 1+\sqrt{4(2\log( \frac{2}{\delta})+1)} \right)^2}{\Lambda}.
\end{aligned}
\end{equation*}

Next, we expect to bound the regularization parameter $\Lambda$. To guarantee that $\mathcal{L}(\bm{\omega}; D)$ is  $\frac{\Lambda}{n}$-strongly convex function, it is required that $\frac{\Lambda}{n} \ge \left \| \mathbf{n}_1 \right \|_2  \ge {\sigma}_{M_1} \sqrt{r}  $  (where $r$ is the minimum of $M$ and twice the upper bound on the rank of $\ell$'s Hessian) and  $\frac{\Lambda}{n} \ge \left \| \mathbf{n}_2 \right \|_2 \ge L{\sigma}_{M_2} $ for $\mathcal{L}(\hat{\bm{\omega}}; D)$. One can verify that by letting
\begin{equation*}
        \begin{aligned}
            \Lambda =\bm{ \Theta } \left( \frac{1}{ \left \| \hat{\bm{\omega}} \right \|_2} \left(  \frac{L \sqrt{rM}} { \epsilon} +  n \sqrt{\frac{L\gamma \sqrt{M}}{ \epsilon}}\right)  \right).
        \end{aligned}
\end{equation*}
it leads to 
\begin{equation*}
        \begin{aligned}
            \mathbb{E}\left[ \mathcal{L}(\bm{\omega}_{\mathrm{out}}; D) - \mathcal{L}(\hat{\bm{\omega}}; D) \right] =\bm{O} \left(  \frac{ \left \| \hat{\bm{\omega}} \right \|_2 L \sqrt{rM}} { n \epsilon} +  \left \| \hat{\bm{\omega}} \right \|_2 \sqrt{\frac{L\gamma \sqrt{M}}{ \epsilon}}\right).
        \end{aligned}
\end{equation*}

\section{Omitted Details in Section~\ref{sec:gradient-erm}}

\subsection{Pseudocode of product noise-based DPSGD}\label{app:pn-dpsgd}

Algorithm \ref{alg:PN-Based-SGD} presents the pseudocode of product noise-based DPSGD.
\begin{algorithm}[htp] 
	\caption{Product Noise-Based DPSGD}\label{alg:PN-Based-SGD}
	\begin{algorithmic}[0]
		\State{\bf Input:} Dataset: $D = \{d_1, \dots, d_n\}$; Empirical risk: $\mathcal{L}(\bm{\omega}; d_i) = \frac{1}{n}\sum_{i=1}^{n}  \ell(\bm{\omega}; d_i)$. Parameters: learning rate $ \eta_t$, subsampling probability $p$, noise scale $\sigma$, group size $I_t$, gradient norm bound $C$.
        \State \textbf{Initialize} model parameters \( \bm{\omega_0} \) randomly
        \begin{algorithmic}[1]
        \For {$t \in [T]$ } 
             \State Take a Poisson subsampling $I_t$ with subsampling probability $p$
             \State \textbf{Compute gradient: } For each $i \in I_t$, compute $\bm{g}_t(d_i) \gets \nabla_{\bm{\omega}_t} \mathcal{L} (\bm{\omega_t}, d_i)$
             \State \textbf{Clip gradient: } $\bm{\bar{g}}_t (d_i) \gets \bm{g}_t(d_i) / \max \left(1, \frac{\left \|\bm{g}_t(d_i)  \right \|_2 }{C}\right)$
             \State \textbf{Add noise: } $\bm{\tilde{g}}_t \gets  \frac{1}{I_t}  \left(\sum_{i \in I_t} \bm{\bar{g}}_t (d_i) +\mathbf{n}_t \right)$,  $\mathbf{n}_{t} = {\sigma_M}_t R_t \boldsymbol{h}_t$,
             \Statex \text{where } ${\sigma_M}_t = \frac{2\sqrt{2} C k^{\frac{2}{M}} \left(\frac{M}{4} +\frac{3}{2} \right)^{ \left( \frac{1}{2} + \frac{2}{M} \right)}}{\epsilon_t e^{\left( \frac{1}{2} + \frac{1}{M}\right)}}$, $R_t \sim \chi_1$ and $\boldsymbol{h}_t \sim \mathbb{S}^{M-1}$
             \State \textbf{Descent: } $\bm{\omega_{t+1}} = \bm{\omega_t} - \eta_t \bm{\tilde{g}}_t$
        \EndFor
        \end{algorithmic}
                {\bf Output: }$\bm{\omega_T}$ and compute the overall privacy cost $(\epsilon, \delta)$ using the  distribution-independent composition~\cite{he2021tighter}.
\end{algorithmic}
\end{algorithm}

\subsection{Proof of Corollary~\ref{privacy_guarantee_np_sgd}}\label{app:prove-pn-dpsgd}

We first consider the maximum difference between clipp\-ed gradient vectors with clipping norm $C$ for neighboring datasets $D$ and $D^\prime$ that differ in the value of the $i$-th individual. 
The difference reaches its maximum when $\bm{\tilde{g}}_t(d_i)$ and $\bm{\tilde{g}}_t(d^{\prime}_i)$ are in opposite directions. Thus, we obtain: 
\begin{equation*}
\begin{aligned}
    &\max_{d_i,d_i'}  \left \{ \left \| \sum_{i \in D}^{n} \bm{\tilde{g}}_t(d_i) - \sum_{i \in D^\prime}^{n} \bm{\tilde{g}}_t(d_i') \right \|_2 \right \} \\
     =&  \max_{d_i,d_i'}  \left \|  \bm{\tilde{g}}_t(d_i) -  \bm{\tilde{g}}_t(d^{\prime}_i) \right \|\\
     =& \sqrt{\left \|  \bm{\tilde{g}}_t(d_i) \right \|_2^2 + \left \| \bm{\tilde{g}}_t(d^{\prime}_i) \right \|_2^2 -  2\bm{\tilde{g}}_t(d_i)^{\top} \bm{\tilde{g}}_t(d^{\prime}_i)}  \\
     \leq& \sqrt{C^2 +C^2 +2 C \cdot C}\\
     =& 2C
\end{aligned}
\end{equation*}

Therefore, the $l_2$-sensitivity of product noise-based DPSGD is given by $2C$. Thus, for each step $t$, setting
\begin{equation*}
    \begin{aligned}
        {\sigma_M}_t =\frac{ 2\sqrt{2} C k^{\frac{2}{M}} \left(\frac{M}{4} +\frac{3}{2} \right)^{ \left( \frac{1}{2} + \frac{2}{M} \right)}}{\epsilon_t e^{\left( \frac{1}{2} + \frac{1}{M}\right)}},
    \end{aligned}
\end{equation*}
leads to $(\epsilon,\delta)$-DP guarantee for each DPSGD step according to Theorem~\ref{thm:mian-thm}.

\section{Omitted Details in Section~\ref{sec:experiments}}

\subsection{Experiment Setup for Convex ERM}\label{sec:convex-setup}

We adopt the loss function definitions for LR and Huber SVM from \cite{iyengar2019towards}, where $ z=  y \langle \omega, x \rangle$, the LR loss function is given by 
\begin{equation*}
    \begin{aligned}
        \ell \left ( \omega,(x, y) \right) = \log \left (1 + e^{-z } \right),
    \end{aligned}
\end{equation*}
and the Huber SVM loss function is defined as
\begin{equation*}
\ell_{\text{Huber}}(\omega; (x, y)) =
\begin{cases} 
    1 - z, & \text{if } 1 - z > h, \\
    0, & \text{if } 1 -z  < -h, \\
    \frac{(1 -z)^2}{4} + \frac{1 -z}{2} + \frac{h}{4}, & \text{otherwise},
\end{cases}
\end{equation*} 
and we set $h = 0.1$.

Details on dataset statistics used in the case studies for the covex cases, i.e., the number sample size, dimension and classes—are summarized in Table~\ref{tab:datasets_objective}. Among them, Synthetic-H is a high dimensional dataset generated using the script described in~\cite{iyengar2019towards}.
\begin{table}[htp]
    \centering
    \caption{Dataset Statistics}
    \label{tab:datasets_objective}
    \setlength{\tabcolsep}{10pt}
    \renewcommand{\arraystretch}{1.0} 
    \small
    \begin{tabular}{|>{\centering\arraybackslash}m{1.5cm}|
                    >{\centering\arraybackslash}m{1.3cm}|
                    >{\centering\arraybackslash}m{1.2cm}|
                    >{\centering\arraybackslash}m{1.2cm}|}
        \hline
        \textbf{Datasets} & \textbf{\# Samples} & \textbf{\# Dim} & \textbf{\# Classes} \\ 
        \hline
        Adult & 45,220 & 104 & 2 \\ 
        \hline
        KDDCup99 & 70,000 & 114 & 2 \\ 
        \hline
        MNIST & 65,000 & 784 & 10 \\ 
        \hline
        Synthetic-H & 70,000 & 20,000 & 2 \\ 
        \hline
        Real-sim & 72,309 & 20,958 & 2 \\ 
        \hline
        RCV1 & 50,000 & 47,236 & 2 \\ 
        \hline
    \end{tabular}
\end{table}

\subsection{Supplemental Experiments for Output Perturbation} \label{app:supplemental_exp_output}
Consistent with the analysis in Section~\ref{sec:case-study-output-perturbation}, product noise-based output perturbation also demonstrates superior performance on LR tasks across the remaining datasets.

\noindent \textbf{Supplemental Experiments for Output Perturbation Using Classic Gaussian Noise and Analytic Gaussian Noise on LR. }
As shown in Figure~\ref{fig:output_lr_other_datasets},  under the same privacy parameter $\epsilon = 10^{-2}$, the test accuracy of product noise method reaches $83.94\%$, $86.80\%$, $93.30\%$, and $92.70\%$ on the Adult, KDDCup99, Synthetic-H, and Real-sim datasets, respectively, while the corresponding accuracy using classic Gaussian noise are only $63.92\%$, $83.60\%$, $54.17\%$, and $55.18\%$, respectively; using analytic Gaussian noise  are only $70.13\%$, $86.34\%$, $62.36\%$, and $76.80\%$, respectively. With a higher privacy parameter ($\epsilon = 10^{-1}$), our method achieves the test accuracy on Synthetic-H and Real-sim to $95.65\%$ and $95.73\%$, exceeding the non-private baselines of $93.52\%$ and $93.85\%$, respectively. Even under a stringent privacy parameter ($\epsilon = 10^{-\frac{5}{2}}$), our method achieves $81.03\%$ on Synthetic-H, while the the one using while the one using classic Gaussian noise and analytic Gaussian noise achieve only $69.07\%$ and $76.60\%$, respectively, under a looser privacy parameter ($\epsilon = 10^{-\frac{3}{2}}$).
\begin{figure}[htp]
    \centering
    \begin{subfigure}{0.49\columnwidth}  
        \centering
        \includegraphics[width=\linewidth]{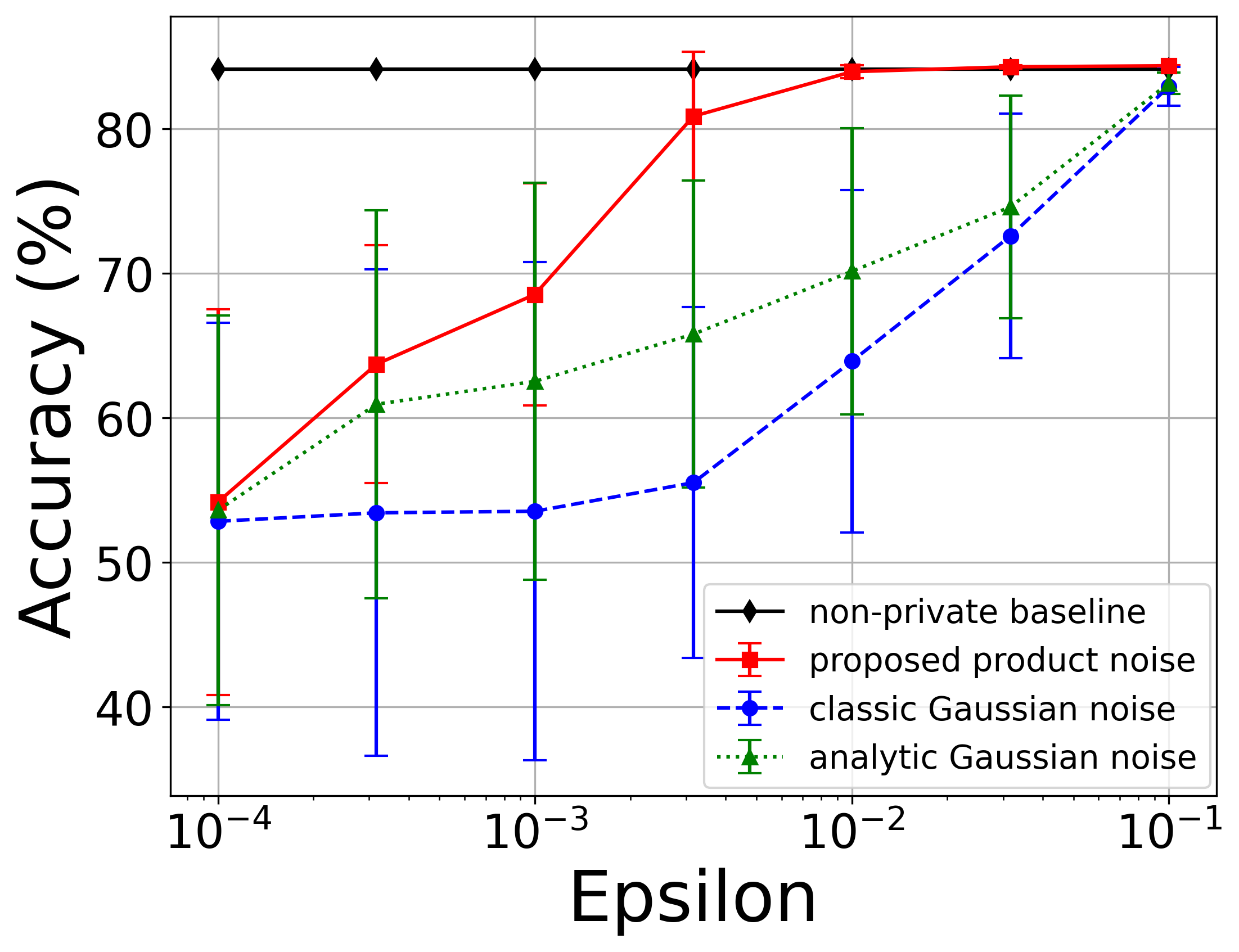}
        \Description{Accuracy results of output perturbation for Adult dataset.}
        \caption{\centering  Adult }
        \label{fig:adult_test_accuracy_output_lr}
    \end{subfigure}
    \hfill
    \begin{subfigure}{0.49\columnwidth}
        \centering
        \includegraphics[width=\linewidth]{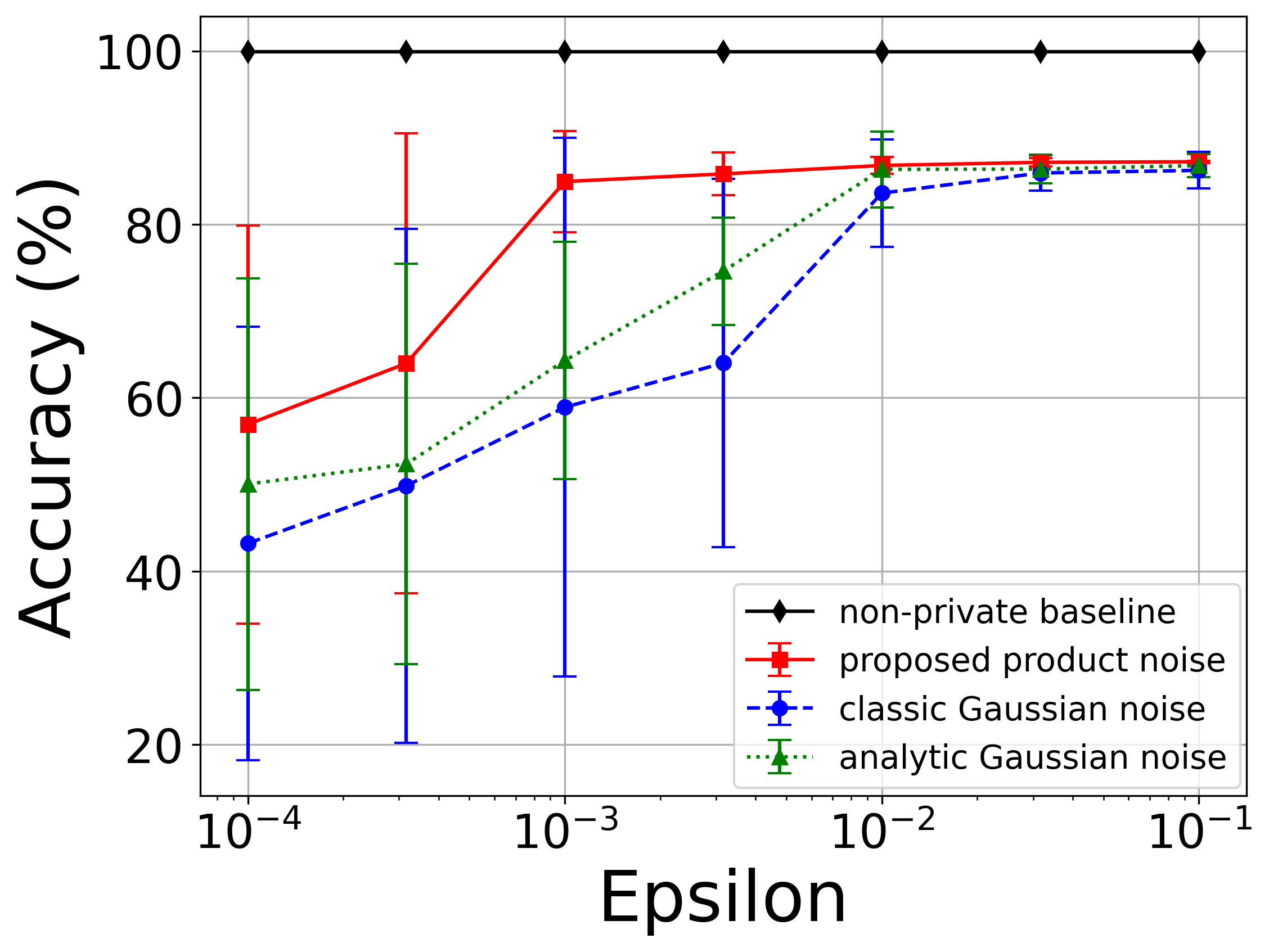}
        \Description{Accuracy results of output perturbation for KDDCup99 dataset.}
        \caption{\centering  KDDCup99 }
        \label{fig:KDDCup99_test_accuracy_output_lr}
    \end{subfigure}
    \hfill
    \begin{subfigure}{0.49\columnwidth}
        \centering
        \includegraphics[width=\linewidth]{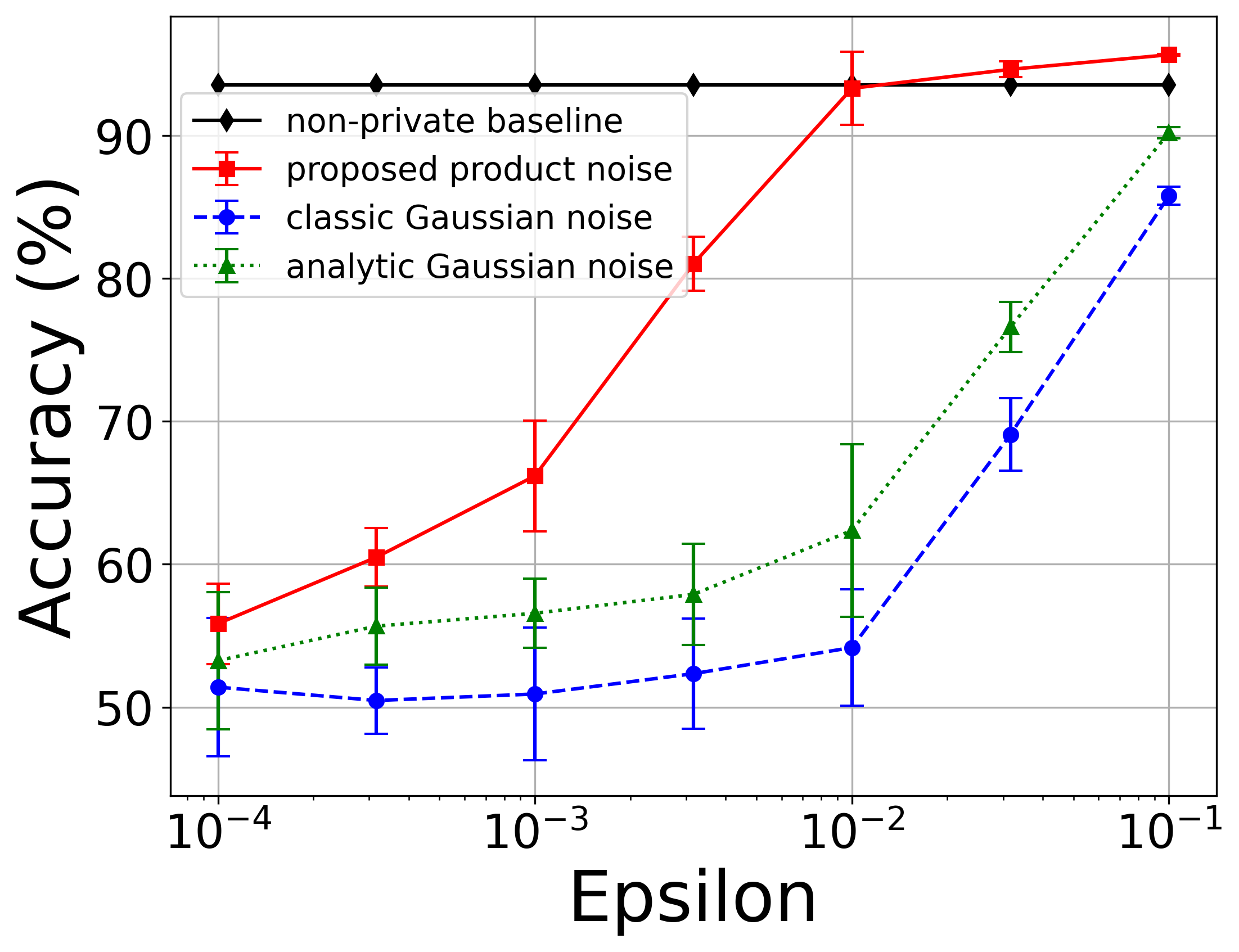}
        \Description{Accuracy results of output perturbation for Synthetic-H dataset.}
        \caption{\centering  Synthetic-H }
        \label{fig:synthetich_test_accuracy_output_lr}
    \end{subfigure}
    \hfill
    \begin{subfigure}{0.49\columnwidth}
        \centering
        \includegraphics[width=\linewidth]{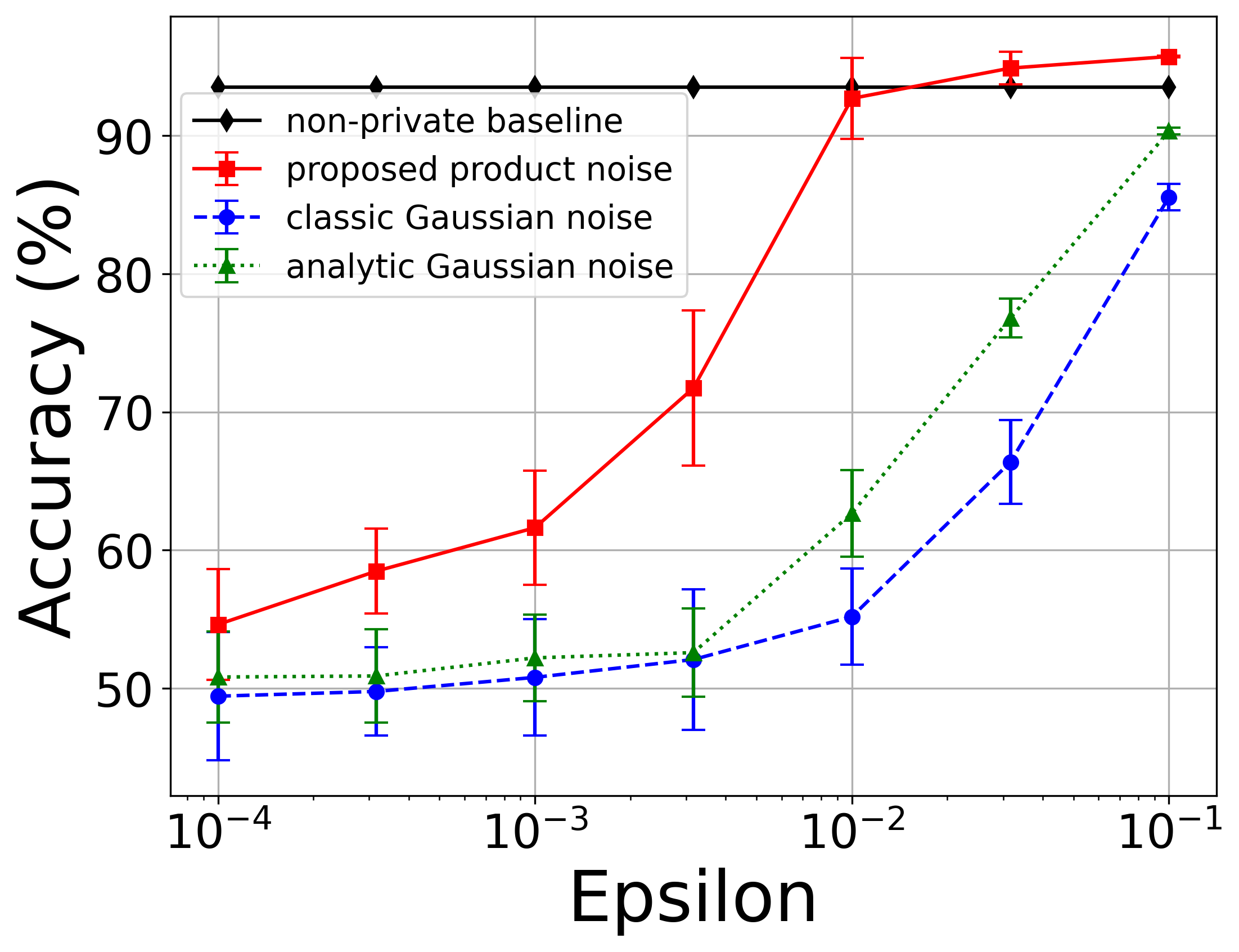}
        \Description{Accuracy results of output perturbation for Real-sim dataset  .}
        \caption{\centering Real-sim }
        \label{fig:realsim_test_accuracy_output_lr}
    \end{subfigure}
  \caption{Test accuracy of Output Perturbation on LR (Adult, KDDCup99, Synthetic-H and Real-sim). Product noise v.s. classic Gaussian noise and analytic Gaussian noise. The experiment setups, e.g., iteration and solver,  follow ~\cite{chaudhuri2011differentially}.}
    \label{fig:output_lr_other_datasets}
\end{figure}

In terms of utility stability, product noise method also exhibits lower standard deviation. For instance, on the Real-sim dataset with $\epsilon = 10^{-\frac{3}{2}}$, the standard deviation of test accuracy under the product noise is 0.01179, significantly lower than that of classic Gaussian noise (0.03025) and analytic Gaussian noise (0.01401).

\begin{table}[htp]
  \caption{$\ell_2$ error of Output Perturbation on LR.}
  \label{tab:l2_error_output_lr}
  \centering
  \footnotesize
  \setlength{\tabcolsep}{4pt}
  \renewcommand{\arraystretch}{1.0}
  \resizebox{\columnwidth}{!}{%
  \begin{tabular}{c|c|c|c|c|c}
    \toprule
    \multirow{2}{*}{\textbf{Dataset}} &
    \multirow{2}{*}{\textbf{Mechanism}} &
    \multicolumn{4}{c}{\textbf{$\boldsymbol{\epsilon}$}} \\
    \cline{3-6}
     &  & {\cellvcenter $\mathbf{10^{-4}}$}
        & {\cellvcenter $\mathbf{10^{-3}}$}
        & {\cellvcenter $\mathbf{10^{-2}}$}
        & {\cellvcenter $\mathbf{10^{-1}}$} \\
    \midrule
    \multirow[c]{3}{*}{Adult} 
      & classic   & 2.7$\times10^{3}$ & 2.8$\times10^{2}$ & 28 & 3 \\
      & analytic  & 5.2$\times10^{2}$ & 97 & 14 & 2 \\
      & ours      & \textbf{2.6\bm{$\times10^{2}$}} & \textbf{31} & \textbf{2} & \textbf{3\bm{$\times10^{-1}$}} \\
    \midrule
    \multirow[c]{3}{*}{KDDCup99} 
      & classic   & 1.8$\times10^{3}$ & 1.8$\times10^{2}$ & 19 & 2 \\
      & analytic  & 3.7$\times10^{2}$ & 66 & 8 & 1 \\
      & ours      & \textbf{87} & \textbf{18} & \textbf{2} & \textbf{2\bm{$\times10^{-1}$}} \\
    \midrule
    \multirow[c]{3}{*}{MNIST} 
      & classic  & 4.9$\times10^{5}$ & 4.8$\times10^{4}$ & 4.9$\times10^{3}$ & 4.8$\times10^{2}$ \\
      & analytic & 9.4$\times10^{4}$ & 1.7$\times10^{4}$ & 2.5$\times10^{3}$ & 3.1$\times10^{2}$ \\
      & ours     & \textbf{3.2\bm{$\times10^{4}$}} & \textbf{4.0\bm{$\times10^{3}$}} & \textbf{3.4\bm{$\times10^{2}$}} & \textbf{46} \\
      \midrule
    \multirow[c]{3}{*}{Synthetic-H} 
      & classic   & 2.4$\times10^{6}$ & 2.4$\times10^{5}$ & 2.4$\times10^{4}$ & 2.4$\times10^{3}$ \\
      & analytic  & 4.7$\times10^{5}$ & 8.6$\times10^{4}$ & 1.2$\times10^{4}$ & 1.5$\times10^{3}$ \\
      & ours      & \textbf{1.5\bm{$\times10^{5}$}} & \textbf{9.9\bm{$\times10^{3}$}} & \textbf{2.2\bm{$\times10^{3}$}} & \textbf{1.7\bm{$\times10^2$}} \\
    \midrule
    \multirow[c]{3}{*}{Real-sim} 
      & classic   & 2.4$\times10^{6}$ & 2.4$\times10^{5}$ & 2.4$\times10^{4}$ & 2.4$\times10^{3}$ \\
      & analytic  & 4.7$\times10^{5}$ & 8.6$\times10^{4}$ & 1.2$\times10^{4}$ & 1.5$\times10^{3}$ \\
      & ours      & \textbf{1.6\bm{$\times10^{5}$}} & \textbf{1.3\bm{$\times10^{4}$}} & \textbf{1.3\bm{$\times10^{3}$}} & \textbf{92} \\
      \midrule
    \multirow[c]{3}{*}{RCV1} 
      & classic  & 5.3$\times10^{6}$ & 5.3$\times10^{5}$ & 5.3$\times10^{4}$ & 5.3$\times10^{3}$ \\
      & analytic & 1.0$\times10^{6}$ & 1.9$\times10^{5}$ & 2.7$\times10^{4}$ & 3.3$\times10^{3}$ \\
      & ours     & \textbf{3.2\bm{$\times10^{5}$}} & \textbf{5.0\bm{$\times10^{4}$}} & \textbf{2.8\bm{$\times10^{3}$}} & \textbf{3.2\bm{$\times10^{2}$}} \\
    \bottomrule
  \end{tabular}}
\end{table}

As shown in Table~\ref{tab:l2_error_output_lr}, we quantify the $\ell_2$ error between the private and non-private models. It is clear that our product noise consistently achieves significantly smaller $\ell_2$ errors than both the classic Gaussian noise and analytic Gaussian noise under the same privacy parameters.

Additionally, we also evaluate model robustness using FPR in Table~\ref{tab:fpr_output_lr}. 
It is clear that, for all given privacy parameters, our method consistently yields lower FPRs, indicating fewer misclassifications of negative samples as positive. 
\begin{table}[htp]
  \centering
  \caption{FPR of Output Perturbation on LR.}
  \label{tab:fpr_output_lr}
  \footnotesize
  \setlength{\tabcolsep}{4pt}
  \renewcommand{\arraystretch}{1.0}
  \resizebox{\columnwidth}{!}{%
  \begin{tabular}{c|c|c|c|c|c}
    \toprule
    \multirow{2}{*}{\textbf{Dataset}} &
    \multirow{2}{*}{\textbf{Mechanism}} &
    \multicolumn{4}{c}{\textbf{$\boldsymbol{\epsilon}$}} \\ 
    \cline{3-6}
     &  & {\cellvcenter $\mathbf{10^{-4}}$}
        & {\cellvcenter $\mathbf{10^{-3}}$}
        & {\cellvcenter $\mathbf{10^{-2}}$}
        & {\cellvcenter $\mathbf{10^{-1}}$} \\
    \midrule
    \multirow[c]{3}{*}{Adult} 
      & classic  & 0.548$\pm$0.177 & 0.496$\pm$0.295 & 0.418$\pm$0.279 & 0.080$\pm$0.045 \\
      & analytic & 0.528$\pm$0.334 & 0.535$\pm$0.336 & 0.306$\pm$0.259 & 0.069$\pm$0.033 \\
      & ours     & \textbf{0.499$\pm$0.160} & \textbf{0.414$\pm$0.224} & \textbf{0.060$\pm$0.019} & \textbf{0.060$\pm$0.004} \\
    \midrule
    \multirow[c]{3}{*}{KDDCup99} 
      & classic  & 0.507$\pm$0.049 & 0.524$\pm$0.057 & 0.514$\pm$0.047 & 0.460$\pm$0.027 \\
      & analytic & 0.510$\pm$0.063 & 0.484$\pm$0.049 & 0.485$\pm$0.067 & 0.401$\pm$0.054 \\
      & ours     & \textbf{0.277$\pm$0.028} & \textbf{0.226$\pm$0.013} & \textbf{0.148$\pm$0.014} & \textbf{0.006$\pm$0.003} \\
      \midrule
    \multirow[c]{3}{*}{MNIST} 
      & classic  & 0.055$\pm$0.022 & 0.072$\pm$0.011 & 0.068$\pm$0.017 & 0.055$\pm$0.016 \\
      & analytic & 0.070$\pm$0.023 & 0.060$\pm$0.020 & 0.072$\pm$0.021 & 0.054$\pm$0.021 \\
      & ours & \textbf{0.049$\pm$0.023} & \textbf{0.055$\pm$0.019} & \textbf{0.064$\pm$0.026} & \textbf{0.022$\pm$0.010} \\
      \midrule
    \multirow[c]{3}{*}{Synthetic} 
      & classic  & 0.447$\pm$0.271 & 0.466$\pm$0.400 & 0.315$\pm$0.303 & 0.212$\pm$0.175 \\
      & analytic & 0.496$\pm$0.296 & 0.698$\pm$0.168 & 0.410$\pm$0.321 & 0.083$\pm$0.080 \\
      & ours     & \textbf{0.409$\pm$0.216} & \textbf{0.363$\pm$0.143} & \textbf{0.139$\pm$0.150} & \textbf{0.039$\pm$0.011} \\
    \midrule
    \multirow[c]{3}{*}{Real-sim} 
      & classic  & 0.632$\pm$0.296 & 0.368$\pm$0.181 & 0.617$\pm$0.333 & 0.162$\pm$0.144 \\
      & analytic & 0.371$\pm$0.228 & 0.457$\pm$0.316 & 0.402$\pm$0.302 & 0.052$\pm$0.044 \\
      & ours     & \textbf{0.291$\pm$0.125} & \textbf{0.337$\pm$0.168} & \textbf{0.031$\pm$0.024} & \textbf{0.038$\pm$0.008} \\
    \midrule
    \multirow[c]{3}{*}{RCV1} 
      & classic  & 0.648$\pm$0.336 & 0.605$\pm$0.256 & 0.368$\pm$0.265 & 0.262$\pm$0.235 \\
      & analytic & 0.456$\pm$0.331 & 0.507$\pm$0.230 & 0.308$\pm$0.181 & 0.309$\pm$0.257 \\
      & ours & \textbf{0.428$\pm$0.246} & \textbf{0.399$\pm$0.171} & \textbf{0.173$\pm$0.131} & \textbf{0.081$\pm$0.022} \\
    \bottomrule
  \end{tabular}}
\end{table}

\noindent \textbf{Supplemental Experiments for Output Perturbation Using Classic Gaussian Noise and Analytic Gaussian Noise on Huber SVM. }
As shown in Figure~\ref{fig:output_svm}, under the same privacy parameter ($\epsilon = 10^{-2}$), on Adult, KDDCup99, MNIST, Synthetic-H, Real-sim, and RCV1, our method achieves test accuracy of $78.01\%$, $96.15\%$, $85.39\%$, $65.79\%$, $64.06\%$, and $92.24\%$, respectively, significantly higher than the one using classic Gaussian noise ($61.08\%$, $63.31\%$, $75.38\%$, $52.35\%$, $50.50\%$, and $85.16\%$), and the one using analytic Gaussian noise ($65.78\%$, $80.04\%$, $81.14\%$, $61.43\%$, $61.43\%$, and $88.75\%$). Moreover,
Our method yields test accuracy comparable to or exceeding non-private baselines at higher privacy parameters (e.g., $\epsilon = 10^{-1}$). It achieves $91.88\%$ on MNIST and $93.11\%$ on RCV1, outperforming their respective baselines ($91.84\%$ and $91.85\%$). Even under stringent privacy parameters (e.g., $\epsilon = 10^{-2.5}$), it maintains high test accuracy, such as $91.20\%$ on RCV1, whereas the classic Gaussian mechanism under more relaxed privacy (e.g., $\epsilon = 10^{-1}$) reaches only $91.85\%$.

Regarding utility stability, our method demonstrates lower standard deviation of test accuracy. For example, on Real-sim with $\epsilon = 10^{-2}$, it achieves a standard deviation of 0.04046, compared to 0.12991 with the classic Gaussian noise and 0.05620 with analytic Gaussian noise, indicating a more stable utility.
\begin{figure}[htp]
    \centering
    \begin{subfigure}{0.49\columnwidth}  
        \centering
        \includegraphics[width=\linewidth]{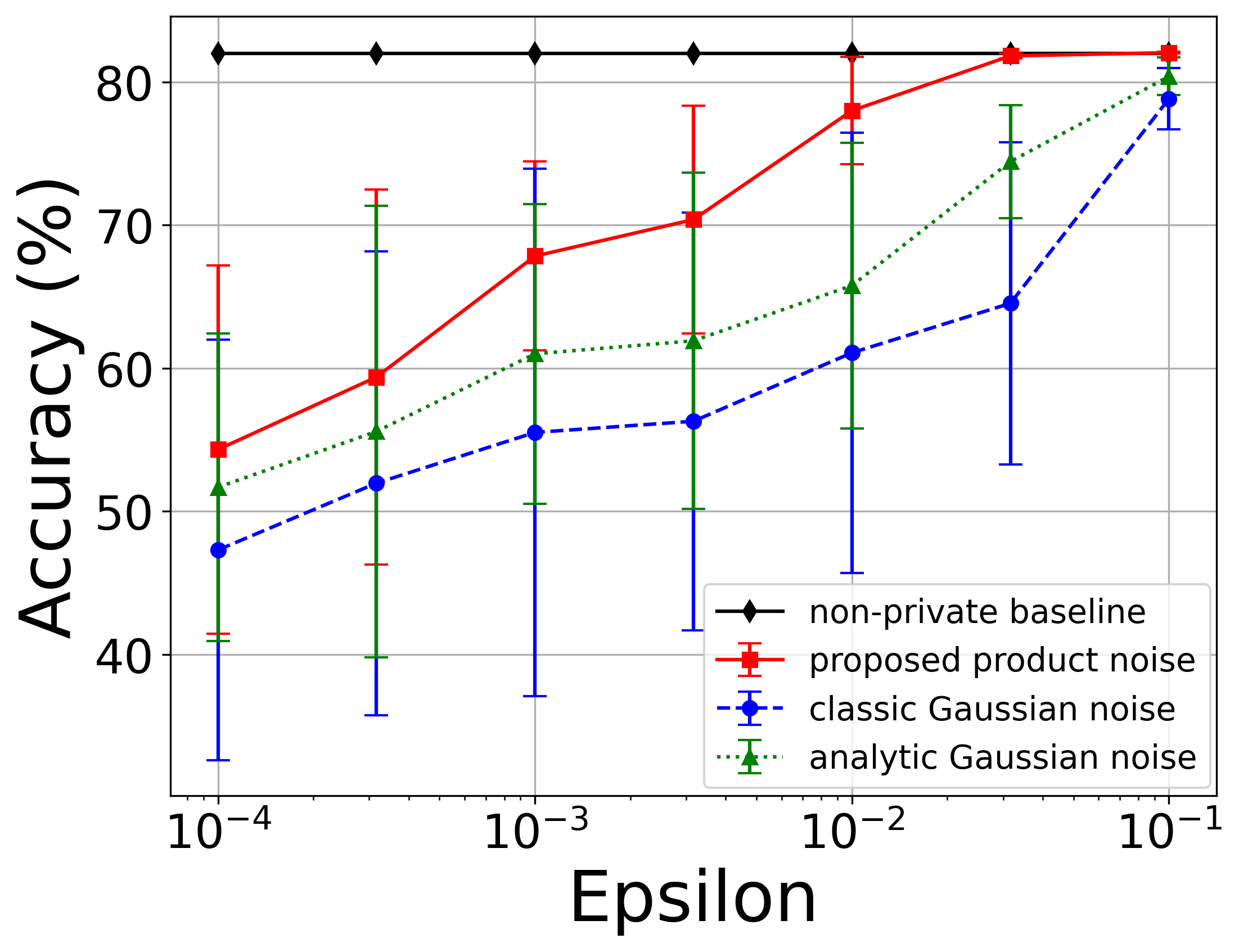}
        \Description{Accuracy results of output perturbation for Adult dataset   .}
        \caption{\centering  Adult}
        \label{fig:adult_test_accuracy_output_svm}
    \end{subfigure}
    \hfill
    \begin{subfigure}{0.49\columnwidth}
        \centering
        \includegraphics[width=\linewidth, height=3.2cm]{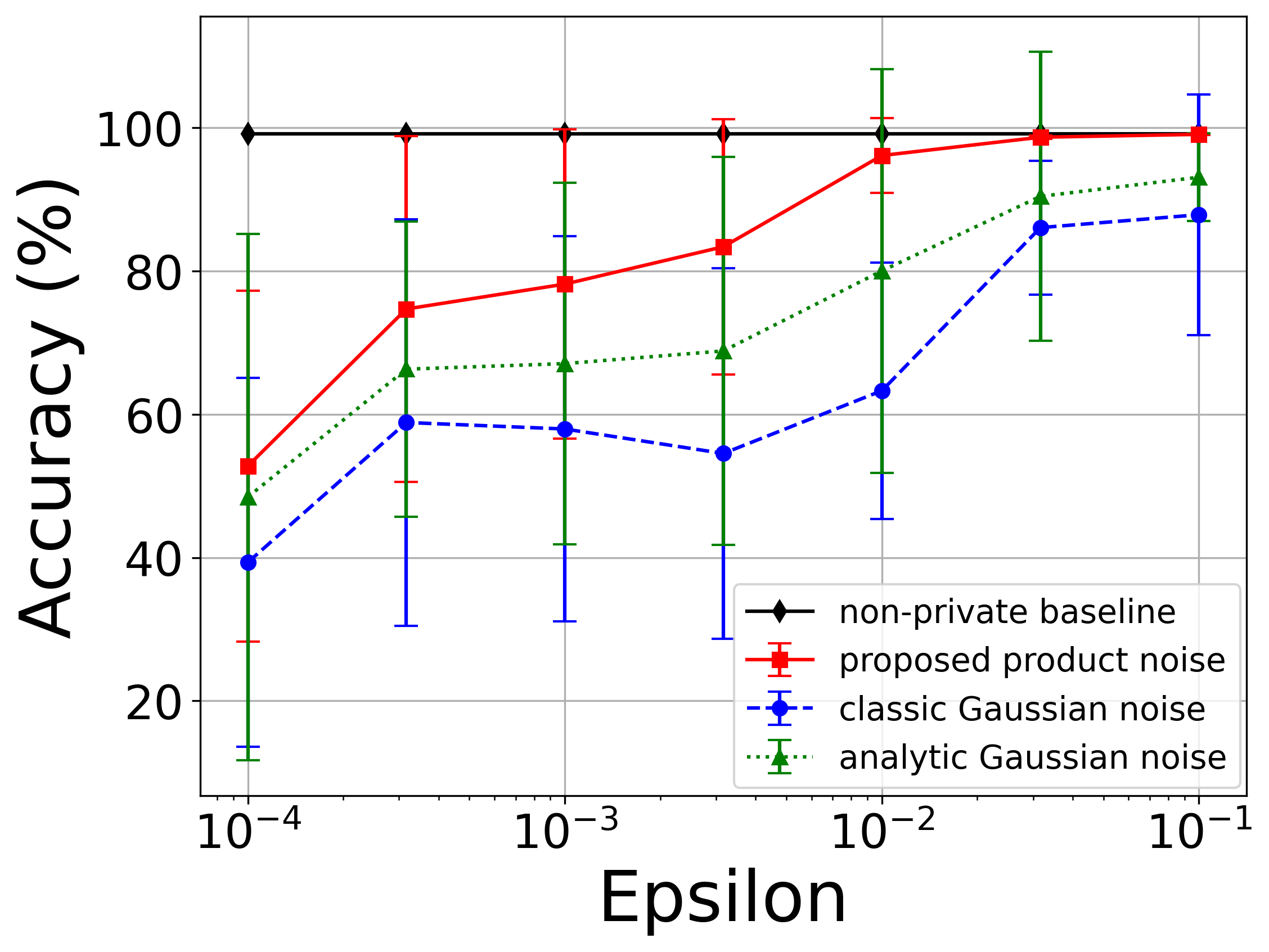}
        \Description{Accuracy results of output perturbation for KDDCup99 dataset   .}
        \caption{\centering  KDDCup99 }
        \label{fig:KDDCup99_test_accuracy_output_svm}
    \end{subfigure}
     \hfill
    \begin{subfigure}{0.49\columnwidth}
        \centering
        \includegraphics[width=\linewidth]{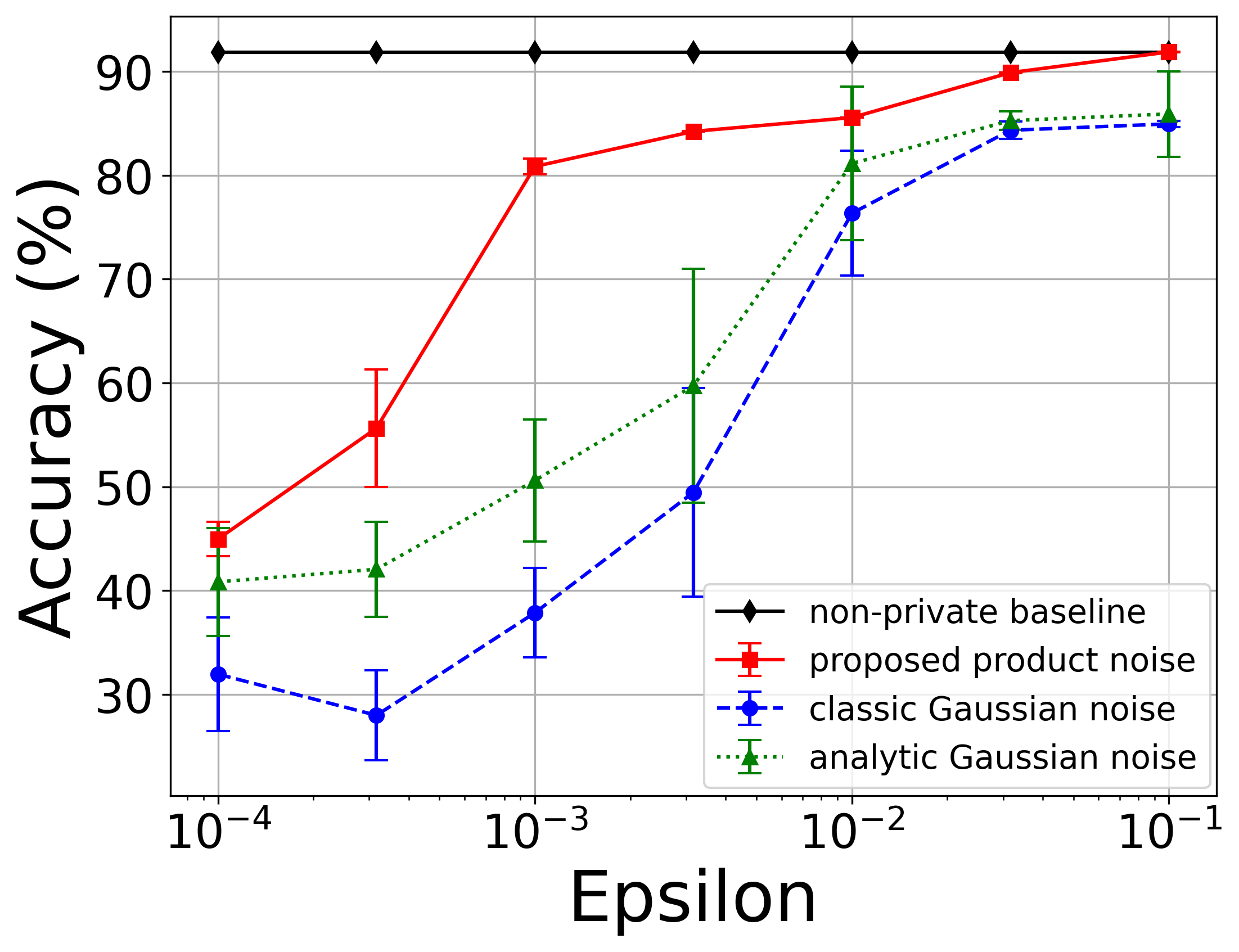}
        \Description{Accuracy results of output perturbation for MNIST dataset   .}
        \caption{\centering  MNIST }
        \label{fig:mnist_test_accuracy_output_svm}
    \end{subfigure}
     \hfill
    \begin{subfigure}{0.49\columnwidth}
        \centering
        \includegraphics[width=\linewidth]{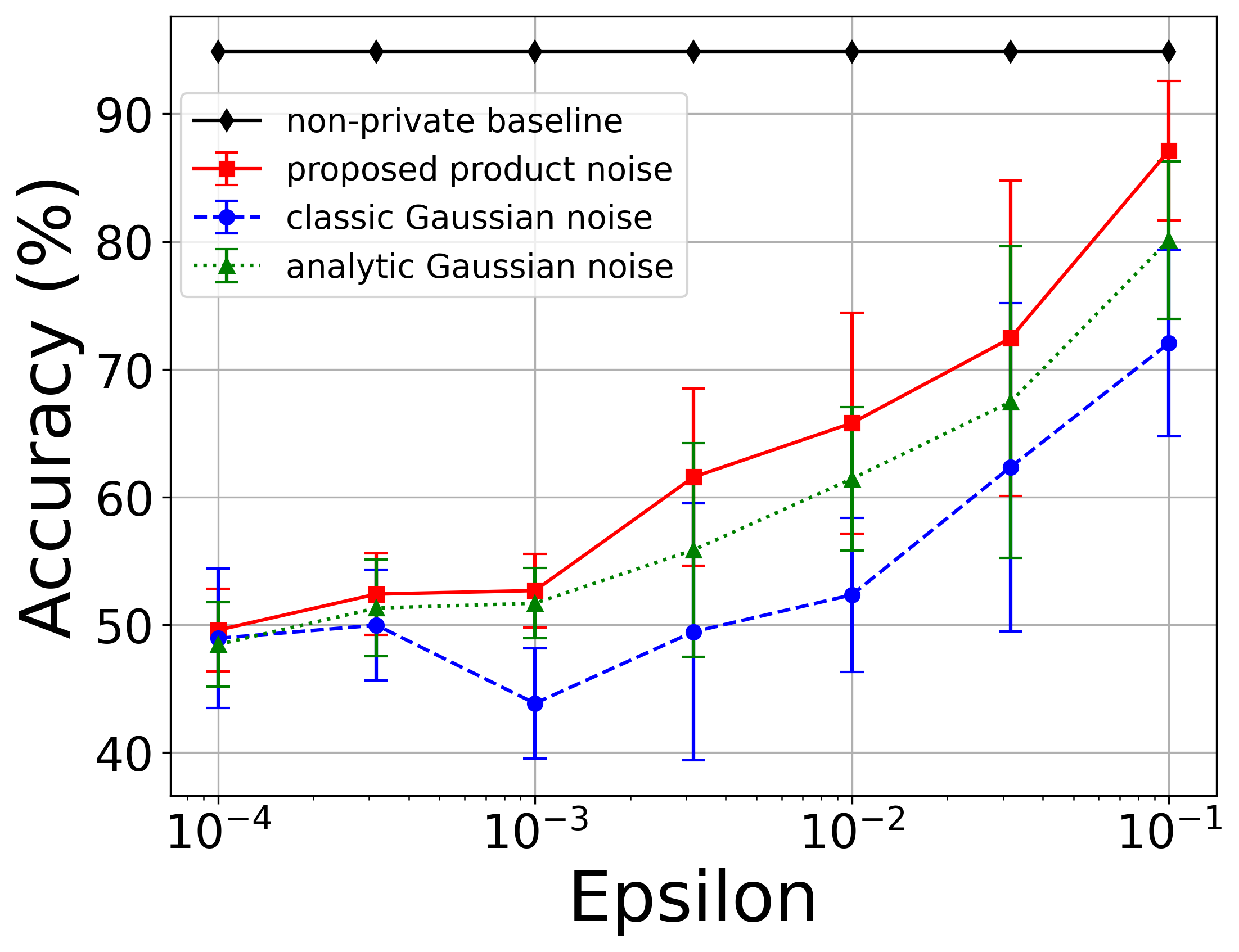}
        \Description{Accuracy results of output perturbation for Synthetic-H dataset  .}
        \caption{\centering  Synthetic-H }
        \label{fig:synthetich_test_accuracy_output_svm}
    \end{subfigure}
     \hfill
    \begin{subfigure}{0.49\columnwidth}
        \centering
        \includegraphics[width=\linewidth]{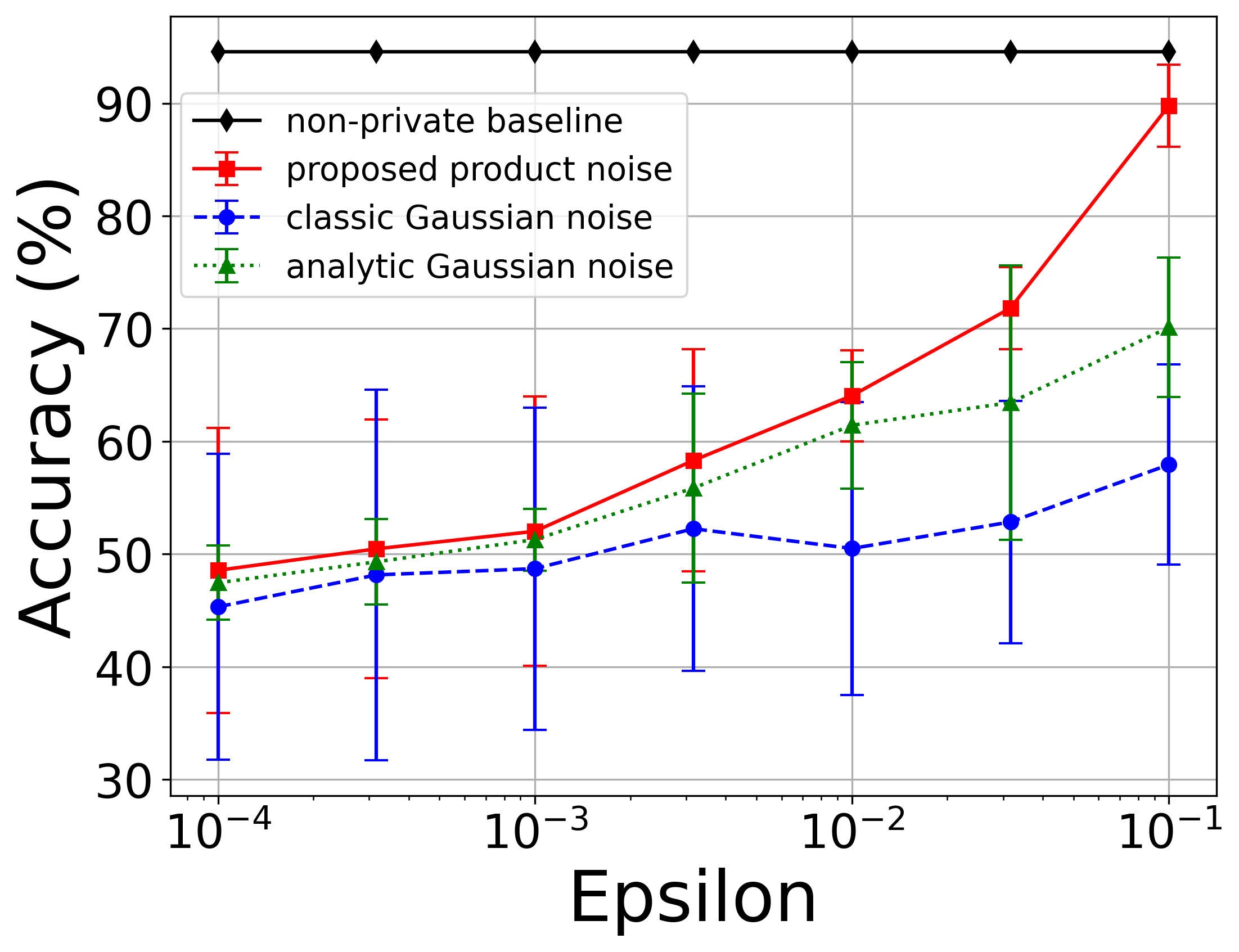}
        \Description{Accuracy results of output perturbation for Real-sim dataset  .}
        \caption{\centering Real-sim  }
        \label{fig:realsim_test_accuracy_output_svm}
    \end{subfigure}
    \hfill
    \begin{subfigure}{0.49\columnwidth}
        \centering
        \includegraphics[width=\linewidth]{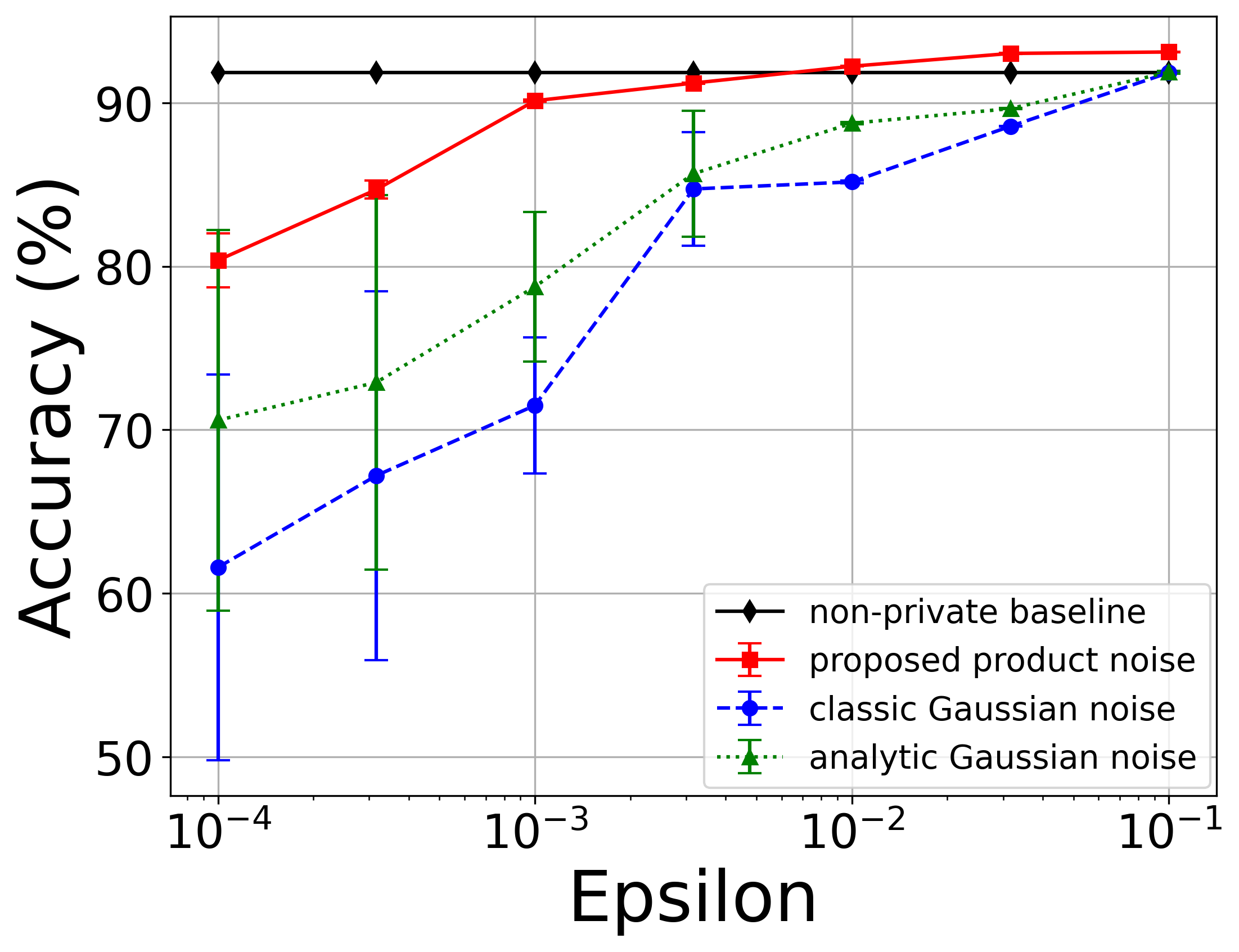}
        \Description{Accuracy results of output perturbation for RCV1 dataset  .}
        \caption{\centering RCV1 }
        \label{fig:rcv1_test_accuracy_output_svm}
    \end{subfigure}
  \caption{Test accuracy  of Output Perturbation on Huber SVM. Product noise v.s. classic Gaussian noise and  analytic Gaussian noise. The experiment setups, e.g., iteration and solver,  follow ~\cite{chaudhuri2011differentially}.} 
    \label{fig:output_svm}
\end{figure}

We also report $\ell_2$ error and FPR for Huber SVM in Table~\ref{tab:l2_error_SVM} and Table~\ref{tab:fpr_SVM}, respectively. The experiment results reach the same conclusion as the earlier LR experiments: product noise consistently exhibits smaller $\ell_2$ errors, and achieves lower FPRs, confirming superior utility and robustness across settings.
\begin{table}[htp]
  \centering
  \caption{$\ell_2$ error of Output Perturbation on Huber SVM.}
  \label{tab:l2_error_SVM}
  \footnotesize
  \setlength{\tabcolsep}{4pt}
  \renewcommand{\arraystretch}{1.0}
  \resizebox{\columnwidth}{!}{%
  \begin{tabular}{c|c|c|c|c|c}
    \toprule
    \multirow{2}{*}{\textbf{Dataset}} &
    \multirow{2}{*}{\textbf{Mechanism}} &
    \multicolumn{4}{c}{\textbf{$\boldsymbol{\epsilon}$}} \\ 
    \cline{3-6}
     &  & {\cellvcenter $\mathbf{10^{-4}}$}
        & {\cellvcenter $\mathbf{10^{-3}}$}
        & {\cellvcenter $\mathbf{10^{-2}}$}
        & {\cellvcenter $\mathbf{10^{-1}}$} \\
    \midrule
    \multirow[c]{3}{*}{Adult} 
      & classic  & 2.7$\times10^3$ & 2.6$\times10^2$ & 27.3 & 2.7 \\
      & analytic & 5.5$\times10^2$ & 98.1 & 13.9 & 1.8 \\
      & ours     & \textbf{2.4\bm{$\times10^2$}} & \textbf{16.1} & \textbf{2.7} & \textbf{2.0\bm{$\times10^{-2}$}} \\
    \midrule
    \multirow[c]{3}{*}{KDDCup99} 
      & classic  & 1.8$\times10^3$ & 1.9$\times10^3$ & 18.5 & 1.9 \\
      & analytic & 3.5$\times10^3$  & 66.6  & 9.4  & 1.2 \\
      & ours     & \textbf{1.3\bm{$\times10^3$}} & \textbf{22.7} & \textbf{1.7} & \textbf{1.6\bm{$\times10^{-1}$}} \\
    \midrule
    \multirow[c]{3}{*}{MNIST} 
      & classic  & 4.8$\times10^6$ & 4.9$\times10^5$ & 4.8$\times10^4$ & 4.9$\times10^3$ \\
      & analytic & 9.4$\times10^5$ & 1.7$\times10^5$ & 2.4$\times10^4$ & 3.1$\times10^3$ \\
      & ours     & \textbf{3.4\bm{$\times10^5$}} & \textbf{3.3\bm{$\times10^4$}} & \textbf{3.8\bm{$\times10^3$}} & \textbf{3.4\bm{$\times10^2$}} \\
    \midrule
    \multirow[c]{3}{*}{Synthetic-H} 
      & classic  & 2.4$\times10^6$ & 2.4$\times10^5$ & 2.4$\times10^4$ & 2.4$\times10^3$ \\
      & analytic & 4.7$\times10^5$ & 8.6$\times10^4$ & 1.2$\times10^4$ & 1.5$\times10^3$ \\
      & ours     & \textbf{2.7\bm{$\times10^5$}} & \textbf{2.1\bm{$\times10^4$}} & \textbf{2.2\bm{$\times10^3$}} & \textbf{1.8\bm{$\times10^2$}} \\
    \midrule
    \multirow[c]{3}{*}{Real-sim} 
      & classic  & 2.4$\times10^6$ & 2.4$\times10^5$ & 2.4$\times10^4$ & 2.4$\times10^3$ \\
      & analytic & 4.7$\times10^5$ & 8.6$\times10^4$ & 1.2$\times10^4$ & 1.5$\times10^3$ \\
      & ours     & \textbf{1.8\bm{$\times10^5$}} & \textbf{2.0\bm{$\times10^4$}} & \textbf{8.8\bm{$\times10^2$}} & \textbf{2.4\bm{$\times10^2$}} \\
    \midrule
    \multirow[c]{3}{*}{RCV1} 
      & classic  & 5.3$\times10^5$ & 5.3$\times10^4$ & 5.3$\times10^3$ & 5.3$\times10^2$ \\
      & analytic & 1.0$\times10^5$ & 1.9$\times10^4$ & 2.6$\times10^3$ & 3.3$\times10^2$ \\
      & ours     & {\textbf{4.3$\bm{\times10^4}$}} & \textbf{4.2\bm{$\times10^3$}} & \textbf{3.8\bm{$\times10^2$}} & \textbf{35.2} \\
    \bottomrule
  \end{tabular}}
\end{table}

\begin{table}[htp]
  \centering
  \caption{FPR of Output Perturbation on Huber SVM.}
  \label{tab:fpr_SVM}
  \footnotesize
  \setlength{\tabcolsep}{4pt}
  \renewcommand{\arraystretch}{1.0}
  \resizebox{\columnwidth}{!}{%
  \begin{tabular}{c|c|c|c|c|c}
    \toprule
    \multirow{2}{*}{\textbf{Dataset}} &
    \multirow{2}{*}{\textbf{Mechanism}} &
    \multicolumn{4}{c}{\textbf{$\boldsymbol{\epsilon}$}} \\ 
    \cline{3-6}
     &  & {\cellvcenter $\mathbf{10^{-4}}$}
        & {\cellvcenter $\mathbf{10^{-3}}$}
        & {\cellvcenter $\mathbf{10^{-2}}$}
        & {\cellvcenter $\mathbf{10^{-1}}$} \\
    \midrule
    \multirow[c]{4}{*}{MNIST} 
      & classic  & 0.419$\pm$0.264 & 0.492$\pm$0.312 & 0.397$\pm$0.293 & 0.101$\pm$0.062 \\
      & analytic & 0.574$\pm$0.232 & 0.462$\pm$0.262 & 0.522$\pm$0.273 & 0.110$\pm$0.076 \\
      & ours     & \textbf{0.405$\pm$0.135} & \textbf{0.358$\pm$0.260} & \textbf{0.093$\pm$0.049} & \textbf{0.062$\pm$0.007} \\
      & baseline & 0.063 & 0.063 & 0.063 & 0.063 \\
    \midrule
    \multirow[c]{4}{*}{Synthetic-H} 
      & classic  & 0.356$\pm$0.251 & 0.546$\pm$0.339 & 0.419$\pm$0.310 & 0.079$\pm$0.091 \\
      & analytic & 0.553$\pm$0.336 & 0.462$\pm$0.402 & 0.177$\pm$0.202 & 0.035$\pm$0.027 \\
      & ours     & \textbf{0.266$\pm$0.225} & \textbf{0.434$\pm$0.275} & \textbf{0.147$\pm$0.112} & \textbf{0.013$\pm$0.002} \\
      & baseline & 0.010 & 0.010 & 0.010 & 0.010 \\
    \midrule
    \multirow[c]{4}{*}{Adult} 
      & classic  & 0.099$\pm$0.003 & 0.101$\pm$0.003 & 0.101$\pm$0.002 & 0.096$\pm$0.004 \\
      & analytic & 0.100$\pm$0.003 & 0.100$\pm$0.004 & 0.099$\pm$0.003 & 0.099$\pm$0.005 \\
      & ours     & \textbf{0.080$\pm$0.001} & \textbf{0.089$\pm$0.001} & \textbf{0.091$\pm$0.001} & \textbf{0.071$\pm$0.001} \\
      & baseline & 0.009 & 0.009 & 0.009 & 0.009 \\
    \midrule
    \multirow[c]{4}{*}{Real-sim} 
      & classic  & 0.489$\pm$0.078 & 0.538$\pm$0.076 & 0.465$\pm$0.069 & 0.460$\pm$0.051 \\
      & analytic & 0.497$\pm$0.088 & 0.491$\pm$0.070 & 0.483$\pm$0.059 & 0.389$\pm$0.066 \\
      & ours     & \textbf{0.411$\pm$0.061} & \textbf{0.468$\pm$0.068} & \textbf{0.392$\pm$0.019} & \textbf{0.089$\pm$0.019} \\
      & baseline & 0.033 & 0.033 & 0.033 & 0.033 \\
    \midrule
    \multirow[c]{4}{*}{KDDCup99} 
      & classic  & 0.527$\pm$0.064 & 0.508$\pm$0.071 & 0.491$\pm$0.078 & 0.458$\pm$0.073 \\
      & analytic & 0.503$\pm$0.082 & 0.485$\pm$0.045 & 0.511$\pm$0.051 & 0.420$\pm$0.059 \\
      & ours     & \textbf{0.420$\pm$0.057} & \textbf{0.393$\pm$0.027} & \textbf{0.319$\pm$0.032} & \textbf{0.095$\pm$0.037} \\
      & baseline & 0.033 & 0.033 & 0.033 & 0.033 \\
    \midrule
    \multirow[c]{4}{*}{RCV1} 
      & classic  & 0.477$\pm$0.055 & 0.478$\pm$0.037 & 0.438$\pm$0.062 & 0.311$\pm$0.079 \\
      & analytic & 0.532$\pm$0.050 & 0.455$\pm$0.054 & 0.506$\pm$0.069 & 0.306$\pm$0.054 \\
      & ours     & \textbf{0.438$\pm$0.045} & \textbf{0.384$\pm$0.024} & \textbf{0.254$\pm$0.041} & \textbf{0.056$\pm$0.001} \\
      & baseline & 0.045 & 0.045 & 0.045 & 0.045 \\
    \bottomrule
  \end{tabular}}
\end{table}

\begin{figure}[htp]
    \centering
    \begin{subfigure}{0.49\columnwidth}  
        \centering
        \includegraphics[width=\linewidth]{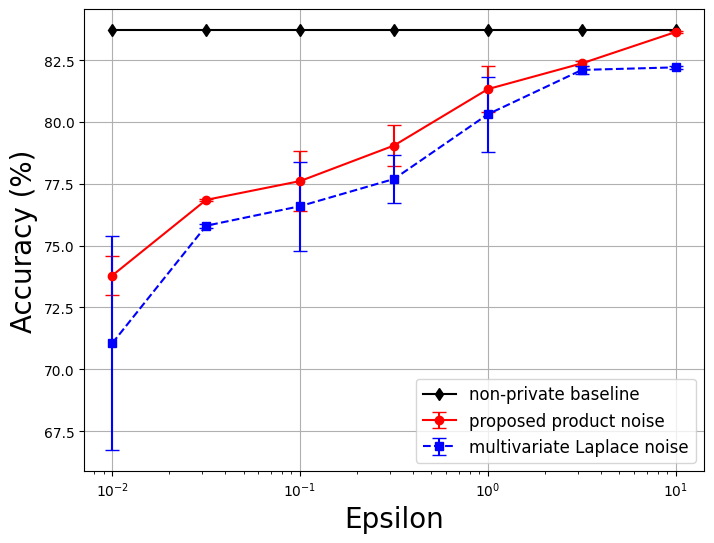}
        \Description{Accuracy results of Private Strongly Convex Permutation-based for Adult dataset   .}
        \caption{\centering  Adult  }
        \label{fig:adult_test_accuracy_ppsgd_lr}
    \end{subfigure}
    \hfill
        \begin{subfigure}{0.49\columnwidth}
        \centering
        \includegraphics[width=\linewidth]{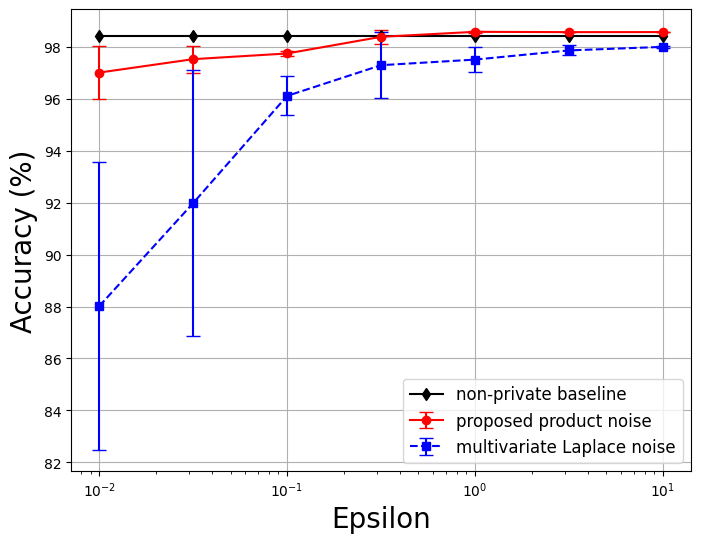}
        \Description{Accuracy results of  Private Strongly Convex Permutation-based for KDDCup99 dataset   .}
        \caption{\centering  KDDCup99 }
        \label{fig:KDDCup99_test_accuracy_ppsgd_lr}
    \end{subfigure}
    \hfill
    \begin{subfigure}{0.49\columnwidth}
        \centering
        \includegraphics[width=\linewidth]{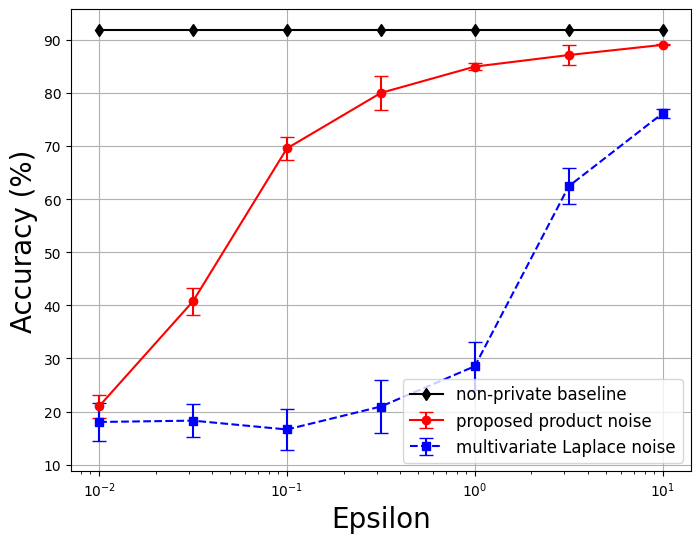}
        \Description{Accuracy results of  Private Strongly Convex Permutation-based for MNIST dataset (Low-Dim).}
        \caption{\centering  MNIST}
        \label{fig:mnist_test_accuracy_ppsgd_lr}
    \end{subfigure}
    \hfill
        \begin{subfigure}{0.49\columnwidth}
        \centering
        \includegraphics[width=\linewidth]{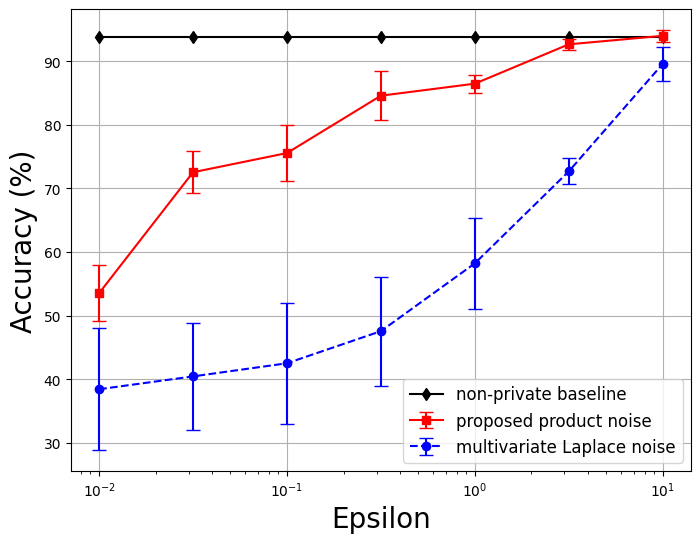}
        \Description{Accuracy results of o Private Strongly Convex Permutation-based for Synthetic-H dataset  .}
        \caption{\centering  Synthetic-H } 
        \label{fig:synthetich_test_accuracy_ppsgd_lr}
    \end{subfigure}
    \hfill
    \begin{subfigure}{0.49\columnwidth}
        \centering
        \includegraphics[width=\linewidth]{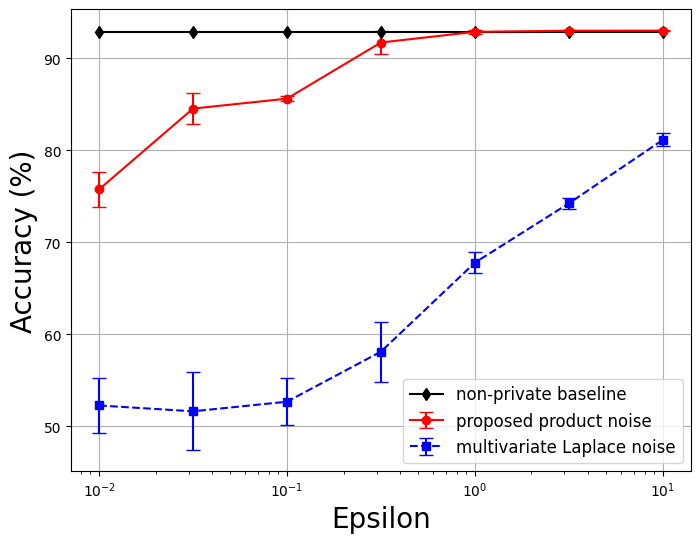}
        \Description{Accuracy results of  Private Strongly Convex Permutation-based for Real-sim dataset.}
        \caption{\centering Real-sim}
        \label{fig:realsim_test_accuracy_ppsgd_lr}
    \end{subfigure}
    \label{fig:ppsgd_lr_other_datasets}
            \hfill
    \begin{subfigure}{0.49\columnwidth}
        \centering
        \includegraphics[width=\linewidth]{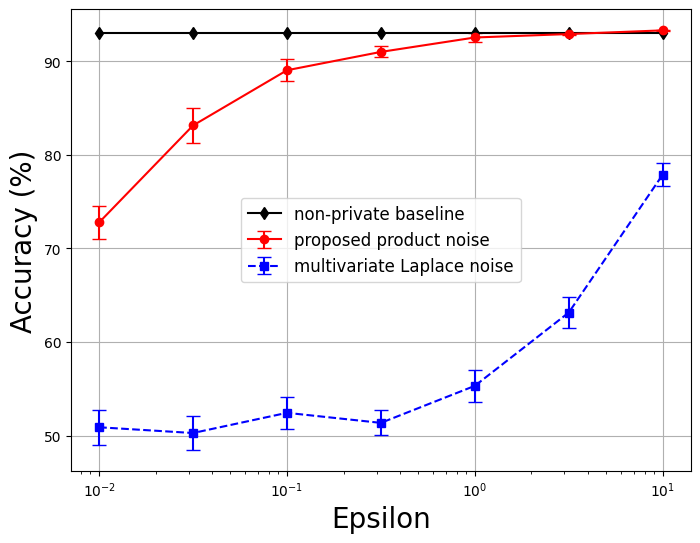}
        \Description{Accuracy results of  Private Strongly Convex Permutation-based for RCV1 dataset (High-Dim).}
        \caption{\centering RCV1 }
        \label{fig:rcv1_test_accuracy_ppsgd_lr}
    \end{subfigure}
    \caption{Test accuracy of Output Perturbation on LR. Product noise v.s.  multivariate Laplace noise. The experiment setups, e.g., iteration and solver,  follow~\cite{wu2017bolt}.}
    \label{fig:ppsgd_lr}
\end{figure}

\noindent \textbf{Supplemental Experiments for Output Perturbation Using Multivariate Laplace Noise on LR. }
We further compare our product noise with multivariate Laplace noise on the Private Strongly Convex Permutation-based SGD algorithm~\cite{wu2017bolt} to evaluate their performance on LR tasks.
As shown in Figure~\ref{fig:ppsgd_lr}, the LR model using product noise consistently achieves higher test accuracy across privacy parameters. At $\epsilon = 1$, our method achieves $84.95\%$ on MNIST and $92.53\%$ on RCV1, while the one using Laplace noise yields only $28.54\%$ and $55.31\%$, respectively. Under relaxed privacy parameters, our method can even surpass the non-private baselines. For example, on RCV1 with $\epsilon = 10$, our method reaches $93.29\%$, outperforming the baseline ($92.99\%$). Notably,  when  $\epsilon = 0.1$, our method still achieves $89.03\%$ test accuracy, while  Laplace mechanism only reaches $77.87\%$ even at $\epsilon = 10$. These results again reaffirm the Observation~\ref{observation-higher-accuracy}, demonstrating the strong utility preservation of product noise under tight privacy guarantees. 

Our method also demonstrates improved utility stability over the one using Laplace noise. The product noise offers lower standard deviation of test accuracy. For instance, on RCV1 with $\epsilon = 1$, our method yields a standard deviation of   0.00472, compared to 0.01677 using  Laplace mechanism, confirming Observation~\ref{observation-higher-stability} again.

\noindent \textbf{Supplemental Experiments of Output Perturbation using multivariate Laplace Noise on Huber SVM.}
As shown in Figure~\ref{fig:ppsgd_svm}, experiment results demonstrate that SVM models trained with product noise consistently achieve better test accuracy across all evaluated settings under the same privacy parameter. When $\epsilon = 1$, the test accuracy of product noise method on the Adult, KDDcup99, MNIST, Synthetic-H, Real-sim and RCV1 datasets reach $ 83.12\%$, $ 98.62\%$, $86.00\%$, $82.98\%$, $91.44\%$, and $93.81\%$, respectively, whereas the one using multivariate Laplace noise achieve $ 80.86\%$, $91.53\%$, $28.94\%$, $70.22\%$, $72.83\%$, and $59.84\%$. At high privacy parameters, product noise method can achieve test accuracy close to or surpasses the non-private baselines in high-dimensional datasets. For example, in the Synthetic-H, Real-sim and RCV1 datasets with $\epsilon = 10$, the test accuracy of product noise method reaches $93.98\%$, $93.02\%$ and $93.90\%$, respectively, surpassing the non-private baseline ($93.80\%$, $92.85\%$ and $93.59\%$).  
Moreover, the method using multivariate Laplace noise suffers from a noticeable performance drop under stricter privacy parameters, whereas our method maintains higher utility. For instance, under a lower privacy parameter ($\epsilon = 10^{-1}$), the test accuracy of product noise method on the Real-sim and RCV1 datasets reach $89.72\%$ and  $ 91.05\%$, respectively. In contrast, even under a higher privacy parameter ($\epsilon = 10$),  the one using multivariate Laplace noise only reaches $84.97\%$ and $82.72\%$.
\begin{figure}[htp]
    \centering
    \begin{subfigure}{0.49\columnwidth}  
        \centering
        \includegraphics[width=\linewidth]{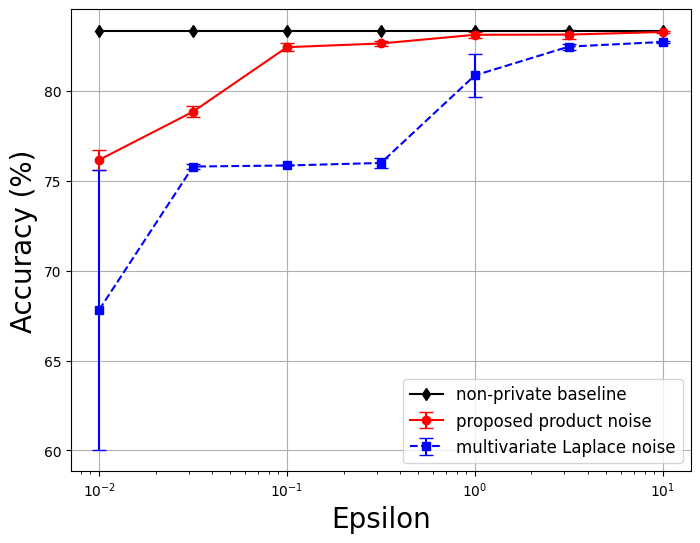}
        \Description{Accuracy results of output perturbation for Adult dataset .}
        \caption{\centering  Adult  }
        \label{fig:adult_test_accuracy_ppsgd_svm}
    \end{subfigure}
    \hfill
    \begin{subfigure}{0.49\columnwidth}
        \centering
        \includegraphics[width=\linewidth, height=3.2cm]{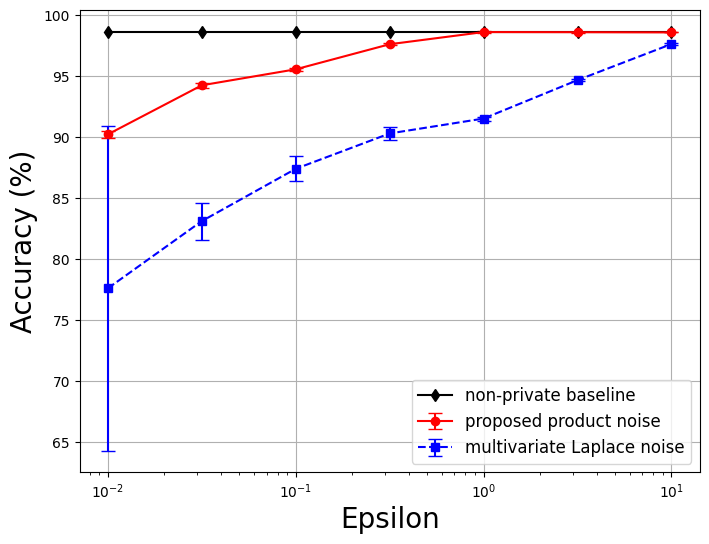}
        \Description{Accuracy results of output perturbation for KDDCup99 dataset .}
        \caption{\centering  KDDCup99 }
        \label{fig:KDDCup99_test_accuracy_ppsgd_svm}
    \end{subfigure}
     \hfill
    \begin{subfigure}{0.49\columnwidth}
        \centering
        \includegraphics[width=\linewidth]{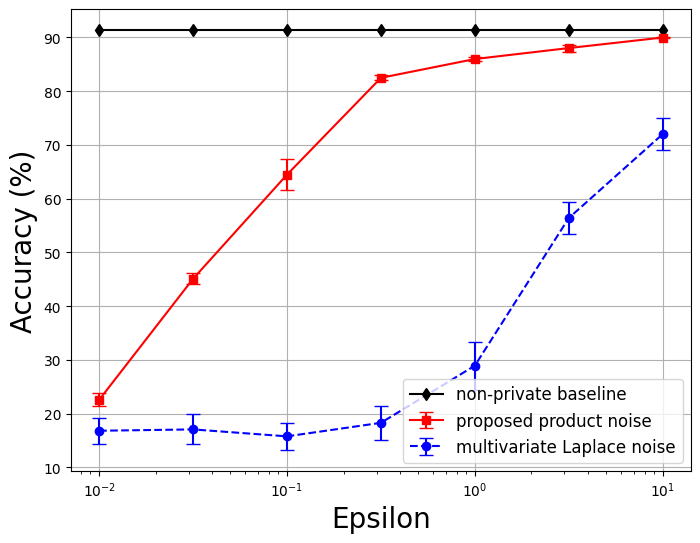}
        \Description{Accuracy results of output perturbation for MNIST dataset.}
        \caption{\centering  MNIST}
        \label{fig:mnist_test_accuracy_ppsgd_svm}
    \end{subfigure}
     \hfill
    \begin{subfigure}{0.49\columnwidth}
        \centering
        \includegraphics[width=\linewidth]{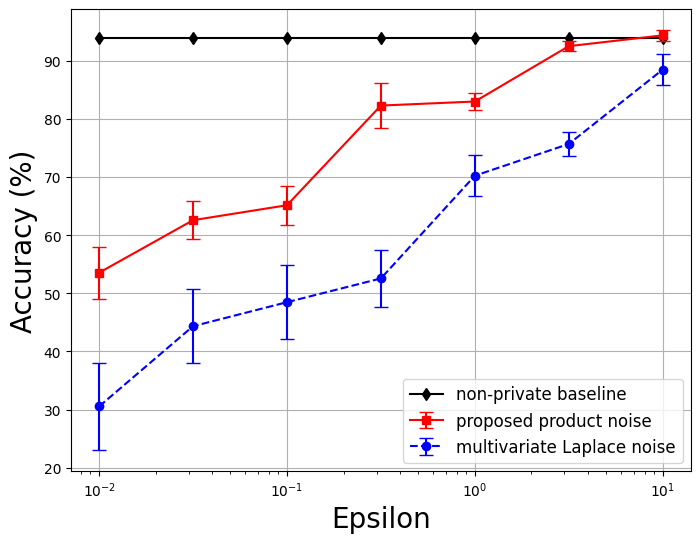}
        \Description{Accuracy results of output perturbation for Synthetic-H dataset  .}
        \caption{\centering  Synthetic-H  }
        \label{fig:synthetich_test_accuracy_ppsgd_svm}
    \end{subfigure}
     \hfill
    \begin{subfigure}{0.49\columnwidth}
        \centering
        \includegraphics[width=\linewidth]{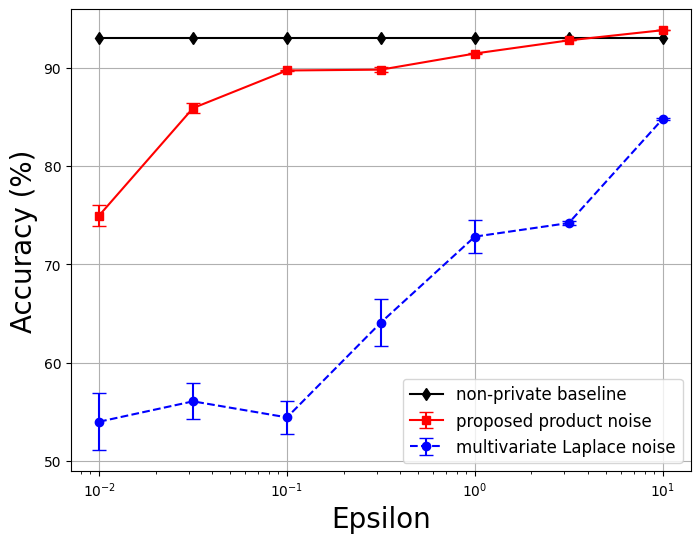}
        \Description{Accuracy results of output perturbation for Real-sim dataset  .}
        \caption{\centering Real-sim  }
        \label{fig:realsim_test_accuracy_ppsgd_svm}
    \end{subfigure}
    \hfill
    \begin{subfigure}{0.49\columnwidth}
        \centering
        \includegraphics[width=\linewidth]{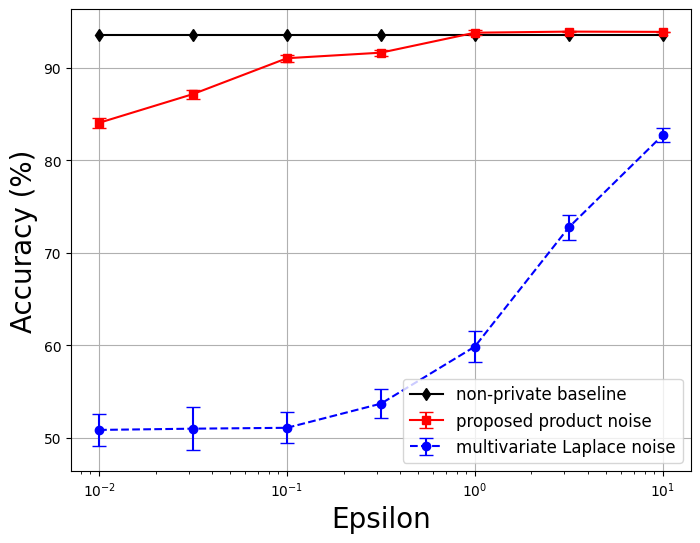}
        \Description{Accuracy results of output perturbation for RCV1 dataset.}
        \caption{\centering RCV1   }
        \label{fig:rcv1_test_accuracy_ppsgd_svm}
    \end{subfigure}
    \caption{Test accuracy  of  Output Perturbation on Huber SVM. Product Noise v.s.  multivariate  Laplace noise. The experiment setups, e.g., iteration and solver,  follow ~\cite{wu2017bolt}.} 
    \label{fig:ppsgd_svm}
\end{figure}

Regarding utility stability, the product noise-based output perturbation exhibits lower standard deviation on all datasets than the one using multivariate Laplace noise. This indicates the training process is more stable under product noise output perturbation. For instance, in the RCV1 dataset with $\epsilon = 10^{-1}$, the standard deviation of the test accuracy achieved by the product noise-based output perturbation is only 0.00360, whereas that of the multivariate Laplace mechanism is 0.01666.

\subsection{Supplemental Experiments for Objective Perturbation on Huber SVM}\label{app:svm-exp-objective}
As illustrated in Figure~\ref{fig:objective_svm}, product noise-based AMP consistently delivers stronger performance compared to the Gaussian noise-based H-F AMP across all evaluated datasets. At $\epsilon = 1$, it achieves test accuracies of $76.54\%$, $89.82\%$, $90.53\%$, and $91.95\%$ on MNIST, Synthetic-H, Real-sim, and RCV1, respectively, significantly surpasses classic  Gaussian noise-based H-F AMP and analytic Gaussian noise-based H-F AMP, particularly on RCV1 ($84.83\%$, $88.93\%$).
Product noise-based AMP retains this advantage at higher privacy parameters and can even surpass non-private baselines. For instance, on the Synthetic-H dataset with $\epsilon = 10$, our method achieves $95.33\%$ test accuracy, outperforming the baseline's $94.85\%$. 
Product noise demonstrates strong utility under tight privacy parameters, especially for high-dimensional tasks. On the RCV1 dataset, even at $\epsilon = 1$, our method reaches $91.95\%$. In contrast, classic Gaussian noise-based H-F AMP achieves only $90.59\%$ and analytic Gaussian noise-based H-F AMP achieves only $91.38\%$ under a much more relaxed privacy Parameter ($\epsilon = 10$). 
\begin{figure}[htp]
    \centering
    \begin{subfigure}{0.49\columnwidth}  
        \centering
        \includegraphics[width=\linewidth]{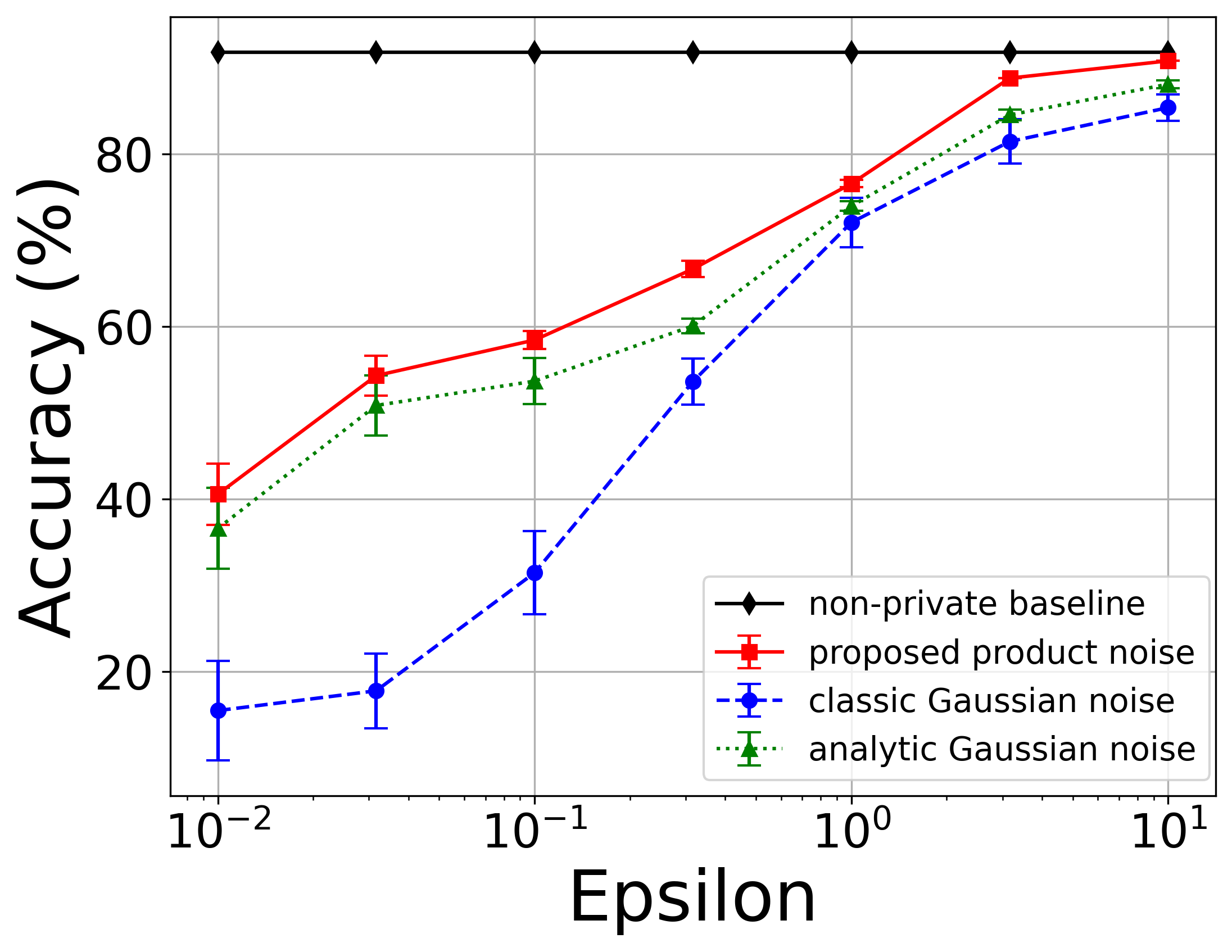}
        \Description{Graph showing accuracy versus epsilon for MNIST dataset    under objective perturbation for Huber SVM.}
        \caption{\centering  MNIST    }
        \label{fig:mnist_test_accuracy_svm}
    \end{subfigure}
    \hfill
    \begin{subfigure}{0.49\columnwidth}
        \centering
        \includegraphics[width=\linewidth]{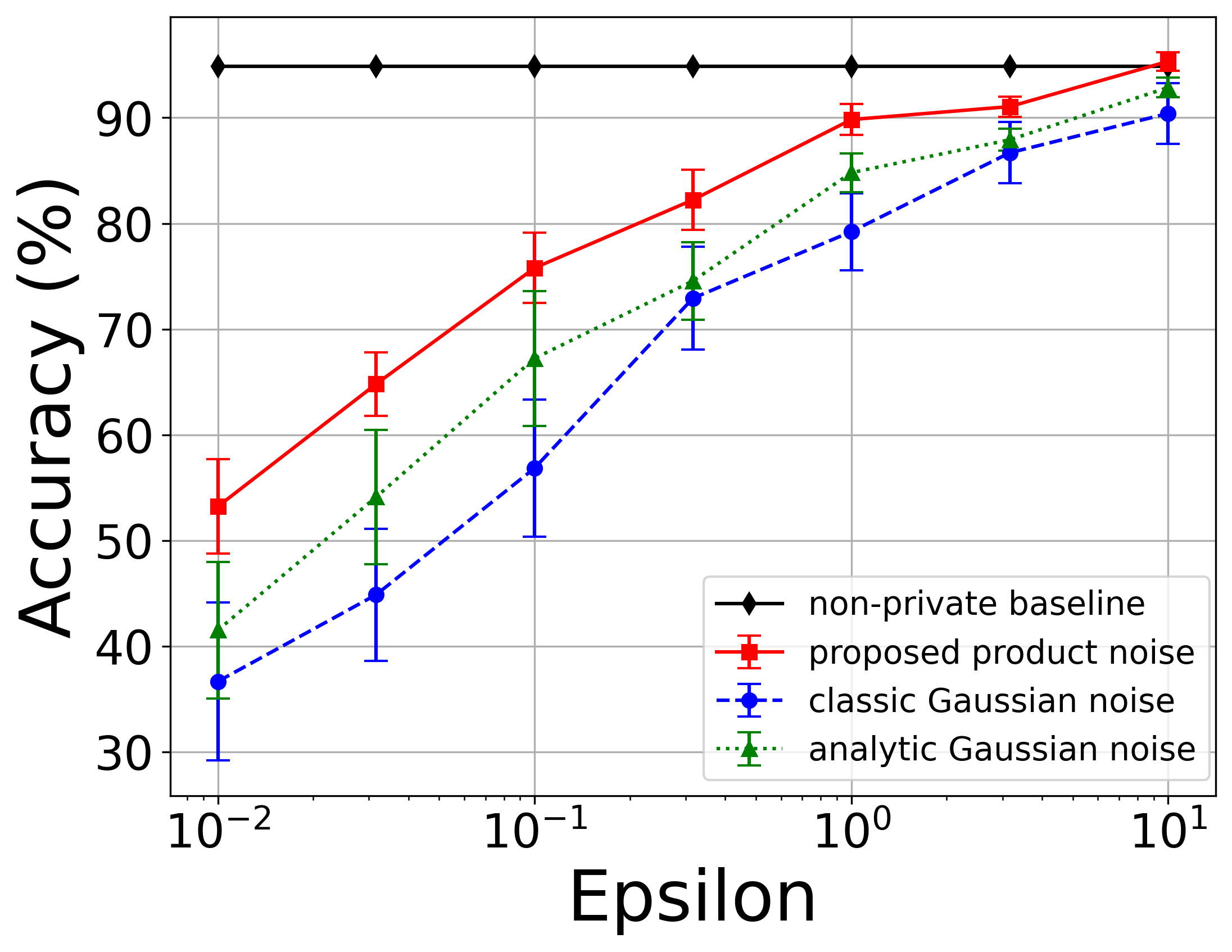}
        \Description{Graph showing accuracy versus epsilon for Synthetic-H dataset   under objective perturbation for Huber SVM.}
        \caption{\centering   Synthetic-H   }
        \label{fig:syntheticH_test_accuracy_svm}
    \end{subfigure}
    \hfill
    \begin{subfigure}{0.49\columnwidth}
        \centering
        \includegraphics[width=\linewidth]{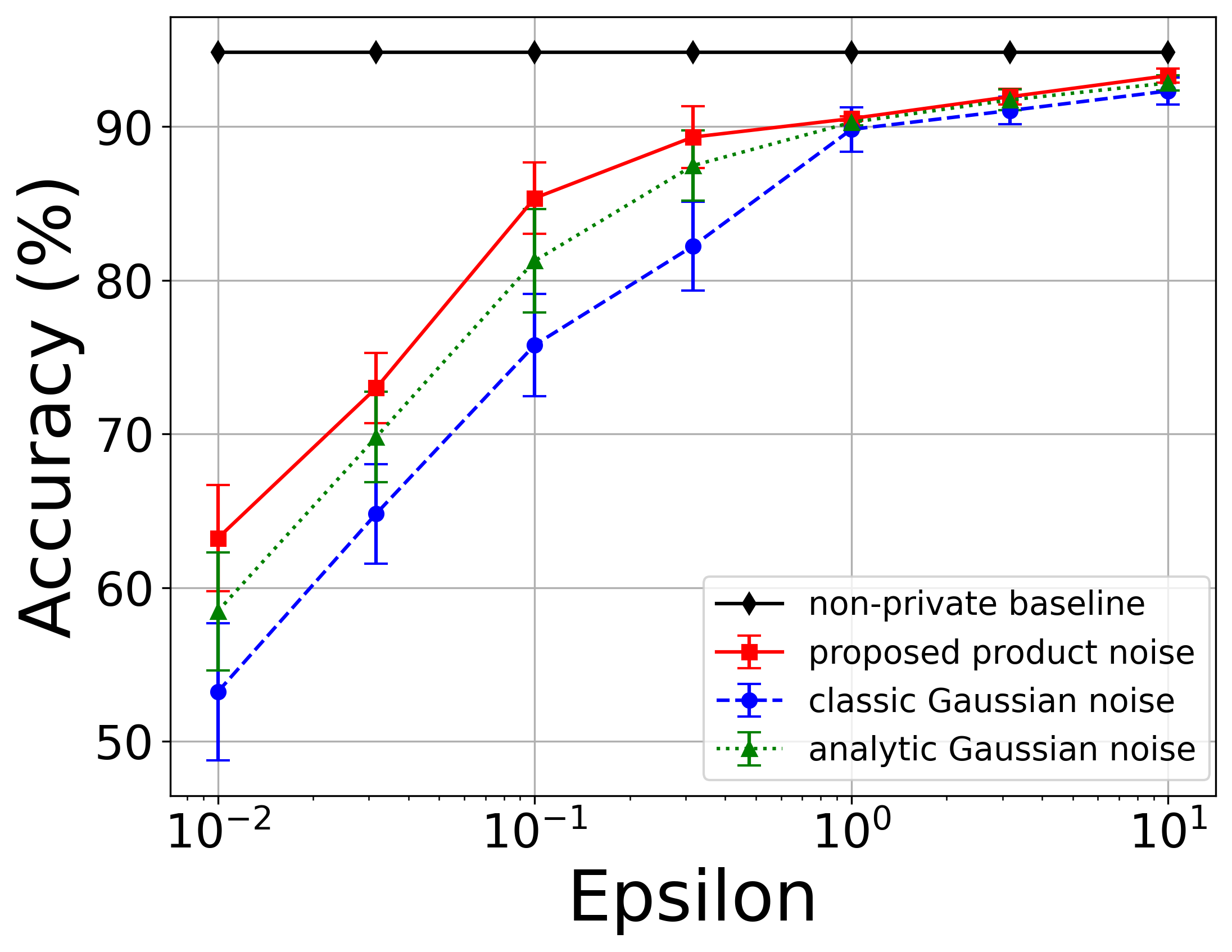}
        \Description{Graph showing accuracy versus epsilon for Real-sim dataset   under objective perturbation for Huber SVM.}
        \caption{\centering Real-sim}
        \label{fig:realsim_test_accuracy_svm}
    \end{subfigure}
    \hfill
    \begin{subfigure}{0.49\columnwidth}
        \centering
        \includegraphics[width=\linewidth]{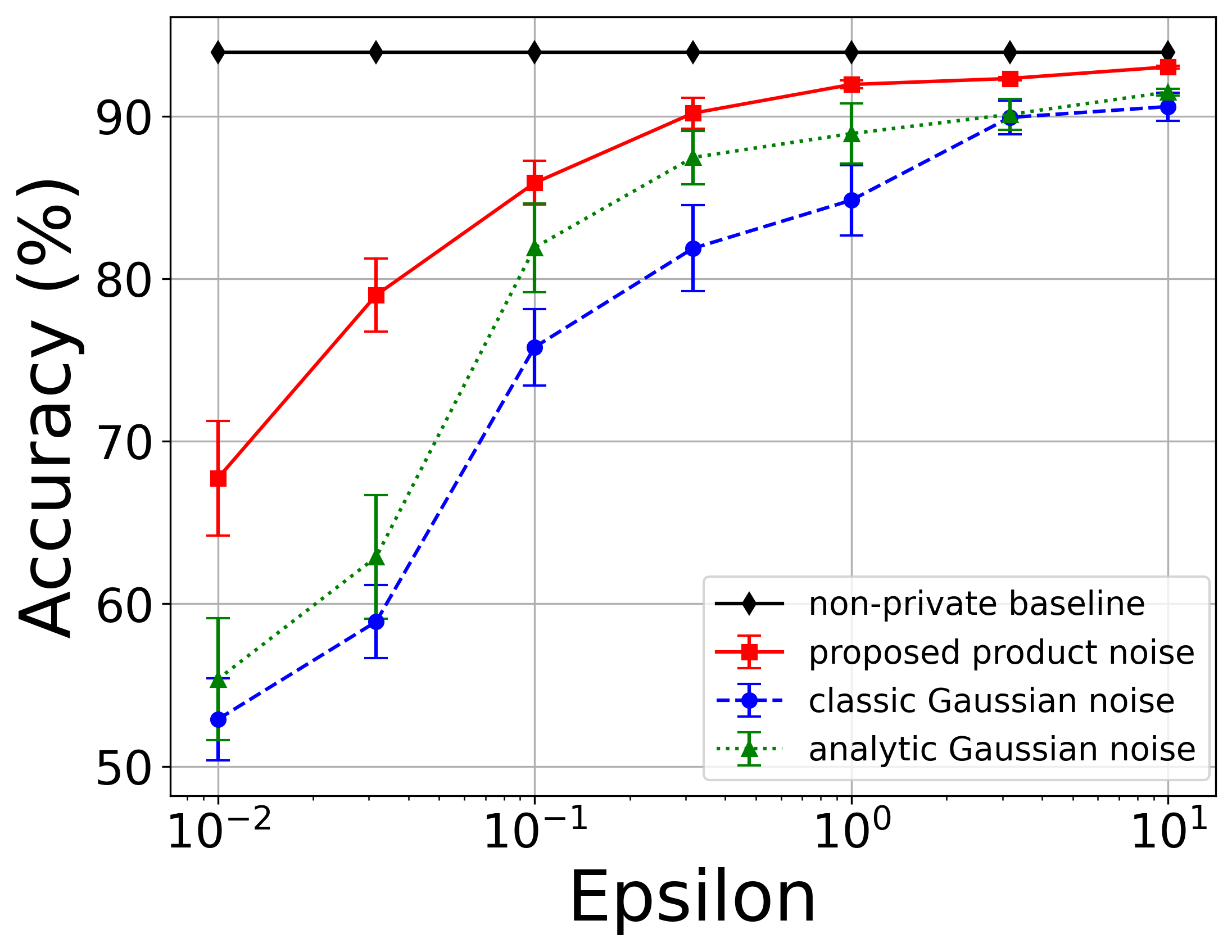}
        \Description{Graph showing accuracy versus epsilon for RCV1 dataset   under objective perturbation for Huber SVM.}
        \caption{\centering  RCV1   }
        \label{fig:rcv1_test_accuracy_svm}
    \end{subfigure}
  \caption{Test accuracy  of Objective Perturbation on Huber SVM.}
    \label{fig:objective_svm}
\end{figure}

Regarding utility stability, the product noise-based AMP consistently shows lower test accuracy standard deviation than the Gaussian noise-based H-F AMP, indicating a more stable optimization process. For example, on the RCV1 dataset with $\epsilon = 1$, the standard deviation under our method is only $0.00245$, significantly lower than $0.02153$ achieved by the classic Gaussian noise-based H-F AMP and $0.01847$ achieved by the analytic Gaussian noise-based H-F AMP.
\begin{table}[htp]
  \centering
  \caption{$\ell_2$ error of Objective Perturbation on Huber SVM.}
  \label{tab:l2_error_obj_SVM}
  \footnotesize
  \setlength{\tabcolsep}{10pt}
  \renewcommand{\arraystretch}{0.8}
  \resizebox{\columnwidth}{!}{%
  \begin{tabular}{c|c|c|c|c|c}
    \toprule
    \multirow{2}{*}{\textbf{Dataset}} &
    \multirow{2}{*}{\textbf{Mechanism}} &
    \multicolumn{4}{c}{\textbf{$\boldsymbol{\epsilon}$}} \\ 
    \cline{3-6}
     &  & \textbf{$10^{-2}$} & \textbf{$10^{-1}$} & \textbf{$10^{0}$} & \textbf{$10^{1}$} \\
    \midrule
    \multirow[c]{3}{*}{MNIST} 
      & classic  & 102.4 & 42.6 & 34.2 & 49.1 \\
      & analytic & 23.4 & 20.6 & 24.6 & 26.1 \\
      & ours     & \textbf{0.2} & \textbf{4.9} & \textbf{2.3} & \textbf{4.6} \\
    \midrule
    \multirow[c]{3}{*}{Synthetic-H} 
      & classic  & 60.2 & 212.0 & 211.4 & 70.8 \\
      & analytic & 37.6 & 35.7 & 41.8 & 40.5 \\
      & ours     & \textbf{3.0} & \textbf{1.0} & \textbf{1.2} & \textbf{2.3} \\
    \midrule
    \multirow[c]{3}{*}{Real-sim} 
      & classic  & 57.8 & 218.8 & 218.4 & 76.3 \\
      & analytic & 36.1 & 39.6 & 42.9 & 48.1 \\
      & ours     & \textbf{5.5} & \textbf{0.0} & \textbf{0.0} & \textbf{3.7} \\
    \midrule
    \multirow[c]{3}{*}{RCV1} 
      & classic  & 86.2 & 141.5 & 457.6 & 139.9 \\
      & analytic & 53.2 & 57.9 & 62.9 & 71.4 \\
      & ours     & \textbf{4.0} & \textbf{2.7} & \textbf{3.5} & \textbf{2.5} \\
    \bottomrule
  \end{tabular}}
\end{table}

\begin{table}[htp]
  \centering
  \caption{FPR of Objective Perturbation on Huber SVM.}
  \label{tab:fpr_obj_svm}
  \footnotesize
  \setlength{\tabcolsep}{4pt}
  \renewcommand{\arraystretch}{1.0}
  \resizebox{\columnwidth}{!}{%
  \begin{tabular}{c|c|c|c|c|c}
    \toprule
    \multirow{2}{*}{\textbf{Dataset}} &
    \multirow{2}{*}{\textbf{Mechanism}} &
    \multicolumn{4}{c}{\textbf{$\boldsymbol{\epsilon}$}} \\ 
    \cline{3-6}
     &  & \textbf{$10^{-2}$} & \textbf{$10^{-1}$} & \textbf{$10^{0}$} & \textbf{$10^{1}$} \\
    \midrule
    \multirow[c]{4}{*}{MNIST} 
      & classic  & 0.102$\pm$0.013 & 0.097$\pm$0.016 & 0.083$\pm$0.008 & 0.006$\pm$0.001 \\
      & analytic & 0.102$\pm$0.022 & 0.097$\pm$0.011 & 0.098$\pm$0.017 & 0.009$\pm$0.002 \\
      & ours     & \textbf{0.099$\pm$0.002} & \textbf{0.081$\pm$0.001} & \textbf{0.075$\pm$0.002} & \textbf{0.003$\pm$0.001} \\
      & baseline & 0.001 & 0.001 & 0.001 & 0.001 \\
    \midrule
    \multirow[c]{4}{*}{Synthetic-H} 
      & classic  & 0.329$\pm$0.046 & 0.158$\pm$0.032 & 0.095$\pm$0.028 & 0.033$\pm$0.018 \\
      & analytic & 0.394$\pm$0.029 & 0.241$\pm$0.013 & 0.118$\pm$0.009 & 0.077$\pm$0.005 \\
      & ours     & \textbf{0.119$\pm$0.011} & \textbf{0.059$\pm$0.008} & \textbf{0.030$\pm$0.004} & \textbf{0.008$\pm$0.002} \\
      & baseline & 0.000 & 0.000 & 0.000 & 0.000 \\
    \midrule
    \multirow[c]{4}{*}{Real-sim} 
      & classic  & 0.463$\pm$0.075 & 0.199$\pm$0.022 & 0.150$\pm$0.024 & 0.069$\pm$0.010 \\
      & analytic & 0.440$\pm$0.068 & 0.292$\pm$0.036 & 0.205$\pm$0.025 & 0.105$\pm$0.002 \\
      & ours     & \textbf{0.444$\pm$0.074} & \textbf{0.091$\pm$0.030} & \textbf{0.020$\pm$0.009} & \textbf{0.007$\pm$0.000} \\
      & baseline & 0.002 & 0.002 & 0.002 & 0.002 \\
    \midrule
    \multirow[c]{4}{*}{RCV1} 
      & classic  & 0.456$\pm$0.088 & 0.430$\pm$0.044 & 0.298$\pm$0.036 & 0.087$\pm$0.013 \\
      & analytic & 0.471$\pm$0.065 & 0.407$\pm$0.054 & 0.413$\pm$0.065 & 0.102$\pm$0.008 \\
      & ours     & \textbf{0.426$\pm$0.036} & \textbf{0.106$\pm$0.005} & \textbf{0.007$\pm$0.001} & \textbf{0.003$\pm$0.000} \\
      & baseline & 0.002 & 0.002 & 0.002 & 0.002 \\
    \bottomrule
  \end{tabular}}
\end{table}

In addition, we also report the $\ell_2$ error and FPR results for Objective Perturbation on Huber SVM in Table~\ref{tab:l2_error_obj_SVM} and Table \ref{tab:fpr_obj_svm}. The results follow the same trend as those in Section~\ref{sec:case-study-objective-perturbation}, where the product noise consistently achieves smaller $\ell_2$ errors and lower FPRs across all privacy parameters. These findings further confirm the generality of our approach, demonstrating that product noise maintains both high utility and robustness under differential privacy guarantees.

\subsection{Experiment Setup for non Convex ERM} \label{app:non-convex-setup}
We evaluate our product noise-based DPSGD on five datasets: Adult~\cite{adult_2}, IMDb~\cite{maas2011learning}, MovieLens~\cite{philips1995moment}, MNIST~\cite{lecun1998gradient}, and CIFAR-10~\cite{cifar10}. For all datasets, we adopt neural network architectures consistent with those used in~\cite{bu2020deep}. The parameter settings, including the initial privacy parameter $\epsilon$ for our method, are carefully chosen to ensure that different methods achieve comparable utility. Table~\ref{tab:dpsgd_parameters} summarizes the detailed experiment configurations for various datasets, including batch size, learning rate, gradient clipping threshold, noise multiplier (only used by classic DPSGD), and initial $\epsilon$ (only used by our product noise-based DPSGD).
\begin{table}[htp]
    \centering
    \caption{Parameters for DPSGD across various datasets.}
    \label{tab:dpsgd_parameters}
    \setlength{\tabcolsep}{4pt}
  \renewcommand{\arraystretch}{1.0}
    \small
    \resizebox{\columnwidth}{!}{ 
    \begin{tabular}{|c|c|c|c|c|c|}
        \hline
        \multirow{2}{*}{\textbf{Parameter}} & \multicolumn{5}{c|}{\textbf{Datasets}} \\ 
        \cline{2-6} 
        & \textbf{Adult} & \textbf{IMDb} & \textbf{MovieLens} & \textbf{MNIST} & \textbf{CIFAR-10} \\ 
        \hline
        Samples & 29,305 & 25,000 & 800,167 & 60,000 & 60,000 \\ 
        \hline
        \makecell{Subsampling \\ Probability ($p$)} & $\frac{256}{29305}$ & $\frac{512}{25000}$ & $\frac{10000}{800167}$ & $\frac{256}{60000}$ & $\frac{512}{60000}$ \\ 
        \hline
        Epochs & 30 & 60 & 20 & 30 & 100 \\ 
        \hline
        \makecell{Learning \\ Rate ($\eta_t$)} & 0.15 & 0.02 & 0.01 & 0.15 & 0.25 \\ 
        \hline
        \makecell{Noise \\ Multiplier ($\sigma$)} & 0.55 & 0.56 & 0.60 & 1.30 & 0.50 \\ 
        \hline
        Clip Norm ($C$) & 1.0 & 1.0 & 5.0 & 1.0 & 1.5 \\ 
        \hline
        Initial $\epsilon$ & 0.8 & 1.0 & 0.8 & 0.3 & 1.8 \\ 
        \hline
        \makecell{Tuning \\ Parameter ($k$)} & 200,000 & 40,000 & 20,000 & 40,000 & 300,000 \\ 
        \hline
    \end{tabular}
    }
\end{table}

\subsection{Gradient Perturbation on Other Datasets}\label{app:dpsgd-more-experiments}
\noindent \textbf{Adult.}  The Adult dataset consists of $32,561$ examples, with $10\%$ (3,256 examples) randomly selected as the test set, while the remaining $29,305$ examples are used for training.
\begin{figure*}[htp]
    \centering
    \begin{subfigure}{0.33\textwidth}  
        \centering
        \includegraphics[width=\linewidth]{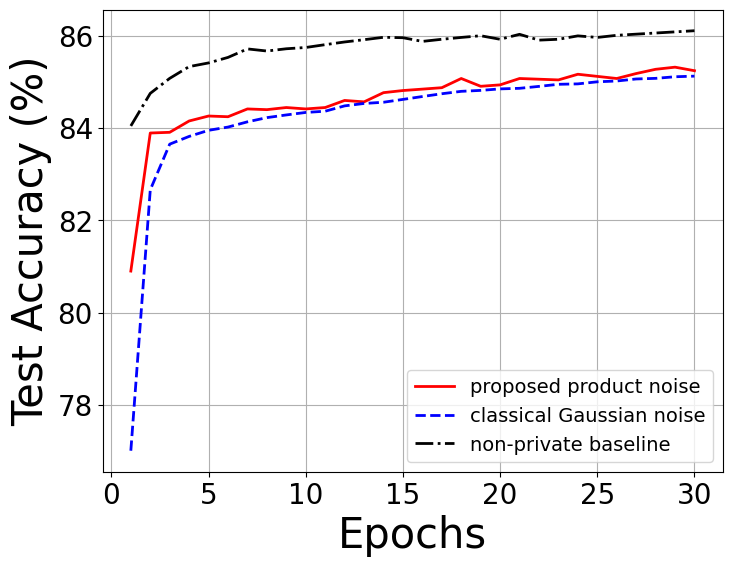}
        \Description{Graph showing test accuracy results on the Adult dataset using DPSGD.}
        \caption{Test Accuracy vs Epochs}
        \label{fig:adult_test_accuracy}
    \end{subfigure}
    \hfill
    \begin{subfigure}{0.33\textwidth}
        \centering
        \includegraphics[width=\linewidth]{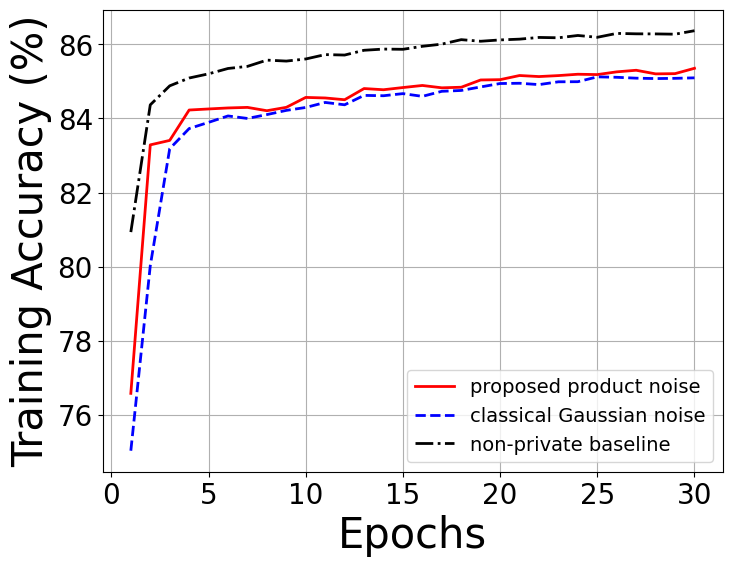}
        \caption{Training Accuracy vs Epochs}
        \Description{Training accuracy curve of the Adult dataset with DPSGD.}
        \label{fig:adult_Training_accuracy}
    \end{subfigure}
        \hfill
      \begin{subfigure}{0.33\textwidth}
        \centering
        \includegraphics[width=\linewidth]{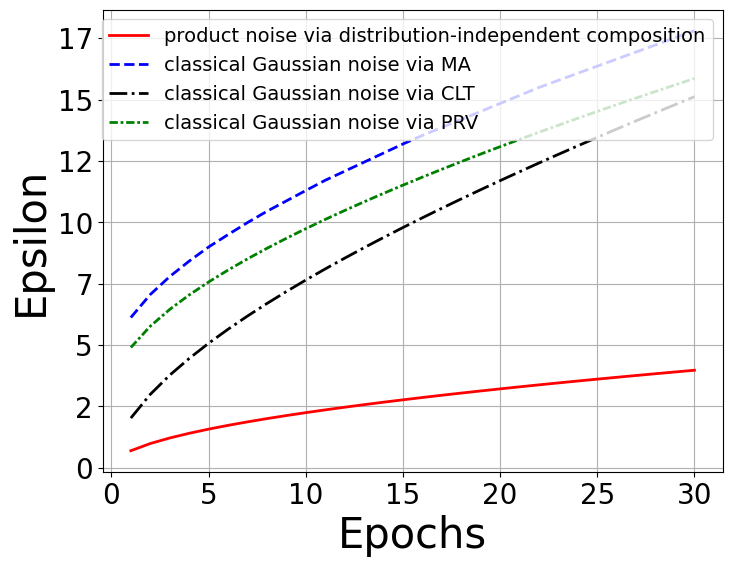}
        \Description{Graph comparing epsilon values on the Adult dataset under DPSGD.}
        \caption{Epsilon vs Epochs}
        \label{fig:adult_epsilon}
    \end{subfigure}
    \caption{DPSGD results on the Adult dataset: Test accuracy,  training accuracy, and privacy parameter ($\epsilon$) over epochs.}
    \label{fig:adult_dpsgd}
\end{figure*}

\begin{figure*}[htp]
    \centering
    \begin{subfigure}{0.33\textwidth}  
        \centering
        \includegraphics[width=\linewidth]{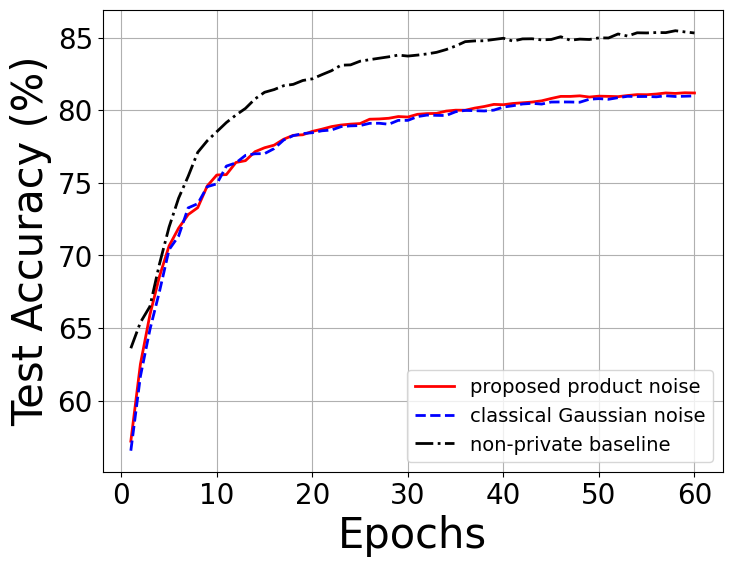}
        \Description{Graph showing test accuracy results on the IMDb dataset using DPSGD.}
        \caption{Test Accuracy vs Epochs}
        \label{fig:IMDb_test_accuracy}
    \end{subfigure}
    \hfill
    \begin{subfigure}{0.33\textwidth}
        \centering
        \includegraphics[width=\linewidth]{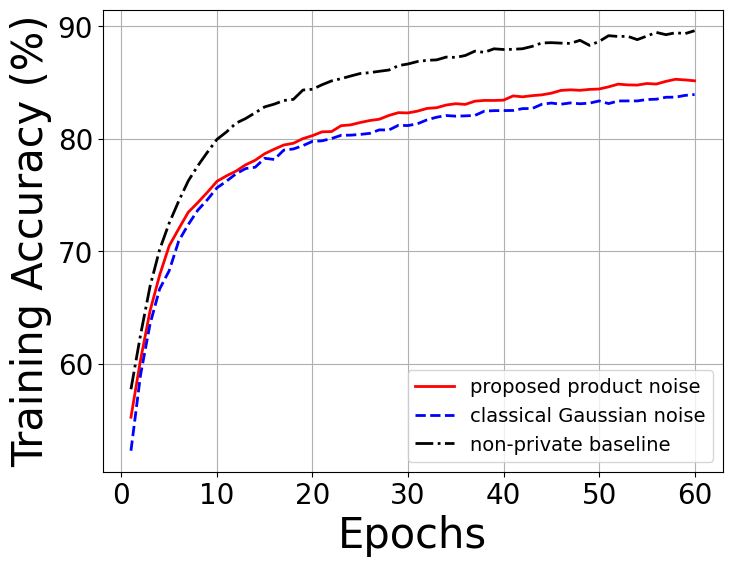}
        \caption{Training accuracy vs Epochs}
        \Description{Training accuracy curve of the IMDb dataset with DPSGD.}
        \label{fig:IMDb_Training_accuracy}
    \end{subfigure}
        \hfill
      \begin{subfigure}{0.33\textwidth}
        \centering
        \includegraphics[width=\linewidth]{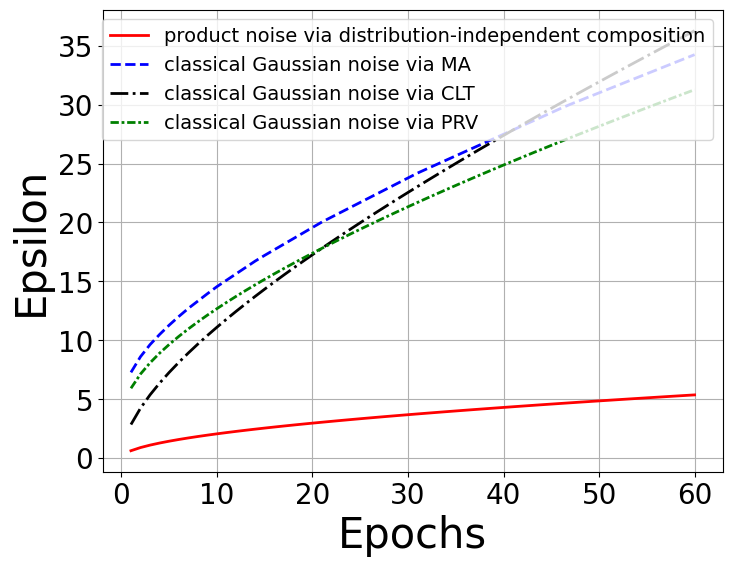}
        \Description{Graph comparing epsilon values on the IMDb dataset under DPSGD.}
        \caption{Epsilon vs Epochs}
        \label{fig:IMDb_epsilon}
    \end{subfigure}
    \caption{ DPSGD results on the IMDb dataset: Test accuracy,  training accuracy, and privacy parameter ($\epsilon$) over epochs.}
    \label{fig:IMDb_dpsgd}
\end{figure*}
As shown in Figure~\ref{fig:adult_dpsgd} (a) and (b), our product noise-based DPSGD and the Gaussian noise-based DPSGD achieve comparable utility. Over 30 epochs, the test accuracy of the product noise-based DPSGD improves from $80.90\%$ to $85.24\%$, while the Gaussian noise-based DPSGD improves from $77.00\%$ to $85.13\%$. Similarly, the training accuracy increases from $76.58\%$ to $85.35\%$ for our method, and from $75.03\%$ to $85.09\%$ for the Gaussian method. These results indicate that both methods achieve similar performance in model utility under the current experiment settings.
Under this comparable utility, our noise-based DPSGD product offers significantly stronger privacy preservation. As shown in Figure~\ref{fig:adult_dpsgd} (c), for 30 epochs, our method achieves $(3.97, 8.10 \times 10^{-6})$-DP (evaluated using privacy amplification followed by distribution-independent composition). This is notably smaller than the privacy guarantees obtained by the Gaussian noise-based DPSGD under MA ($(17.80, 10^{-5})$-DP), PRV ($(15,86, 10^{-5})$-DP), and CLT ($(15.12, 10^{-5})$-DP) composition approaches. Furthermore, the smaller $\delta$ value indicates a lower failure probability for the $(\epsilon,\delta)$-DP guarantee. These results support our argument in Observation~\ref{observation-stronger-privacy} that our method can achieve tighter privacy guarantees under comparable utility.

\begin{figure*}[htp]
    \centering
    \begin{subfigure}{0.33\textwidth}  
        \centering
        \includegraphics[width=\linewidth]{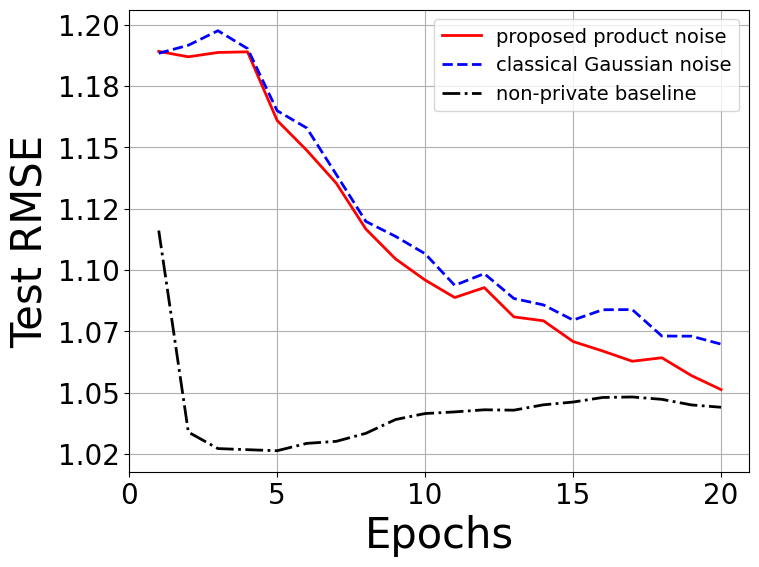}
        \Description{Graph showing test accuracy results on the MovieLens dataset using DPSGD.}
        \caption{Test Accuracy vs Epochs}
        \label{fig:MovieLens_test_accuracy}
    \end{subfigure}
    \hfill
    \begin{subfigure}{0.33\textwidth}
        \centering
        \includegraphics[width=\linewidth]{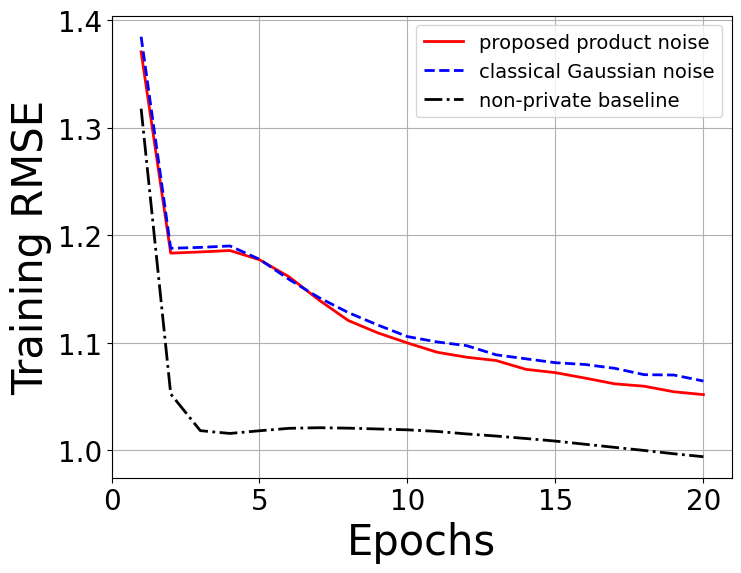}
        \caption{Training Accuracy vs Epochs}
        \Description{Training accuracy curve of the MovieLens dataset with DPSGD.}
        \label{fig:MovieLens_Training_accuracy}
    \end{subfigure}
        \hfill
      \begin{subfigure}{0.33\textwidth}
        \centering
        \includegraphics[width=\linewidth]{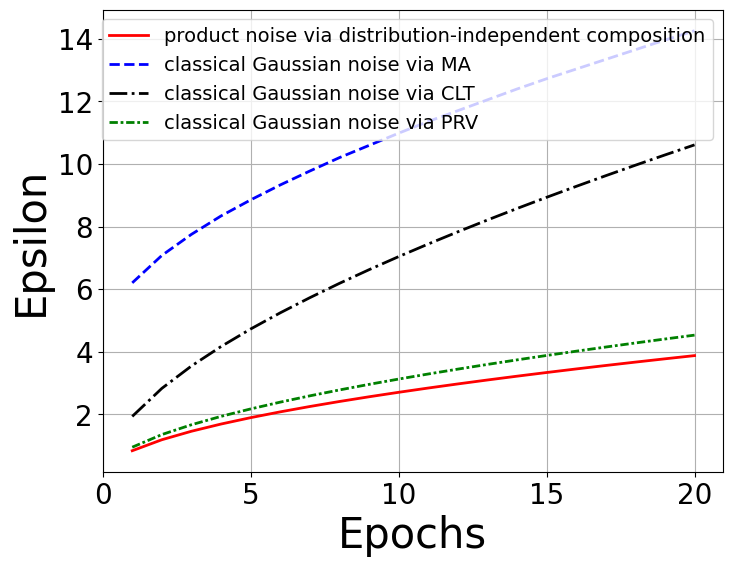}
        \Description{Graph comparing epsilon values on the MovieLens dataset under DPSGD.}
        \caption{Epsilon vs Epochs}
        \label{fig:MovieLens_epsilon}
    \end{subfigure}
    \caption{ DPSGD results on the MovieLens dataset: Test RMSE,  training  RMSE, and privacy parameter ($\epsilon$) over epochs.}
    \label{fig:MovieLens_dpsgd}
\end{figure*}

\noindent \textbf{IMDb.}  The IMDb movie review dataset contains $25,000$ training examples and $25,000$ test examples. We set each example to 256 words, truncating the length or padding with zeros if necessary.

Figure~\ref{fig:IMDb_dpsgd} (a) and (b) demonstrate that our product noise-based DPSGD performs comparably to the Gaussian noise-based DPSGD in terms of utility. For 60 epochs, the test accuracy for the product noise method increases from $57.21\%$ to $81.18\%$, closely matching the Gaussian mechanism, which rises from $56.55\%$ to $80.98\%$. The training accuracy also exhibits similar progress, with our method improving from $55.26\%$ to $85.15\%$ and the Gaussian mechanism from $53.29\%$ to $83.95\%$. These results suggest that under the current experiment configuration, both methods achieve nearly equivalent performance in model utility.
Under comparable utility, our method exhibits a clear advantage in privacy preservation. As illustrated in Figure~\ref{fig:IMDb_dpsgd} (c), at the 60th epoch, our method achieves a privacy guarantee of $(12.99, 6.50 \times 10^{-6})$-DP, which is significantly tighter than those obtained by the Gaussian mechanism using the MA ($(34.24, 10^{-5})$-DP), PRV ($(31.26, 10^{-5})$-DP), and CLT ($(36.27, 10^{-5})$-DP) composition approaches, which indicates that our method can provide stronger privacy protection under the comparable utility.

\noindent \textbf{MovieLens.} The MovieLens 1M dataset contains $1,000,209$ movie ratings, and there are totally $6,040$ users rated $3,706$ different movies. We randomly selected $20\%$ of the examples as the test set and the rest as the training set, and chosen Root Mean Square Error (RMSE) as the performance measure.

As shown in Figure~\ref{fig:MovieLens_dpsgd}(a) and (b), the product noise-based DPSGD and the Gaussian noise-based mechanism exhibit comparable utility performance during training. Over 20 epochs, the test RMSE of the product noise method decreases from 1.91 to 1.05, while that of the Gaussian mechanism decreases from 1.21 to 1.06. The training RMSE shows a similar trend, decreasing from 1.37 to 1.05 and from 1.38 to 1.06, respectively. These results indicate that, under the current experiment configuration, both methods achieve similar levels of model utility.
Under this comparable utility, the product noise mechanism also demonstrates a clear advantage in terms of privacy protection. As shown in Figure~\ref{fig:MovieLens_dpsgd} (c), at the 20th  epoch, our method achieves a privacy guarantee of $(3.88, 5.90 \times 10^{-6})$-DP, which is significantly tighter than those obtained by the classical DPSGD evaluated via MA ($(14.26, 10^{-5})$-DP), PRV ($(4.53, 10^{-5})$-DP), and CLT ($(10.61, 10^{-5})$-DP) composition approaches. This result further validates our conclusion that the product noise mechanism can achieve stronger privacy guarantees under comparable utility.

\bibliographystyle{ACM-Reference-Format}
\bibliography{ref}

\end{document}